\documentclass{article}

\RequirePackage{amsthm,amsmath,amsfonts,amssymb}
\RequirePackage[authoryear]{natbib}%
\RequirePackage[colorlinks,citecolor=blue,urlcolor=blue]{hyperref}
\RequirePackage{graphicx}
\usepackage[capitalize]{cleveref}
\usepackage{xr}
\usepackage{bm}
\usepackage{subfig}
\usepackage{tabularx}
\usepackage{autonum}

\usepackage[margin=1in]{geometry}
\usepackage[capitalize]{cleveref}
\Crefname{assumption}{Assumption}{Assumptions}
\Crefname{appsec}{Appendix}{Appendices}

\usepackage{tabularx}
\newcolumntype{b}{X}
\newcolumntype{s}{>{\hsize=.25\hsize}X}
\newcolumntype{m}{>{\hsize=.75\hsize}X}

\newtheorem{theorem}{Theorem}
\newtheorem{lemma}{Lemma}
\newtheorem{proposition}{Proposition}
\newtheorem{corollary}{Corollary}
\theoremstyle{definition}
\newtheorem{example}{Example}
\newtheorem{assumption}{Assumption}
\newtheorem{definition}{Definition}
\newtheorem{remark}{Remark}

\newcommand{\retainlabel}[1]{\label{#1}\sbox0{\ref{#1}}}

\newcommand{\Fsminone}{\mathcal{F}_{s-1}}

\newcommand{\adpt}{\mathrm{adpt}} %
\newcommand{\svs}{\mathrm{svs}} %
\newcommand{\N}{\mathcal{N}} %
\newcommand{\skg}{\mathrm{skg}}  %

\newcommand{\bias}{\mathrm{bias}}

\newcommand{\YNsvs}{\widetilde{Y}^{\svs}}
\newcommand{\Ynpsvs}{Y^{\svs}}
\newcommand{\YNadpt}{\widetilde{Y}_\lambda}
\newcommand{\Ynpadpt}{Y_\lambda}

\newcommand{\muhatN}{\widetilde{\mu}_\lambda(s)} %
\newcommand{\muhatnp}{\widehat{\mu}_{\lambda, s}}

\newcommand{\Deltaqvar}{\Delta^{\mathrm{qvar}}}
\newcommand{\muhatnpzero}{\widehat{\mu}^{0}_{\lambda, s}}

\newcommand{\rmlSPRT}{\mathrm{rmlSPRT}}
\newcommand{\nmSPRT}{\mathrm{nmSPRT}}

\newcommand{\Natints}{\mathbb{N}}

\newcommand{\eg}{\emph{e.g.}}

\usepackage{appendix}

\title{%
Near-Optimal, Near-Calibrated, Non-Parametric Sequential Tests and Confidence Sequences with Possibly Dependent Observations}

\author{Aur\'elien Bibaut$^1$,\quad  Nathan Kallus$^{1,2}$,\quad 
~Michael Lindon$^1$}
\date{
$^1$Netflix\qquad 
$^2$Cornell University
}

\allowdisplaybreaks

\begin{document}

\maketitle

\begin{abstract}
Sequential tests and their implied confidence sequences, which are valid at arbitrary stopping times, promise flexible statistical inference and on-the-fly decision making. However, strong guarantees are limited to parametric sequential tests that under-cover in practice or concentration-bound-based sequences that over-cover and have suboptimal rejection times. In this work, we consider classic delayed-start normal-mixture sequential probability ratio tests, and we provide the first asymptotic type-I-error and expected-rejection-time guarantees under general non-parametric data generating processes, where the asymptotics are indexed by the test's burn-in time. The type-I-error results primarily leverage a martingale strong invariance principle and establish that these tests (and their implied confidence sequences) have type-I error rates asymptotically equivalent to the desired (possibly varying) $\alpha$-level. The expected-rejection-time results primarily leverage an identity inspired by It\^o's lemma and imply that, in certain asymptotic regimes, the expected rejection time is asymptotically equivalent to the minimum possible among $\alpha$-level tests. We show how to apply our results to sequential inference on parameters defined by estimating equations, such as average treatment effects. Together, our results establish these (ostensibly parametric) tests as general-purpose, non-parametric, and near-optimal. We illustrate this via numerical simulations and a real-data application to A/B testing at Netflix.
\end{abstract}

\tableofcontents

\section{Introduction}

Inference based on randomized experiments forms the basis of important decisions in an incredibly diverse range of domains, from medicine \citep{Schulzc332} to development economics \citep{banerjee2015miracle} to technology business \citep{amatriain2012netflix}. Toward the aim of making \emph{better} and \emph{faster} decisions, it is particularly helpful for statistical procedures to be \emph{flexible} \citep{grunwald}. Experimental designs which require a pre-specified sample size can be quite rigid in practice. For example, unless some oracle information is known about the effect size, such designs will ultimately include some over- or under-experimentation; collecting more samples than necessary when the treatment effect is under-estimated, and not collecting enough when over-estimated. Sequential designs, on the other hand, enable \emph{faster} decision-making as they support the ability to analyze data as and when it arrives, enabling us to stop experimenting when the data strongly supports a conclusion. The ability to continuously monitor experiments in turn leads to better decisions, as it prevents bad practices that arise from attempting to use more rigid statistical procedures in applications where resources and schedules often change dynamically.
 
\subsection{The sequential testing framework}\label{subsec:seq_testing_problem}

The sequential testing problem can in general be modelled by the following framework, which encompasses most settings from the literature we refer to.  The analyst receives a stream of random data points $X_1,X_2,\ldots$, all lying in a certain set $\mathcal{X}$, adapted to a filtration $\mathfrak{F}=(\mathcal{F}_t)_{t \geq 1}$. For any $t \geq 1$, we denote $\psi_t = E[X_t \mid \mathcal{F}_{t-1}]$ the conditional mean of $X_t$ given previous observations. We consider the problem of testing the statistical hypothesis $H_0$ that $\psi_t = \psi'_0$ for every $t \geq 1$, for some $\psi'_0 \in \mathbb{R}$, against the composite alternative $\widetilde{H}_{\backslash 0} : \cup_{\psi'_1 \neq \psi'_0} \widetilde{H}(\psi'_1)$, where $\widetilde{H}(\psi'_1) : \lim_{t \to \infty} \psi_t = \psi'_1$. In words, $\widetilde{H}_{\backslash 0}$ is the hypothesis that $\psi_t$ has a limit, and that this limit is not $\psi'_0$.
A sequential test is an $\mathfrak{F}$-adapted stopping time $\tau$. If $\tau = t$ for some $t < \infty$, we say that the sequential time rejects $H_0$ at $t$, while if $\tau = \infty$, we say that the test doesn't reject $H_0$.

\subsection{Challenges of sequential testing and limitations of the state of the art}

Analyzing data in this fashion requires special inference that explicitly accounts for the sequential nature of the decision-making process. It is well known that repeated application of classical significance tests to accumulating sets of data results in procedures with drastically inflated type-I error rates \citep{armitage}. Even under the null hypothesis, the absolute value of the $t$-statistic applied to an independent and identically distributed (i.i.d.) sequence is guaranteed to fall into the rejection region at some point, regardless of the chosen $\alpha$-level \citep{strassen1964invariance}. An analyst intent on disproving any one hypothesis can keep collecting data until this occurs.

For these reasons, a long line of literature has explored sequential testing and confidence sequences. However, we claim that the existing literature still does not provide all the guarantees necessary for the widespread use of confidence sequences and sequential tests in statistical practice. We believe indeed that practitioners would want to know three things before adopting them.
\begin{itemize}
\item Is type-I error, or equivalently, coverage guaranteed to hold at the predetermined level $\alpha$?
\item Is the procedure sample efficient? In particular, is its expected sample size unimprovable among all level-$\alpha$ sequential tests?
\item Is the procedure non-parametric? Indeed, practitioners aren't (or at least shouldn't be) willing to make parametric assumptions, which we believe are unrealistic in most settings.
\end{itemize}

In our understanding, the prevailing view is that practitioners have essentially two choices at their disposal. One is to use parametric-model-based confidence sequences which tend to under-cover and over-reject. The other is to use concentration-inequalities-based sequences, which tend to over-cover and under-reject. A prominent example of parametric confidence sequences is the mixture SPRT boundaries for Brownian data sequences introduced by \cite{robbins1970boundary} and further refined in \citep{robbins1970statistical}. Notable contributions to the field of concentration-inequality-based sequences include \cite{howard2021} and \cite{howardquantile}.

These two types of sequences are not actually the only options but the alternatives may be less well-known as of now. These alternatives are asymptotic in nature and are based on weak and strong invariance principles (WIPs and SIPs). \cite{robbins1970boundary} propose a sequence of delayed-start mixture SPRT boundaries and show, via Donsker's WIP, that coverage tends to $\alpha$ as the delay before the start of monitoring diverges to infinity. More recently, \cite{bibaut2021sequential} use \cite{mcleish1974dependent}'s WIP to provide a sequence of delayed-start confidence sequences for dependent data. \cite{smith} propose a novel definition of asymptotic confidence sequences and demonstrate how SIPs allow to construct such asymptotic sequences. The advantage of these WIP- and SIP-based
sequences (or sequences of sequences) is that they hold the promise of tight type-I error control (as opposed to over or anti-conservativity) under minimal assumptions -- typically mere moment assumptions. We discuss more in detail why going for an asymptotic WIP-based or SIP-based solution is the right strategy in the next subsection.

Nevertheless, we claim that, as of the time of writing of this article, no existing work simultaneously proposes confidence sequences and an analysis showing that these satisfy the aforementioned three criteria. In particular, we discuss in \cref{sec:modern_IP-based_CSs} how the current work compares to \cite{smith} in that respect.

\subsection{The case for an asymptotic procedure}

In fixed-sample-size inference, arguably the most used procedure is to construct approximate central-limit-theorem- (CLT) based confidence intervals. The central limit theorem and extensions thereof only require a second-order moment (or a Lindeberg or Lyapunov condition), and allow for an extremely wide range of data-generating processes. The coverage error one makes by using a CLT-based confidence intervals is controlled by Berry-Esseen bounds and decreases rapidly, scaling as the inverse square root of the sample size (provided the observations have a third moment).

Comparatively, parametric-model-based confidence sequences and nonparametric concentration bounds aren't as commonly used in applied statistical practice for the same reasons we mentioned for confidence sequences: the former become anti-conservative as soon as the parametric assumptions do not hold, while the latter are over-conservative in general. 

This therefore motivates looking for a CLT-like solution to sequential inference and testing. 
Fortunately, there exist sequential extensions of CLT results: these are precisely the WIPs and the SIPs we mentioned above. These guarantee convergence (in distribution and almost surely, respectively) of partial sum processes to Gaussian processes, under weak moment assumptions.

\subsection{Aim of the current paper}\label{subsec:paper_aim}

We aim to justify the use of existing sequential tests and the corresponding confidence sequences, namely the delayed-start versions \citep{robbins_siegmund1974} of the running maximum likelihood SPRT (rmlSPRT) \citep{robbins1972class, robbins_siegmund1974}, and the normal mixture SPRT (nmSPRT). In other words, we give guarantees supporting that it is safe, efficient, and advisable to use these. We show that, under mild non-parametric assumptions, sequences of these sequential tests indexed by the burn-in time (1) have type-I error converging asymptotically to the prespecified level, and (2) have expected stopping time converging to the lower bounds known in the i.i.d. case both in the regime where $\psi_1' \to \psi_0'$, and in the regime where $\alpha \to 0$.

\subsection{Contributions and organization of this article}

We show in \cref{sec:type-I_error} that the type-I error of delayed-start nmSPRTs and rmlSPRTs converges to $\alpha$ as burn-in time goes to infinity. We show in \cref{sec:expected_rejection_time} that, in our non-parametric dependent setting, in the asymptotic regime where $\alpha \to 0$, the expected rejection time is asymptotically equivalent, up to a small constant, to the lower bound known for simple-vs-simple SPRTs under parametric i.i.d. data. We also show that, in the asymptotic regime where the effect size $\psi'_1 - \psi'_0$ converges to 0, the expected stopping time converges, up to a constant factor, to the lower bound known in the parametric i.i.d. setting. We propose in \cref{sec:application} a sample splitting and variance stabilization method to conduct estimating-equations-based inference under sequentially collected data. In \cref{sec:tuning_lambda}, we study heuristically  the optimal choice of the hyperparameter in nmSPRTs, and we propose a heuristic procedure that auto-tunes this hyperparameter. In \cref{sec:numerics} we conduct extensive simulation studies so as to demonstrate numerically the validity of our theoretical results and the empirical robustness of our heuristic auto-tuned nmSPRTs, and in \cref{sec:realdata} we study a real-data application to A/B testing for quality control of the Netflix client application. In \cref{sec:historical_progression} we review the historical and logical progression from \cite{wald}'s simple-vs-simple SPRT to the results we present in this paper. In \cref{sec:lit} we further discuss the related literature.

\section{Setup}

We work in the sequential setting introduced in \cref{subsec:seq_testing_problem}. For any $t \geq 1$, we denote $S_t = \sum_{s=1}^t X_s$. Our methods are designed to work with observations for which the conditional variance given the past converges to one. We will make rigorous what are the specific variance convergence requirements in the upcoming sections, but it might be helpful to think throughout of $\mathrm{Var}(X_t \mid \mathcal{F}_{t-1})$ as being approximately one.

For any $t \geq 1$, and $\lambda > 0$, we define the rmlSPRT and nmSPRT statistics as 
\begin{align}
Y_t^{\rmlSPRT} = \frac{1}{2} \left(\frac{S_t^2}{t} - \log t \right) \qquad \text{and} \qquad Y_{\lambda, t}^{\nmSPRT} = \frac{1}{2} \left( \frac{S_t^2}{t+\lambda} - \log \frac{t + \lambda}{\lambda} \right). \label{eq:test_statistics_definition}
\end{align}
We refer the reader to \cref{sec:historical_progression} for the motivation behind the expressions of the test statistics and the rationale of their names.

So as to avert any confusion, the reader should keep in mind that while we present the guarantees for both test statistics along one another, the analyst will opt for one of the two and monitor only the chosen one (as opposed for instance to monitoring which of the two crosses a given threshold first, or any variation thereof).
Let $t_0 \geq 1$ be the burn-in period, that is the time the analyst waits before starting to monitor whether the test statistic crosses a certain rejection threshold. We will justify that for burn-in period $t_0$ (and parameter $\lambda$ for the nmSPRT), the ``right'' rejection thresholds for the rmlSPRT and the nmSPRT are $-\log \widetilde{\alpha}_1(\alpha)$ and $-\log \widetilde{\alpha}_{2,\lambda / t_0}(\alpha)$, respectively, where $\widetilde{\alpha}_1(\alpha)$ solves $h_1(-\log \widetilde{\alpha}_1) = \alpha$ and, for any $\eta > 0$, $\widetilde{\alpha}_{2,\eta}(\alpha)$ solves $h_{2,\eta}(- \log \widetilde{\alpha}_{2,\eta}) = \alpha$,
where
\begin{align}
h_1(a) =& 2 \exp(-a) \sqrt{\frac{a}{\pi}} + 2 \left(1 - \Phi\left(\sqrt{2 a}\right)\right), \\
\text{and} \qquad  h_{2,\eta}(a) =& \exp(-a) \left\lbrace 2 \Phi\left(\sqrt{2 \eta \left( a + \frac{1}{2} \log \frac{1+\eta}{\eta} \right)}  \right) - 1 \right\rbrace \\
&+ 2 \left\lbrace 1-\Phi \left(\sqrt{2 (1+\eta) \left( a + \frac{1}{2} \log \frac{1+\eta}{\eta} \right) } \right) \right\rbrace.
\end{align}
The combination of a test statistic, a burn-in period, and a rejection level specifies a sequential test. Formally we define the the $t_0$-burn-in rmlSPRT and nmSPRT as the stopping times
\begin{align}
\tau^{\rmlSPRT}(\alpha, t_0) =& \inf \left\lbrace t \geq t_0 : Y_t^{\rmlSPRT} \geq -\log \widetilde{\alpha}_1(\alpha) \right\rbrace, \\
\text{and} \qquad \tau^{\nmSPRT}(\alpha, \lambda, t_0) =& \inf \left\lbrace t \geq t_0 : Y_{\lambda,t}^{\nmSPRT} \geq -\log \widetilde{\alpha}_{2, \lambda / t_0}(\alpha) \right\rbrace.
\end{align}
For any $t \geq 1$, we define $\mathcal{C}^{\rmlSPRT}_{\alpha, t_0}(t)$ (resp. $\mathcal{C}^{\nmSPRT}_{\alpha, \lambda, t_0}(t)$) as the set of values $\psi'_0 \in \mathbb{R}$ for which the rmlSPRT (resp. the nmSPRT) doesn't reject at $t$ the hypothesis $H_0: \psi_t = \psi'_0 \ \forall t$. It is then straightforward to observe that the sequences $(\mathcal{C}^{\rmlSPRT}_{\alpha, t_0}(t))_{t \geq 1}$ and $(\mathcal{C}^{\nmSPRT}_{\alpha, \lambda, t_0}(t))_{t \geq 1}$ have coverage equal to 1 minus the level of the corresponding tests. Inverting the expressions of the test statistics yields that, for $t\geq t_0$, $\mathcal{C}^{\rmlSPRT}_{\alpha, t_0}(t) = [\pm c^{\rmlSPRT}_{\alpha, t_0}(t)]$ and $\mathcal{C}^{\nmSPRT}_{\alpha, \lambda, t_0}(t) = [\pm c^{\nmSPRT}_{\alpha, \lambda, t_0}(t)]$ with
\begin{align}
c^{\rmlSPRT}_{\alpha, t_0}(t) =& \sqrt{t \left(-2\log \widetilde{\alpha}_1(\alpha) + \log \frac{t}{t_0}\right)} \label{eq:delayed_rmle_CS} \\
\text{and} \qquad c^{\nmSPRT}_{\alpha, \lambda, t_0}(t) =& \sqrt{(t+\lambda) \left( -2 \log \widetilde{\alpha}_{2, \lambda/t_0}(\alpha) + \log \frac{t+\lambda}{\lambda}  \right)}. \label{eq:delayed_nmSPRT_CS}
\end{align}
We now present two applications that exemplify our setting and the use of the delayed-start confidence sequences.
\begin{example}[Sample mean]\label{ex:mean}
Suppose we sequentially observe an i.i.d. scalar sequence $Z_1,Z_2,\dots$ with common mean $\theta =E[Z_1]$. So as to ensure that the conditional variance of the data we use as input of the testing procedure converges to 1, we use a variance stabilization device. Specifically we set $\widehat \sigma^2_t$ to be the empirical variance of $Z_1,\ldots,Z_t$, that is $\widehat \sigma^2_t=(t-1)^{-1} \sum_{s=1}^t (Z_s - \bar{Z}_t)^2$, with $\bar{Z}_t = t^{-1} \sum_{s=1}^t Z_t$.
Then, to test $\theta=\theta_0$ for a given $\theta_0$, we set $X_t=\omega_t(Z_t-\theta_0)$, where $\omega_t=(\widehat \sigma_{t-1} \vee t^{-\iota})^{-1}$, for some $\iota > 0$ such that $t^{-\iota}$ is a time-decreasing threshold that ensures that $\omega_t$ is finite. To construct a confidence sequence for $\theta$, we consider all $\theta_0$'s not rejected at time $t$. Namely, letting $\Gamma_t=\sum_{s=1}^t\omega_s$, $\widehat \theta_t=\Gamma_t^{-1}\sum_{s=1}^t\omega_s Z_s$, the delayed-start running MLE and the normal-mixture SPRT confidence sequences are 
\begin{align}
&\mathbb{R} \text{ for } t < t_0 \qquad \text{ and } \qquad  \left[\widehat{\theta}_t \pm \Gamma_t^{-1} c^{\rmlSPRT}_{\alpha,t_0}(t) \right] \qquad \text{for } t \geq t_0,\\
\text{ and } \qquad &\mathbb{R} \text{ for } t < t_0 \qquad \text{ and } \qquad \left[\widehat{\theta}_t \pm \Gamma_t^{-1} c^{\nmSPRT}_{\alpha,\lambda, t_0}(t) \right] \qquad \text{for } t \geq t_0,
\end{align}
respectively.
Our results provide guarantees about the probability that $\theta$ is ever excluded from this interval at \emph{any} time and the expected time until any one $\theta' \neq \theta$ is excluded.
\end{example}

\begin{example}[Bernoulli trial with covariates]\label{ex:abtest} Suppose we conduct a randomized controlled trial with two arms, in which we enroll patients sequentially. Upon enrolling the $t$-th patient, we collect a vector of pre-treatment covariates $L_t \in \mathbb{R}^d$ (potentially $d=0$ if we do not collect or opt not to use covariates), we randomly sample their treatment arm allocation $A_t \in \{0,1\}$ from a Bernoulli distribution with mean $p \in (0,1)$, and we then observe their health outcome $U_t \in \mathbb{R}$. We suppose that the triples $(L_1, A_1, U_1),\allowbreak (L_2, A_2, U_2),\allowbreak \ldots$ are i.i.d. Let the statistical parameter of interest be $\theta = E[U_1 \mid A_1=1] - E[U_1 \mid A_1=0]$. For inference on $\theta$, set 
\begin{align}
Z_t = \widehat{\eta}_{t-1}(1, L_t) - \widehat{\eta}_{t-1}(0, L_t) + \frac{A_t - p}{p(1-p)}\left(U_t - \widehat{\eta}_{t-1}(A_t,L_t)\right)
\end{align}
where $\widehat{\eta}_{t-1}$ is an $\mathcal{F}_{t-1}$-measurable estimate of the outcome regression function $\eta:(a,l) \mapsto E[U_1 \mid A_1 = a, L_1 = l]$, set $\widehat \sigma^2_{t-1}$ to be an $\mathcal{F}_{t-1}$ measurable estimate of $\mathrm{Var}(Z_t \mid \widehat{\eta}_{t-1})$ and apply the approach from \cref{ex:mean}.
\end{example} 

\section{Type I error}\label{sec:type-I_error}

Characterizing the type-I error of the delayed start rmlSPRT and nmSPRT is a priori easier if the data distribution under the null is known. As invariance principles guarantee convergence of partial sum processes of general nonparametric centered data sequences to Wiener (a.k.a. Brownian) processes, we start our type-I error study by characterizing the distributional properties of the test statistics under Brownian data. Let $W = (W(t))_{t \geq 0}$ be a Wiener process adapted to a filtration $\widetilde{\mathfrak{F}} = (\widetilde{\mathcal{F}}(t))_{t \geq 0}$.
Let $\widetilde{Y}_t^{\rmlSPRT}$ and $\widetilde{Y}_{\lambda, t}^{\nmSPRT}$ be the Brownian-data counterparts of $Y_t^{\rmlSPRT}$ and $Y_{\lambda, t}^{\nmSPRT}$:
\begin{align}
\widetilde{Y}_t^{\rmlSPRT} = \frac{1}{2} \left(\frac{W(t)^2}{t} - \log t \right) \qquad \text{and} \qquad \widetilde{Y}_{\lambda, t}^{\nmSPRT} = \frac{1}{2} \left( \frac{W(t)^2}{t+\lambda} - \log \frac{t + \lambda}{\lambda} \right).
\end{align}
The following result shows that choices $\widetilde{\alpha}_1(\alpha)$ and $\widetilde{\alpha}_{2, \lambda/t_0}(\alpha)$ above ensure that $-\log \widetilde{\alpha}_1(\alpha)$ and $-\log \widetilde{\alpha}_{2, \lambda / t_0}(\alpha)$ are $(1-\alpha)$-boundaries for the Brownian test statistics $\widetilde{Y}_t^{\rmlSPRT}$ and $\widetilde{Y}_{\lambda, t}^{\nmSPRT}$ when monitored continuously in time at all $t \geq t_0$.
\begin{lemma}\label{lemma:typeIerr_delayed_Wiener}
For any $t_0 \geq 1$ and $\lambda > 0$,
\begin{align}
&P\left[ \sup_{t \geq t_0} \widetilde{Y}_t^{\rmlSPRT} + \frac{1}{2} \log t_0 \geq -\log \widetilde{\alpha}_1(\alpha)  \right] = \alpha\\
 \text{ and } & P\left[ \sup_{t \geq t_0} \widetilde{Y}_{\lambda, t}^{\nmSPRT} \geq -\log \widetilde{\alpha}_{2, \lambda/t_0}(\alpha) \right] = \alpha.
\end{align}
Equivalently,
\begin{align}
&P\left[ \forall t \geq t_0,\ W(t) \in \mathcal{C}^{\rmlSPRT}_{\alpha, t_0}(t) \right] =  1-\alpha \\
\text{ and } & P\left[ \forall t \geq t_0,\ W(t) \in \mathcal{C}^{\nmSPRT}_{\alpha, \lambda, t_0}(t) \right] =  1-\alpha.
\end{align}
\end{lemma}
The claims on $\widetilde{Y}^{\nmSPRT}_{\lambda, t}$ and $\mathcal{C}^{\nmSPRT}_{\alpha, \lambda, t_0}$ are a direct consequence of Theorem 2 in \cite{robbins1970boundary}. The claims on $\widetilde{Y}_t^{\rmlSPRT}$ and $\mathcal{C}^{\rmlSPRT}_{\alpha, t_0}$ follow from similar techniques. Version 7 of \cite{smith} also includes a proof of these. We include a proof of \cref{lemma:typeIerr_delayed_Wiener} in the appendix for self-containedness.

We now show that type-I error converges along any joint sequence of rmlSPRTs or nmSPRTs and of values of $\alpha$, such that the burn-in time $t_0$ diverges to $\infty$. The guarantees will hold even with discrete-time, non-normal, dependent observations. The proof hinges on constructing a joint probability space for both these observations and the continuous-time Brownian motion, in which $\sup_{t \geq t_0} |\widetilde{Y}_t^{\rmlSPRT} - Y_t^{\rmlSPRT}|$ and $\sup_{t \geq t_0} |\widetilde{Y}_{\lambda, t}^{\nmSPRT} - Y_{\lambda,t}^{\nmSPRT}|$ are small. 
Almost-sure approximations of partial sums by Brownian motions under lax conditions are the object of strong invariance principles.
We will specifically leverage the SIP result from theorem 1.3 of \cite{strassen1967}. To do so, we impose the following conditions. 
\begin{assumption}\label{asm:sip} Let $V_t=\sum_{s=1}^tv_s$. For some non-decreasing function $f:[0,\infty) \to (0,\infty)$ such that $f(t) = O(t (\log t)^{-4} (\log \log t)^{-2})$, it holds that
\begin{align}
\sum_{t=1}^\infty \frac{1}{f(V_t)}E\left[\left(X_t- \psi_t\right)^2 \bm{1} \left\lbrace (X_t-\psi_t)^2>f(V_t) \right\rbrace \mid \mathcal{F}_{t-1} \right]<\infty, \label{eq:sip_tails} 
\end{align}
and
\begin{align}
 V_t - t = o \left( \sqrt{t f(t)} \log t \right) \qquad \text{almost surely.} \retainlabel{eq:sip_var}
\end{align}
\end{assumption}
For now, we take \cref{asm:sip} as a primitive assumption. In \cref{sec:application}, we will discuss simple sufficient conditions for it for the case of stabilized estimating equations. 
Condition \eqref{eq:sip_tails} is akin to (but different from) a martingale Lindeberg condition, and in \cref{sec:application} we satisfy it by assuming a moment higher than 2 exists.

A SIP guarantee is, of course, asymptotic, so we need to consider regimes where the stopping time is large.
One such regime is when we set a long burn-in period $t_0$, essentially to wait for the data to look normal. The following theorem provides a type-I error guarantee (equivalently, a coverage guarantee) along sequences $((t_{0,m}, \alpha_m, \lambda_m))_{m\geq 1}$ such that  $t_{0,m} \to \infty$.
\begin{theorem}\label{thm:typeI_error}
Suppose that \cref{asm:sip} holds. Then, for any sequence $((t_{0,m},\allowbreak \alpha_m,\allowbreak \lambda_m))_{m \geq 1}$ such that $t_{0,m} \to \infty$, $\lambda_m > 0$, $\alpha_m \in [0,1]$ $\forall m \geq 1$, 
\begin{align}
&\lim_{m \to \infty} \Pr\left[ \exists t \geq t_{0,m},\ Y^{\rmlSPRT}_t + \frac{1}{2} \log t_{0,m} \geq -\log \widetilde{\alpha}_1(\alpha_m)  \right] / \alpha_m = 1,\\
 \text{ and }  &\lim_{m \to \infty} \Pr \left[\exists t \geq t_{0,m},\  Y^{\nmSPRT}_{\lambda_m, t} \geq -\log \widetilde{\alpha}_{2, \lambda_m / t_{0,m}}(\alpha_m) \right] / \alpha_m = 1.
\end{align}
Equivalently,
\begin{align}
&\lim_{m \to \infty} \Pr \left[ \exists t \geq t_{0,m},\ S_t \not\in \mathcal{C}^{\rmlSPRT}_{\alpha_m, t_{0,m}}(t) \right] / \alpha_m = 1,\\
\text{ and } &\lim_{m \to \infty} \Pr \left[ \exists t \geq t_{0,m},\ S_t \not\in \mathcal{C}^{\nmSPRT}_{\alpha_m, \lambda_m, t_{0,m}}(t) \right]  / \alpha_m = 1.
\end{align}
\end{theorem}

\section{Representation of boundaries as inverted running-effect-size-estimate SPRTs}\label{sec:representation}

Before presenting our representation results for the rmlSPRT and nmSPRT statistics, let us first discuss known related facts. Let $\widetilde{S}(t) = \widetilde{\psi} t + W(t)$ be a time-continuous Brownian observation sequence with constant drift $\widetilde{\psi}$. \cite{robbins_siegmund1974} show via It\^o's lemma that, under input data process $\widetilde{S}$, the nmSPRT test statistic $\widetilde{Y}^{\nmSPRT}_{\lambda, t} = \frac{1}{2}\{\widetilde{S}^2(t) / (t+\lambda) - \log ((t+\lambda)/\lambda)\}$ can be represented alternatively as $\widetilde{Y}^{\nmSPRT}_{\lambda, t} = \int_{0}^t \widetilde{\psi}_{\lambda}(s) d\widetilde{S}(s) - \frac{1}{2} \widetilde{\psi}_{\lambda}^2(s) ds$, where $\widetilde{\psi}_{\lambda}(t) = \widetilde{S}(t) / (t + \lambda)$ is the posterior mean of the drift $\widetilde{\psi}$ under prior $\mathcal{N}(0, 1/\lambda)$. We refer to this latter representation as a \emph{running-posterior-mean SPRT}. In the same article, they discuss the stopping time properties of a related test of which the test statistic is $\int_0^t \widetilde{\psi}_0(s) d\widetilde{S}(s) - \frac{1}{2}\widetilde{\psi}^2_0(s) ds$. Observing that $\widetilde{\psi}_0(t)$ is the maximum likelihood estimate of $\widetilde{\psi}$, we refer to this latter test statistic as a \emph{running-maximum-likelihood-estimate SPRT}. As \cite{robbins_siegmund1974} observe, these running-effect-size-estimate SPRT representations are amenable to rejection time analysis. This motivates us to look for such representations of the rmlSPRT and nmSPRT statistics.

While we cannot directly apply It\^o's lemma in our nonparametric discrete-time setting, we derive a finite-differences equivalent of a certain It\^o-derived identity, namely that for $h_{\lambda}(x,t) = \frac{1}{2} x^2 / (t+\lambda)$, $dh_\lambda(\widetilde{S}(t),t) = \widetilde{\psi}_\lambda(t) d\widetilde{S}(t) - \frac{1}{2} \widetilde{\psi}_\lambda^2(t) + \frac{1}{2} (t+\lambda)^{-1}$.
The following lemma is our discrete analog of this identity.
\begin{lemma}\label{lemma:discrete_Ito_identity}
Let $(z_s)_{s \in \Natints}$ be a sequence of real numbers, and let $\lambda \geq 0$ and $s \in \mathbb{N}$. Whenever $s+\lambda >0$, it holds that
\begin{align}
    &\frac{1}{2} \left(\frac{z_{s+1}^2}{s+1+\lambda} -  \frac{z_s^2}{s+\lambda}\right) = \frac{s+\lambda}{s+1+\lambda} \left( \frac{z_s}{s+\lambda}(z_{s+1} - z_s) - \frac{1}{2} \left(\frac{z_s}{s+\lambda}\right)^2 \right) + \frac{1}{2} \frac{(z_{s+1} - z_s)^2}{s+1+\lambda}.
\end{align}
\end{lemma}
As a direct corollary of \cref{lemma:discrete_Ito_identity} we obtain the following representation results for the rmlSPRT and nmSPRT statistics defined in \eqref{eq:test_statistics_definition}. Let $X^0_t = X_t - E[X_t | \mathcal{F}_{t-1}]$, $S^0_t = \sum_{s=1}^t X^0_s$, $\breve{X}_t = \psi + X^0_t$, $\breve{S}_t = \sum_{s=1}^t \breve{X}_s$. where $\psi = 0$ under $H_0$ and $\psi= \lim_t \psi_t$ under $H_{\backslash 0}$.

\begin{theorem}\label{thm:rep_non_anticip_magle}
For any $\lambda > 0$ and $t \geq 1$, it holds that 
\begin{align}
Y^{\rmlSPRT}_t =& \sum_{s=0}^{t-1} \rho_{0,s} \left( \breve{\psi}_{0,s} \breve{X}_{s+1} - \frac{1}{2} \breve{\psi}_{0,s}^2 \right) + \frac{1}{2} \left(\sum_{s=1}^t \frac{\breve{X}_s^2}{s} - \log t  \right) + R^{\bias}_{0,t} \\
\text{and} \qquad Y^{\nmSPRT}_{\lambda, t} =& \sum_{s=0}^{t-1} \rho_{\lambda, s} \left( \breve{\psi}_{\lambda,s} \breve{X}_{s+1} - \frac{1}{2} \breve{\psi}_{\lambda,s}^2 \right) + \frac{1}{2} \left( \sum_{s=1}^t \frac{\breve{X}_s^2}{s + \lambda} - \log \frac{t + \lambda}{\lambda} \right) + R^{\bias}_{\lambda, t},
\end{align}
where $\rho_{\lambda, s} = (s+\lambda) / (s + 1 + \lambda)$, $\breve{\psi}_{0,0} = 0$, $\breve{\psi}_{\lambda, s} = \breve{S}_s / (\lambda + s)$ for any $s, \lambda \geq 0$ such that $s+\lambda > 0$, and $R^\bias_{\lambda, t} = (S_t^2 - \breve{S}_t^2) / (2(t+\lambda))$.
\end{theorem}

\cref{thm:rep_non_anticip_magle} shows that the test statistics $Y^{\rmlSPRT}_t$ and $Y_{\lambda,t}^{\nmSPRT}$ can be represented, up to some remainder terms, as sums of log-likelihood ratios of which the numerator is the likelihood under an \textit{non-anticipating} (that is, for the $s$-th term, $\mathcal{F}_{s-1}$-measurable) estimate of the parameter. The \textit{non-anticipating martingale} terminology was introduced by \cite{lorden2005nonanticipating} to refer to the device introduced by \cite{robbins1970statistical} and \cite{robbins_siegmund1974} to construct their test statistics.

The log non-anticipating martingales can be further expanded to yield the representations in the following theorem.

\begin{theorem}\label{thm:expanded_representation}
For any $\lambda > 0$ and $t \geq 1$, it holds that 
\begin{align}
Y^\rmlSPRT_t =& \frac{1}{2} \psi^2 t + M^{(0)}_t + M^{(1)}_{0,t} + R^\bias_{0,t} - R^\adpt_{0,t} + \Deltaqvar_{0,t} + \frac{1}{2} (X^0_1)^2\\
\text{and} \qquad Y^\nmSPRT_{\lambda, t} =& \frac{1}{2} \psi^2 t + M^{(0)}_t + M^{(1)}_{\lambda,t} - M^\skg_{\lambda,t} + R^\bias_{\lambda,t} - R^\adpt_{\lambda,t} + \Deltaqvar_{\lambda,t} - R^{\skg,1}_{\lambda, t} + R^{\skg,2}_{\lambda,t}
\end{align}
where
\begin{align}
M^{(0)}_t =& \psi S^0_t, \qquad M^{(1)}_{\lambda, t} = \sum_{s=0}^{t-1} \rho_{\lambda, s} \psi^0_{\lambda, s} X^0_{s+1}, \qquad M^\skg_{\lambda, t} = \psi \lambda \sum_{s=1}^t \frac{X^0_s}{s+\lambda},\\
\psi^0_{\lambda, s} =& \frac{S^0_s}{s + \lambda}, \\
R^\adpt_{\lambda,t} =& \frac{1}{2} \sum_{s=0}^{t-1} \rho_{\lambda, s} (\psi^0_{\lambda, s})^2,\\
\Deltaqvar_{\lambda,t} =&\begin{cases}
 \frac{1}{2} \left(\sum_{s=1}^t \frac{(X^0_s)^2}{s + \lambda} - \log \frac{t +\lambda}{\lambda} \right) \text{ if } \lambda > 0, \\
\frac{1}{2} \left( \sum_{s=2}^t \frac{(X^0_s)^2}{s} - \log t \right) \text{ if } \lambda =0,
\end{cases}\\
R^{\skg,1}_{\lambda, t} =& \frac{1}{2} \psi^2 \frac{\lambda t}{t + \lambda} \qquad \text{and} \qquad R^{\skg, 2}_{\lambda, t} = \psi \lambda \sum_{s=0}^{t-1} \frac{S^0_s}{(s+\lambda)(s+1+\lambda)}.
\end{align}
\end{theorem}
Observe that $(M^{(0)}_t)_{t \geq 1}$, $(M^{(1)}_{\lambda, t})_{t \geq 1}$ and $(M^\skg_{\lambda,t})_{t\geq 1}$ are martingales with initial expectation 0. This will facilitate their analysis via the optional stopping theorem in the rejection time analysis. The terms with the superscript ``skg'' are what we refer to as \textit{shrinkage terms}. Their presence arises from the presence of the shrinkage parameter $\lambda$ in the running mean estimates $(\psi_{\lambda,s})$, and they converge to zero as $\lambda \to 0$. The $\Deltaqvar_{\lambda, t}$ term is the difference between the quadratic variation of a certain discrete time martingale and that of its time-continuous Brownian approximation. 

The terms $(\psi^0_{\lambda, s})$ are the centered errors in the estimates $(\psi_{\lambda, s})$. The term $R^\adpt_{\lambda,t}$ behaves roughly as the sum of the squared errors of the estimates $\psi_{\lambda, s}$. As will appear explicitly from \cref{lemma:stopping_time_and_random_thresh}, it contributes positively to the rejection time. We interpret it as the cost of adaptively estimating $\psi$ relative to knowing it a priori.

\section{Expected rejection time}\label{sec:expected_rejection_time}

\subsection{Rejection time after burn-in as a function of a random threshold}

The presence of the burn-in period in the delayed start running-mean estimate SPRTs adds a slight level of complexity to rejection time analysis as compared to when monitoring starts from the beginning of data collection.  Fortunately, one can observe that monitoring the crossing of a fixed threshold $A$ by $Y_{\lambda,t}$ after $t_0$ steps is the same as monitoring the crossing of the offset threshold $A - Y_{\lambda, t_0}$ by the offset test statistic $Y_{\lambda, t} - Y_{\lambda, t_0}$. Therefore, conditional on $\mathcal{F}_{t_0}$, we can analyze the rejection time from $t_0$ as we would do in a situation without burn-in, with the difference that the threshold is now the offset threshold $A - Y_{\lambda, t_0}$. We can thus readily obtain a characterization of the expected rejection time given $\mathcal{F}_{t_0}$ given the $\mathcal{F}_{t_0}$-measurable random threshold $A - Y_{\lambda, t_0}$. In what follows, we introduce the shorthand notation $\tau_1 = \tau^\rmlSPRT(\alpha, t_0)$ and $\tau_2 = \tau^\nmSPRT(\alpha, \lambda, t_0)$.

The following lemma connects the test statistics, the stopping times, and the random rejection thresholds. 

\begin{lemma}\label{lemma:stopping_time_and_random_thresh}
    For any $\lambda > 0$, $t_0 \geq 1$, $\alpha \in (0,1)$, 
    \begin{align}
       & Y^\rmlSPRT_{(\tau_1-1)\vee t_0} - Y^\rmlSPRT_{t_0} \leq  \left(-\log \widetilde{\alpha}_1(\alpha) - \frac{1}{2} \log t_0 - Y^\rmlSPRT_{t_0} \right)_+ \leq
     Y^\rmlSPRT_{\tau_1} - Y^\rmlSPRT_{t_0} \\
     &\text{and } \\
     &Y^\nmSPRT_{\lambda, (\tau_2-1)\vee t_0} - Y^\nmSPRT_{\lambda, t_0}  \leq  \left(-\log \widetilde{\alpha}_{2, \lambda/ t_0} (\alpha)  - Y^\nmSPRT_{\lambda, t_0} \right)_+ \leq Y^\nmSPRT_{\lambda, \tau_2} - Y^\nmSPRT_{\lambda, t_0}.
    \end{align}
\end{lemma}

So as to obtain bounds on the marginal expected (as opposed to conditional on $\mathcal{F}_{t_0}$) rejection times $E[\tau_1]$ and $E[\tau_2]$, we need to characterize the marginal expectations (that is the expectations w.r.t. the distributions of $Y^\rmlSPRT_{t_0}$ and $Y^\nmSPRT_{\lambda, t_0}$) of the random rejection thresholds. We make the following three assumptions.

\begin{assumption}[Expected conditional mean convergence]\label{asm:L1_conv_psi}
It holds that $E[|\psi_t - \psi|^2] / \psi^2 \to 0$ as $t \to \infty$.
\end{assumption}

\begin{assumption}[L1-convergence of conditional variance]\label{asm:L1_conv_vs}
It holds that $E[|v_t-1|] \to 0$.
\end{assumption}

\begin{assumption}[Classic conditional Lindeberg]\label{asm:cond_Lindeberg}
For any $\epsilon > 0$, it holds that $$\lim_{t \to \infty} \sum_{s=1}^t E\left[(t^{-1/2} X_s^0)^2 \bm{1} \left\lbrace |t^{-1/2} X_s^0| \geq \epsilon \right\rbrace \mid \mathcal{F}_{s-1} \right] = 0.$$
\end{assumption}

\begin{lemma}\label{lemma:random_thresh_lemma1}
Suppose \cref{asm:L1_conv_psi}, \cref{asm:L1_conv_vs} and \cref{asm:cond_Lindeberg} hold. Then, as $\psi \sqrt{t_0} \to 0$,
\begin{align}
&E\left[ \left(-\log \widetilde{\alpha}_1(\alpha) - \frac{1}{2} \log t_0 - Y^\rmlSPRT_{t_0} \right)_+ \right] = \widetilde{h}_0(-\log \widetilde{\alpha}_1(\alpha)) + o(1),\\
\text{and }  &E \left[ \left(-\log \widetilde{\alpha}_{2, \lambda / t_0}  - Y^\nmSPRT_{\lambda, t_0} \right)_+ \right] = \widetilde{h}_{\lambda / t_0}\left(-\log \widetilde{\alpha}_{2, \lambda / t_0}(\alpha) - \frac{1}{2} \log \frac{\lambda}{t_0+\lambda}\right) + o(1)
\end{align}
where, for any $a \in \mathbb{R}$, $\eta \geq 0$,
\begin{align}
\widetilde{h}_{\eta}(a) =& \left(a - \frac{1}{2(1+\eta)} \right) \left( 2 \Phi\left( \sqrt{2(1+\eta) a} \right) - 1 \right)  + \exp\left(-(1+\eta) a \right) \sqrt{\frac{a}{\pi (1 + \eta)}}.
\end{align}
\end{lemma}
The following lemma characterizes the behavior of the expected random thresholds as $\alpha \to 0$.
\begin{lemma}\label{lemma:random_thresh_lemma2}
Suppose $\alpha \to 0$ and $\eta \to \infty$. Then
\begin{align}
\widetilde{h}_0(-\log\widetilde{\alpha}_1(\alpha)) \sim - \log \alpha \qquad \text{and} \qquad \widetilde{h}_{\eta}\left(-\log \widetilde{\alpha}_{2, \eta}(\alpha) - \frac{1}{2} \log \frac{\eta}{1+\eta}\right) \sim -\log \alpha.
\end{align}
\end{lemma}

The decompositions of test statistics from \cref{thm:expanded_representation}, together with \cref{lemma:stopping_time_and_random_thresh} and \cref{lemma:random_thresh_lemma1} imply the following bounds on the expected stopping times.
\begin{corollary}\label{corollary:bound_stopping_time_remainders}
Suppose \cref{asm:L1_conv_psi} and \cref{asm:L1_conv_vs} hold. Then, as $t_0 \to \infty$, $\psi \sqrt{t_0} \to 0$,
\begin{align}
E[\tau_1] \leq & 2 \psi^{-2} \left( \widetilde h_0(-\log \widetilde{\alpha}_1(\alpha)) + E\left[R^\adpt_{0,\tau_1} - R^\adpt_{0,t_0}\right]+ E\left[R^\bias_{0,\tau_1} - R^\bias_{0,t_0}\right] \right. \\
& \qquad \left. + E\left[\Deltaqvar_{0,\tau_1} - \Deltaqvar_{0,t_0} \right] \right) + o(\psi^{-2}),\\
\text{and} \qquad E[\tau_2] \leq & 2 \psi^{-2} \left( \widetilde h_{\lambda / t_0}(-\log \widetilde{\alpha}_{2, \lambda / t_0}(\alpha)) + E\left[R^\adpt_{\lambda,\tau_2} - R^\adpt_{0,t_0}\right]+ E\left[R^\bias_{\lambda,\tau_2} - R^\bias_{\lambda,t_0}\right]  \right. \\
& \qquad \left.+ E\left[\Deltaqvar_{\lambda,\tau_2} - \Deltaqvar_{\lambda,t_0} \right] + E\left[R^{\skg, 1}_{\lambda,\tau_2} - R^{\skg, 1}_{\lambda, \tau_2}\right]  \right. \\  &\qquad \left.+ E\left[R^{\skg, 2}_{\lambda,\tau_2} - R^{\skg, 2}_{\lambda, \tau_2}\right] \right) + o(\psi^{-2})
\end{align}
\end{corollary}

From \cref{corollary:bound_stopping_time_remainders} above, we will be able to obtain bounds on the expected rejection times if we can bound the expected remainder terms and the expected quadratic variation difference $E\left[\Deltaqvar_{\lambda,\tau_i} - \Deltaqvar_{\lambda,t_0} \right]$, $i=1,2$. We study these terms in  \cref{sec:bounding_ER_adapt}, \cref{sec:bounding_ER_bias}, and \cref{sec:bounding_other_remainder_terms}.

\subsection{Asymptotic equivalent of the adaptivity remainder term}\label{sec:bounding_ER_adapt}

In bounding the adaptivity remainder term $E[R^\adpt_{0,\tau_i} - R^\adpt_{0,t_0}]$, $i=1,2$, we use a technique inspired by the proof of lemma 8 in \cite{robbins_siegmund1974}. As \cite{robbins_siegmund1974}, we point out that we should expect $E[\tau_i]$, $i=1,2$, to be at least as large as the expected rejection time of the simple-vs-simple SPRT under the alternative $(\psi_t = \psi,\ \forall t\geq 1)$, that is $-2 \psi^{-2} \log \alpha$. Therefore, since $E[(\psi^0_{0,s})^2] \sim s^{-1}$ under \cref{asm:L1_conv_vs}, we expect that
\begin{align}
E\left[R^\adpt_{0,\tau_i} - R^\adpt_{0,t_0}\right] = \frac{1}{2} E\left[ \sum_{s=1}^{\tau_i}  (\psi^0_{0,s})^2  \right] \geq \frac{1}{2}  \sum_{s=t_0 + 1}^{\left\lfloor -2 \psi^{-2} \log \alpha \right\rfloor} E\left[(\psi^0_{0,s})^2\right] \sim \log \psi^{-1},
\end{align}
as $\psi \to 0$, $t_0 \to \infty$, $\psi \sqrt{t_0} \to 0$, under fixed $\alpha$. The upcoming lemma shows that this is indeed the correct asymptotic equivalent. The result relies on the following assumptions.

\begin{assumption}\label{asm:fourth_moment}
It holds that $\sup_{t \geq 1} t^{-2} \sum_{s=1}^t E\left[(X^0_s)^4\right] < \infty$.
\end{assumption}

\begin{assumption}\label{asm:sup_L2_norm_vs_min1}
It holds that $\sup_{t \geq 1} E \left[ ( v_t - 1 )^2 \right] < \infty$.
\end{assumption}

\begin{lemma}\label{lemma:adaptivity_term}
Suppose \cref{asm:L1_conv_vs}, \cref{asm:fourth_moment} and \cref{asm:sup_L2_norm_vs_min1} hold. Then, for any $\lambda > 0$, as $\psi \to 0$, $\psi \sqrt{t_0} \to 0$.
\begin{align}
E\left[R^\adpt_{0,\tau_1} - R^\adpt_{0,t_0}\right] =& \log \psi^{-1} + \psi^2 (\log \psi^{-1})^{-1/2} E [\tau_1] + o\left( \log \psi^{-1} \right), \\
\text{and} \qquad E\left[R^\adpt_{\lambda,\tau_2} - R^\adpt_{\lambda,t_0}\right] =& \log \psi^{-1} + \psi^2 (\log \psi^{-1})^{-1/2} E [\tau_2] + o\left( \log \psi^{-1} \right).
\end{align}

\end{lemma}

\subsection{Bounding the bias remainder term}\label{sec:bounding_ER_bias}

Our bound on the bias term relies on the following two assumptions.
\begin{assumption}\label{asm:relative_bias_moments_summable} There exists $C < \infty$ and $\delta > 0$ such that, for any $\psi > 0$, $\sum_{t = 1}^\infty E[(|\psi_t - \psi| / \psi)^{2  + \delta}] < C$.
\end{assumption}

\begin{assumption}\label{asm:Evsmin1_over_s_summable}
$(t^{-1}E[|v_t-1|])_{t\geq 1}$ is summable.
\end{assumption}
Note that \cref{asm:relative_bias_moments_summable} implies \cref{asm:L1_conv_psi}, and that \cref{asm:Evsmin1_over_s_summable} implies \cref{asm:L1_conv_vs}.

\begin{lemma}\label{lemma:bias_term}
Suppose that \cref{asm:relative_bias_moments_summable} and \cref{asm:Evsmin1_over_s_summable} hold. Then, for any $\lambda \geq 0$ and for any sequence of $\mathfrak{F}$-adapted stopping times $(\tau(t_0))_{t_0 \geq 1}$ such that $\tau(t_0) \geq t_0$ and $\tau(t_0)$ is almost surely finite, 
\begin{align}
E\left[R^\bias_{\lambda,\tau(t_0)} - R^\bias_{\lambda,t_0}\right] = o\left( \psi^2 E[\tau(t_0)] + \psi \sqrt{E[\tau(t_0)] \log E[\tau(t_0)]} + 1 \right)
\end{align}
\end{lemma}

\subsection{Bounding the shrinkage and the quadratic variation difference terms}\label{sec:bounding_other_remainder_terms}

\begin{lemma}\label{lemma:ERskg1}
For $\lambda > 0$, $\lambda = o(\psi^{-2} \log \psi^{-1})$ and for any sequence of $\mathfrak{F}$-adapted stopping times $(\tau(t_0))_{t_0 \geq 1}$ such that $\tau(t_0) \geq t_0$, we have that
$R^{\skg, 1}_{\lambda, \tau(t_0)} - R^{\skg, 1}_{\lambda, t_0} = o(\log \psi^{-1})$.
\end{lemma}

\begin{lemma}\label{lemma:ERskg2}
Suppose \cref{asm:L1_conv_vs} holds. For $\lambda > 0$, $\lambda = o(\psi^{-2} (\log \psi^{-1})^2)$ and for any sequence of $\mathfrak{F}$-adapted stopping times $(\tau(t_0))_{t_0 \geq 1}$ such that $\tau(t_0) \geq t_0$ and $\tau(t_0)$ is almost surely finite, we have that $E\left[R^{\skg, 2}_{\lambda, \tau(t_0)} - R^{\skg, 2}_{\lambda, t_0}\right] = o(\log \psi^{-1})$.
\end{lemma}

\begin{lemma}\label{lemma:Deltaqvar}
Suppose \cref{asm:Evsmin1_over_s_summable} holds. Then, for any $\lambda \geq 0$ and for any sequence of $\mathfrak{F}$-adapted stopping times $(\tau(t_0))_{t_0 \geq 1}$ such that $\tau(t_0) \geq t_0$ and $\tau(t_0)$ is almost surely finite, we have that, $E\left[\Deltaqvar_{\lambda,\tau(t_0)} - \Deltaqvar_{\lambda,t_0} \right] = O(1)$.
\end{lemma}

\subsection{Expected rejection time: main theorem}

We obtain upper bounds on the expected rejection times as a direct consequence of \cref{corollary:bound_stopping_time_remainders}, the bounds on the remainder terms, and the fact that under assumptions \ref{asm:fourth_moment}-\ref{asm:relative_bias_moments_summable}, the rejection time is almost surely finite (\cref{prop:as_finite_tau} in the appendix). We consider two asymptotic regimes under which these bounds hold. We define these asymptotic regimes below.

\begin{definition}[Asymptotic regimes]
We say that a sequence $((\psi_m, \allowbreak\alpha_m,\allowbreak t_{0,m},\allowbreak \lambda_m))_{m \geq 1}$ follows Asymptotic Regime 1 (AR(1)) (resp. Asymptotic Regime 2, AR(2)) if it satisfies the conditions of the first (resp. second) column of Table \ref{table:ARs}. For $i=1,2$, We write $\lim_{AR(i)} a(\psi, \alpha, t_0, \lambda) = b$ if, for any sequence $((\psi_m, \alpha_m, t_{0,m}, \lambda_m))_{m \geq 1}$ following AR(i), it holds that $\lim_{m \to \infty} a(\psi_m, \alpha_m, t_{0,m}, \lambda_m) = b$.
\end{definition}

\begin{table}
\centering
\begin{tabular}{c|c}
Asymptotic regime 1 & Asymptotic Regime 2 \\
\hline
\multicolumn{2}{c}{$t_{0,m} \to \infty$} \\
\hline
\multicolumn{2}{c}{$\lambda_m = o(\psi_m^{-2} \log \psi_m^{-1})$} \\
\hline
\multicolumn{2}{c}{$\lambda_m / t_{0,m} \to \infty$} \\
\hline
\multicolumn{2}{c}{$\psi_m \sqrt{t_{0,m}} \to \infty$}  \\
\hline
$\log \alpha_m^{-1} / \log \psi_m^{-1} \to \infty$ & $\log \alpha_m^{-1} / \log \psi_m^{-1} \to 0$ 
\end{tabular}
\caption{Asymptotic regimes definition}
\label{table:ARs}
\end{table}

\begin{theorem}\label{thm:rejection_time_upper_bounds}
Suppose that assumptions \ref{asm:cond_Lindeberg}-\ref{asm:Evsmin1_over_s_summable} hold. Then, for any $i=1,2$,
\begin{align}
\underset{AR(1)}{\mathrm{\lim sup}} \frac{E[\tau_i(t_0)]}{2 \psi^{-1} \log \psi^{-1}} \leq  1
\qquad \text{and} \qquad 
\underset{AR(2)}{\mathrm{\lim sup}} \frac{E[\tau_i(t_0)]}{2 \psi^{-1} \log \alpha^{-1}} \leq  1.
\end{align}
\end{theorem}

\begin{remark}
The two asymptotic regimes we consider in the above theorem are defined in particular by the condition that $\psi \sqrt{t_0} \to 0$. This implies that rejections happen much later than the end of the burn-in period. This is usually a realistic condition as, in many practical situations, near-normality is achieved at sample sizes of a few hundred observations, while detection of the effect size requires much larger sample sizes.
\end{remark}

\begin{remark} The upper bound in the second asymptotic regime above (when $\log \psi^{-1} = o(\log \alpha^{-1})$) is asymptotically equivalent to the expected rejection time $2 \psi^{-2} \log \alpha^{-1}$ of the simple-vs-simple SPRT under i.i.d. data, which is known to be optimal among all sequential tests with type-I error at most $\alpha$ \citep{waldwolfowitz}.
\end{remark}

\begin{remark}
\cite{robbins_siegmund1974} show that the expected sample size (of a slightly different version) of the rmlSPRT under i.i.d. observations from a distribution lying in an exponential family parametric model is asymptotically equivalent to $(1-\alpha) \psi^{-2} \log \psi^{-1}$ as $\psi \to 0$ and $\log \alpha^{-1} = o(\log \psi^{-1})$, which is approximately twice smaller (as $1-\alpha \approx 1$ for small $\alpha$) than the upper bound we provide here for our general nonparametric setting.
\end{remark}

\section{Application to Stabilized Estimating Equations with Sequentially Estimated Nuisances}\label{sec:application}

We next consider one simple setting where we can apply delayed-start rmlSPRTs and nmSPRTs, and consider interpretable sufficient conditions for our results to hold.

\subsection{Robust estimating equations.} Suppose we observe an i.i.d. sequence $O_1,O_2,\dots$ and we wish to test whether a parameter $\theta$ of the common distribution of the observations is equal to a certain value $\theta_0$. Suppose that $\eta$ is a (potentially infinite dimensional) nuisance parameter lying in a set $\mathcal{T}$ and that we have an estimating function $D(\cdot ; \eta', \theta')$ for $\theta$ defined over the observation space that satisfies the following robustness assumption.
\begin{assumption}[Identification and robustness]\label{asm:robustness}
The estimating function $D$ is such that, 
\begin{align}
\forall \theta' \in \Theta, \eta' \in \mathcal{T},\ E[D(O_1 ; \eta', \theta')] = \mu(\theta')
\end{align}
for some mapping $\mu : \Theta \to \mathbb{R}$ that is continuous at $\theta$ and is such that $\mu(\theta') = 0$ if, and only if, $\theta' = \theta$.
\end{assumption}

\subsection{Application to our examples.}
\Cref{ex:mean} fits into this framework with $D(z,\theta_0) = z  - \theta_0$ to test $\theta = E[Z_1] = \theta_0$, and \cref{ex:abtest} with 
\begin{align}
D((l,a,u); \eta', \theta_0) = \frac{a-p}{p(1-p)} ( u - \eta'(a,l)) + \eta'(1,l) - \eta'(0,l)  - \theta_0
\end{align}
to test $\theta = E[U_1\mid A_1=1]-E[U_1\mid A_1=0] = \theta_0$.

\subsection{Construction of stabilized estimating equations via sequential sample splitting.}

We now apply our method to this setting. We use a combination of sequential estimation and sample splitting in order to avoid any metric entropy assumptions. Specifically, we will estimate nuisances on half of all the past data and estimate the variance on the other half. Let $\mathcal{I}_{0,t} = \{ t, t-2, t-4,\ldots \}$ and $\mathcal{I}_{1,t} = \{t-1, t-3, t-5, \ldots	\}$ index two complementary sample splits.
Fix some $\widehat\eta_t$ sequence adapted to $\mathfrak{F}$ such that $\widehat\eta_t$ is independent of $\mathcal{D}_{0,t} = \{O_s : s \in \mathcal{I}_{0,t}\}$ given $\mathcal{D}_{1,t} = \{ O_s : s \in \mathcal{I}_{1,t} \}$, that is, an ``estimate" of $\eta$ based only on the ${\lceil t/2\rceil}$ data points from times $\mathcal{I}_{1,t}$.
Set 
\begin{align}
\widehat \sigma^2_t(\theta')=&\frac1{|\mathcal{I}_{0,t}|-1}\sum_{s\in\mathcal{I}_{0,t}}  \left( D(O_s; \widehat{\eta}_t, \theta') - \bar{D}_t(\theta') \right)^2\\
\qquad \text{with} \qquad \bar{D}_t(\theta') =& \frac{1}{|\mathcal{I}_{0,t}|}\sum_{s \in \mathcal{I}_{0,t}} D(O_s; \widehat{\eta}_t, \theta').
\end{align}
Then, to test $\theta = \theta_0$, set $X_t=\omega_t(\theta_0) D(O_t ; \widehat{\eta}_{t-1}, \theta_0)$, where 
$\omega_t(\theta')=(\widehat \sigma_{t-1} (\theta') \vee (\chi^{-1} t^{-\iota}))^{-1}$ with some $\chi>0,\iota\in(0,1)$.
Then, under \cref{asm:robustness}, the hypothesis that $X_1,X_2,\dots$ is an MDS holds if and only if $\theta = \theta_0$ holds.

Note that if we have no nuisances, as in \cref{ex:mean}, then we can just use \emph{all} data up to $t-1$ to compute $\hat\sigma^2_t$, not just the past data having the same parity as $t-1$. We can actually do this as long as $\mathcal T$ is sufficiently simple (\eg, a subset of $\mathbb{R}^d$ for some fixed $d$, rather than, say, a space of nonparametric functions). However, to avoid any such assumptions altogether, we focus here on the case where we split the data by parity. Similarly, we here analyze the case where we clip $\omega_t$ by $O(t^\iota)$ to control for the risk of outlying variance estimates, but this is mostly done to make analysis simple. In practice, we do not recommend this, and in our experiments in \cref{sec:numerics}, we simply recommend using $\omega_t=1$ whenever the variance estimate is zero and not clipping at any other value. This reduces the need to specify hyperparameters.

\subsection{Confidence sequences.}
As mentioned earlier, a confidence sequence $(\mathcal{C}_{\alpha, t_0}(t))_{t \geq 1}$ for $\theta$ may be obtained by setting $\mathcal{C}_{\alpha, t_0}(t)$ to the set of values $\theta_0$ that  the test of $\theta = \theta_0$ doesn't reject at $t$. That is, letting $c_{\alpha, t_0}(t)$ to be either $c_{\alpha, t_0}^{\rmlSPRT}(t)$ or $c_{\alpha, \lambda, t_0}^{\nmSPRT}(t)$, $\mathcal{C}_{\alpha, t_0}(t) = \left\lbrace \theta_0 \in \Theta : |S_t| \leq c_{\alpha, t_0}(t) \right\rbrace$. In the case where $D(o; \eta', \theta')=D_1(o; \eta') + D_2(o; \eta')  \theta'$ is linear in a scalar parameter $\theta'$, with $\theta =-E[D_2(O; \eta')]^{-1}  E[D_1(O; \eta')]$ being the parameter of interest, the resulting confidence sequence simplifies considerably: let 
\begin{align}
\Gamma_t=\sum_{s=1}^t\omega_s D_2(O_s ; \widehat{\eta}_{s-1}), \qquad \text{and} \qquad \widehat \theta_t=\Gamma_t^{-1}\sum_{s=1}^t\omega_s D_1(O_s;\widehat{\eta}_{s-1}),
\end{align}
then our confidence sequence is given by $\mathcal{C}_{\alpha, t_0}(t) = [\widehat{\theta}_t \pm \Gamma_t^{-1} c_{\alpha, t_0}(t)]$.
Note we do not actually need to require that $\widehat\eta_t\to\eta$, but we may wish this to be the case so as to obtain a smaller confidence sequence. In particular, in the case that $D_2(o; \eta') = d_2$ for a constant $d_2$, the width of the interval may be decomposed as the product of a factor that doesn't depend on the nuisance, and of $(\Gamma_t / t)^{-1} \to d_2^{-1} \sqrt{\mathrm{Var}(D_1(O_1; \eta_1))}$, with $\eta_1$ the limit of $\widehat{\eta}_t$ in an appropriate sense. In various situations, the latter quantity is minimized at $\eta_1 = \eta_0$. For example, it is the case in \cref{ex:abtest} that the width of the confidence intervals is minimized at the true regression function $\eta: (a,l) \mapsto E[U_1|A_1 = a, L_1 = l_1]$.

\subsection{Guarantees.}
We now verify our assumptions for this simple setting based on simple sufficient conditions. In the following, let $\sigma^2(\eta', \theta')=E[D(O_1;\eta', \theta')^2]-E[D(O_1;\eta', \theta')]^2$. As we make explicit next, it only takes two relatively mild assumptions in addition to \cref{asm:robustness} for Assumptions \ref{asm:sip}-\ref{asm:Evsmin1_over_s_summable} to hold in our estimating equations setting. In fact, for \cref{asm:sip}, and therefore the type-I error guarantee to hold, it only takes the following moment condition and variance lower bound condition on the estimating function.
\begin{assumption}[Estimating function moment condition]\label{asm:moment_estimating_function}
It holds that $$\sup_{\theta' \in \Theta, \eta' \in \mathcal{T}} E[D(O_1; \theta', \eta')^6] < \infty.$$
\end{assumption}
In what follows, denote $\sigma^2(\eta',\theta') = E[D(O_1;\eta',\theta')^2] - E[D(O_1;\eta',\theta')]^2$. 
\begin{assumption}\label{asm:variance_lower_bound}
It holds that $\inf_{\theta' \in \Theta, \eta' \in \mathcal{T}} \sigma^2(\eta',\theta') > 0$.
\end{assumption}
We can now state the type-I error result.
\begin{proposition}[Type-I error estimating equations]\label{prop:typeI_err_esteq}
Let $\alpha \in (0,1)$. Suppose that \cref{asm:moment_estimating_function} and \cref{asm:variance_lower_bound} hold. Let $\delta, \iota, \kappa \in (0,1)$ and $\nu \in (0,4]$ be such that (i) $(1-\delta)/(3-\delta) - 2 \iota > \kappa/ 2$ and (ii) $(1+\nu/2)(1-\kappa) - (2 +\nu) \iota > 1$. Then
\cref{asm:sip} holds, and therefore, if $\theta = \theta_0$, it holds, for $i=1,2,$ that 
$\lim_{t_0 \to \infty} P \left[ \tau_i < \infty \right] = \alpha$,
which, in terms of confidence sequences, is equivalent to 
$\lim_{t_0 \to \infty} \allowbreak P \left[ \forall t \geq t_0 : S_t \in \mathcal{C}_{\alpha, t_0}(t) \right] = 1-\alpha.$
\end{proposition}
The expected rejection time results take one more assumption, which we now state.
\begin{assumption}[Nuisance estimate convergence]\label{asm:nuisance_convergence}
There exists $\delta, \nu > 0$, and $\eta_1 \in \mathcal{T}$ such that $$\left(t^\nu E\left[ (\sigma(\widehat{\eta}_t, \theta_0) - \sigma(\eta_1, \theta_0))^{2 + \delta} \right] \right)_{t \geq 1}$$ is summable.
\end{assumption}

The condition in \cref{asm:nuisance_convergence} formalizes that $\widehat\eta_t$ admits a limit and characterizes the rate of convergence. Usually, this condition would be obtained from a similar bound on $\|\hat\eta_t-\eta_1\|_{\mathcal T}$ in some norm (\eg, Euclidean norm for a vector of nuisances or $L_p$ for nuisance functions) and establishing (or, assuming) that $\sigma(\cdot,\theta_0)$ is Lipschitz (or, H\"older) continuous in this norm. Such guarantees on $\widehat\eta_t$ can be obtained when it is estimated by maximum likelihood or more generally empirical risk minimization \citep{van2000empirical}. Generally, if (i) $\eta$ is a vector of pathwise differentiable parameters of the distribution of $O_1$, and (ii) the efficient influence function of $\eta$ admits a moment of order $(2+\delta')$, $\delta' > \delta$, then \cref{asm:nuisance_convergence} will hold. If $\eta$ is the best predictor of some $g_1(O_1)$ as a function of $g_2(O_1)$ in $\mathcal T$, then we can generally obtain guarantees in terms of the rate of the critical radii of $\mathcal T$ \citep{wainwright2019high}. 

We can now state the expected rejection time result.

\begin{proposition}\label{prop:expected_rejection_time}
Suppose that Assumptions \ref{asm:robustness}, \ref{asm:moment_estimating_function}, and \ref{asm:nuisance_convergence} hold. Then Assumptions \ref{asm:L1_conv_psi}-\ref{asm:Evsmin1_over_s_summable} hold and therefore, denoting $\psi = \sigma^{-1}(\eta_1, \theta_0) \mu(\theta_0)$, we have, for $i=1,2$, that
\begin{align}
E[\tau_i] \leq  2 \psi^{-2} \log \psi^{-1} + o\left(\psi^{-2} \log \psi^{-1} \right)
\end{align}
as $\theta_0 \to 0$, $t_0 \to \infty$, $\psi \sqrt{t_0} \to 0$, $\log \alpha^{-1} / \log \psi^{-1} \to 0$, and
\begin{align}
E[\tau_i] \leq  2 \psi^{-2} \log \psi^{-1} + o\left(\psi^{-2} \log \psi^{-1} \right)
\end{align}
as $\theta_0 \to 0$, $t_0 \to \infty$, $\psi \sqrt{t_0} \to 0$, $\log \alpha^{-1} / \log \psi^{-1} \to \infty$.
\end{proposition}

\subsection{Extension to other settings.} Here we considered just a simple setting with i.i.d. data and a nuisance-invariant estimating equation in order to show how one would verify our assumptions. Our results do also apply to more intricate settings, but further analysis would be needed to verify the assumptions using simple conditions. One example of a possible extension to the simple setting herein where our results still apply is where the estimating equation $D(\cdot ; \eta', \theta')$ is not completely invariant to $\eta'$, but instead we only have Neyman orthogonality \citep{chernozukhov2018double} in that $\lim_{\epsilon\to0} \epsilon^{-1} E[D(O_1;\eta+\epsilon(\eta'-\eta), \theta')]=0$ for $\eta'\in\mathcal T$. This is, for example, relevant to sequentially observing data from an observational study, where we do not know the propensity score. For example, we may have a sequential trial involving only the intervention of interest (\eg, experimental drug or surgery) and we have an offline pool of controls to compare to, where we assume selection into our trial at random given observed covariates (observed in the trial and in the pool of controls). Another extension may be to parameters that are path differentiable (that is, an influence function for them exists), but may not necessarily be defined in terms of an estimating equation. Yet another example of a possible extension to the simple setting herein where our results still apply is where the data is not i.i.d. but coming instead from an adaptive experiment such as a contextual bandit. In this case, many of our assumptions could be verified in a very similar manner to how theorems 1 and 2 of \cite{bibaut2021post} are proven, and guarantees for $\widehat\eta_t$ in the form of \cref{asm:nuisance_convergence} can be obtained from \cite{bibaut2021risk}. These would all simply be applications of our theory.

\section{Tuning $\lambda$ in burn-in nmSPRT}\label{sec:tuning_lambda}

As the reader might have noticed, we have so far left out the question of tuning $\lambda$ in the delayed-start nmSPRT. We address this question in the present section. To the best of our knowledge, this question hasn't been treated rigorously even in the case of the standard (that is, non-delayed-start) normal-mixture SPRT under normal observations. Our approach in the current section is to first, in the case of normal observations and without burn-in period, aim to uncover experimentally an identity for the optimal value of $\lambda$ as a function of the effect size $\psi$ and the nominal significance level $\alpha$ (\cref{sec:empirical_lambda_calibration}), and then to heuristically justify it mathematically (\cref{sec:heuristic_mathematical_justification_lambda_opt}). We then propose a heuristic design for a delayed-start sequential test that auto-tunes $\lambda$ as observations are collected.

\subsection{Calibrating $\lambda$ empirically}\label{sec:empirical_lambda_calibration}

For the purpose of experimentally calibrating $\lambda$ as a function of $\alpha$ and the effect size, we consider a Brownian observation sequence with drift $\mu$, that is we consider $S_t = \psi t + W(t)$, $\psi \neq 0$, with $W$ a standard Wiener process. We consider various values of $\psi$ and $\alpha$ and scan through values of $\lambda$. We plot the median, first, and third quartile of the rejection time of the burn-in nmSPRT on the left plot of \cref{fig:empirical_lambda_calibration}, and we find and represent the optimal value in the $\lambda$ grid for each couple $(\psi, \alpha)$. We work without a burn-in period, that is we set $t_0 = 1$. For the sake of efficiency comparison, also represent the rejection time of the rmlSPRT (which doesn't depend on $\lambda)$ and the oracle simple-vs-simple SPRT where the numerator corresponds to the (a priori unknown to the analyst) true value of $\psi$.

\begin{figure}
    \centering
    \subfloat[Median, first and third quartiles of burn-in nmSPRT rejection time evaluated at a grid of $(\alpha, \psi, \lambda)$ values. We run 20000 trajectories of $(S_t)_{t \in \{t_1,\ldots,t_{\max}\}}$ for each triple $(\alpha, \psi, \lambda)$. Row labels on the right represent $\alpha$ values, column labels at the top represent $\mu$ values. The vertical black line represents the optimal value of $\lambda$ for each couple $(\alpha, \psi)$]{{\includegraphics[width=6.5cm]{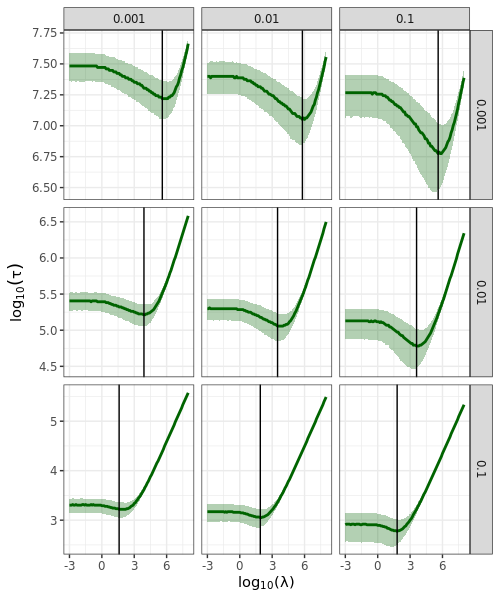} }}
    \qquad 
    \subfloat[Optimal $\lambda$ for each $(\alpha, \psi)$. We obtain each point by simulating 20000 trajectories of $W_t$ for $t \in \{\left\lfloor 10^i \right\rfloor : i \in \{0, 1/1000 \cdot 10^{12}, 2 /1000  \cdot 10^{12},\ldots, 10^{12}\}\}$ and performing a grid search over values of $\lambda$ in $\{10^{-3}, 10^{-3 + 0.1}, 10^{-3 + 0.2}, \ldots,  10^8 \}.$ The dashed line corresponds to $\lambda = \psi^{-2}$. The solid lines represent the fit of a local polynomial regression of $\log_{10}(\lambda_{opt})$ onto $\log_{10} \psi^{-2}$ obtained using the \texttt{loess} function of the \texttt{stats} R package.]{{\includegraphics[width=6.5cm]{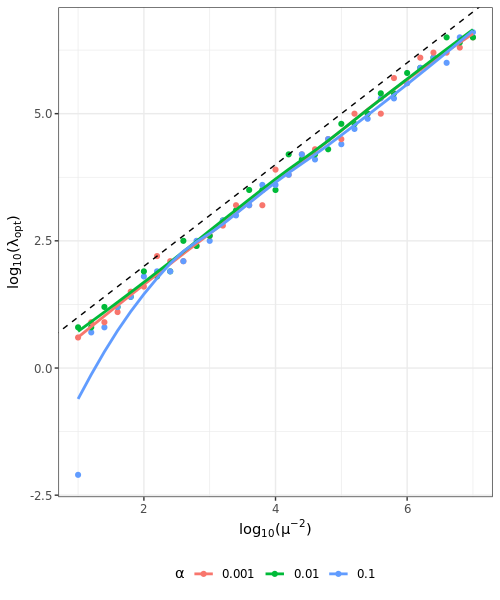} }} \label{fig:tau_function_of_lambda} \label{fig:lambda_opt_plot}
    \caption{Empirical calibration on Brownian data of $\lambda$ for the normal mixture SPRT without burn-in period.}\label{fig:empirical_lambda_calibration}
\end{figure}

We plot the results of the grid search for the optimal $\lambda$ for each $(\alpha, \mu)$ on the right plot of \cref{fig:empirical_lambda_calibration}. The plot seems to indicate, as might be expected from the prior-on-effect-size intuition, that the choice $\lambda = \mu^{-2}$ is close to optimal, at least for small values of $\alpha$. 

\subsection{A heuristic justification of the optimal choice of $\lambda$}\label{sec:heuristic_mathematical_justification_lambda_opt}

Under a Brownian observation sequence $\widetilde{S}(t) = \psi t + W(t)$, the rejection time $\tau_2$ of the nmSPRT (without burn-in, which shouldn't matter as the discussion would be the essentially the same under $t_0$ such as $\psi \sqrt{t_0} \to 0$, as per our asymptotic regimes) of level $\alpha$ and parameter $\lambda$ satisfies 
\begin{align}
\frac{\left(\psi \tau_2 + W(\tau_2)\right)^2}{\tau_2 + \lambda} - \log \frac{\tau_2 + \lambda}{\lambda} = - 2 \log \alpha.
\end{align}
As $\alpha \to 0$, $\tau_2 \to \infty$, and therefore, $\psi \tau_2  + W(\tau_2) \sim \mu \tau_2$. Therefore, for small $\alpha$, $\tau_2$ satisfies approximately 
\begin{align}
\frac{\left(\psi \tau_2 \right)^2}{\tau_2 + \lambda} - \log \frac{\tau_2 + \lambda}{\lambda} = - 2 \log \alpha. \label{eq:heuristic_lambda_eq1}
\end{align}
Differentiating the above equation with respect to $\lambda$ and using that, at the optimum $\lambda^*$ we have that $(\partial \tau_2 / \partial \lambda)|_{\lambda = \lambda^*} = 0$, yields that
\begin{align}
- \frac{\psi^2 \tau_2^2}{(\tau_2 + \lambda^*)^2} - \frac{1}{\tau_2 + \lambda^*} + \frac{1}{\lambda^*} = 0, \qquad \text{that is} \qquad \frac{\tau_2+\lambda^*}{\tau_2} = \lambda^* \psi^2. \label{eq:heuristic_lambda_eq2}
\end{align}
Injecting the last equality in \eqref{eq:heuristic_lambda_eq1} yields that 
\begin{align}
\tau_2 \lambda^{-1} - \log (\tau_2 \lambda^{-1} + 1) = - 2 \log \alpha
\end{align}
which implies that $\tau_2 / \lambda$ must diverge to infinity as $\alpha \to 0$. Therefore, as $\alpha \to 0$, $(\tau_2 + \lambda) / \tau_2 \to 1$, and thus from \eqref{eq:heuristic_lambda_eq2}, we must have that that $\lambda^* / \psi^{-2} \to 1$. 

\subsection{A sequential test that adaptively tunes $\lambda$}

Given the discussion of the previous two sections, it seems natural to use a test statistic that uses an estimate $\widehat{\lambda}_t$ of $\psi^{-2}$ instead of a prespecified value of $\lambda$. In keeping with the design logic of non-anticipating running-estimate SPRT statistics, we propose the following test statistic:
\begin{align}
Y^{\adpt-\lambda}_t = \sum_{s=1}^t \widehat{\psi}_{\widehat{\lambda}_{s-1}, s-1} X_s - \frac{1}{2} \widehat{\psi}_{\widehat{\lambda}_{s-1}, s-1}^2, \qquad \text{where} \qquad \widehat{\lambda}_t = \widehat{\psi}_{0,t}^{-2},
\end{align}
and $\widehat{\psi}_{\lambda,t} = S_t / (t+\lambda)$, for any $t + \lambda >0$, and $\widehat{\psi}_{0,0} = 0$, as defined earlier. So as to fully specify a test of level $\alpha$, it remains to determine the rejection threshold. We speculate that by enforcing a long enough burn-in period, $Y^{\adpt-\lambda}_t$ will behave similarly to an analogous test statistic \begin{align}
\widetilde{Y}^{\adpt-\lambda}_t = \sum_{s=1}^t \widetilde{\psi}_{\widetilde{\lambda}_{s-1}, s-1} \widetilde{X}_s - \frac{1}{2} \widetilde{\psi}_{\widetilde{\lambda}_{s-1}, s-1}^2
\end{align}
obtained from a sequence $(\widetilde{X}_t)_{t \geq 1}$ of standard normal i.i.d. observations, and therefore, that the rejection threshold for the latter should yield approximately the same type-I error for the former. That is, for a burn-in period of duration $t_0$, we are looking for $-\log \alpha^{\adpt-\lambda}(\alpha, t_0)$ such that
\begin{align}
P \left[ \sup_{t \geq t_0} \widetilde{Y}^{\adpt-\lambda}_t \geq -\log \alpha^{\adpt-\lambda}(\alpha, t_0) \right] = \alpha, 
\end{align}
that is
\begin{align}
E \left[ P \left[ \sup_{t \geq t_0} \left( \widetilde{Y}^{\adpt-\lambda}_t - \widetilde{Y}^{\adpt-\lambda}_{t_0} \right) \geq -\log \alpha^{\adpt-\lambda}(\alpha, t_0) - \widetilde{Y}^{\adpt-\lambda}_{t_0}\mid \widetilde{\mathcal{F}}_{t_0}  \right] \right] = \alpha.
\end{align}
Observing that $(\exp(\widetilde{Y}^{\adpt-\lambda}_t - \widetilde{Y}^{\adpt-\lambda}_{t_0})_{ t \geq t_0}$ is a martingale with initial value 1, we should have, from Ville's inequality case of equality, that the above equation is approximately equivalent, for $t_0$ large enough, to 
\begin{align}
E \left[ \exp \left( -\left( -\log \alpha^{\adpt-\lambda}(\alpha, t_0)- \widetilde{Y}^{\adpt-\lambda}_{t_0} \right)_+ \right) \right] = \alpha.
\end{align}
We propose to solve the above equation by performing a grid search over candidate values of $-\log \alpha^{\adpt-\lambda}(\alpha, t_0)- \widetilde{Y}^{\adpt-\lambda}_{t_0}$ and evaluating the above expectation at the grid points by Monte-Carlo simulations. (Note that the Monte-Carlo simulation entails drawing multiple trajectories of an i.i.d. sequence $(\widetilde{X}_t)_{t \geq t_0}$ of standard normal random variables, as opposed to drawing multiple trajectories of the data sequence, which of course is not possible). Our proposed procedure is then the one that starts monitoring the test statistic $Y^{\adpt-\lambda}_t$ after a burn-in period of length $t_0$ and rejects the null hypothesis as soon as it crosses the rejection threshold $-\log \alpha^{\adpt-\lambda}(\alpha, t_0)$ after $t_0$. 

We do not analyze formally this procedure in the current version of this work, but we do evaluate it empirically alongside the previously discussed delayed-start nmSPRT and rmlSPRT in the next section.

\section{Numerical experiments}\label{sec:numerics}

We now confirm experimentally the type-I error and expected rejection time guarantees for delayed-start rmlSPRT and nmSPRT confidence sequences, and we compare them to alternative confidence sequences.

\subsection{Type-I error} We start with type-I error experiments.  We consider a sequence of i.i.d. observations $O_1,O_2,\ldots \sim \mathrm{Bernoulli}(0.03) - 0.03$. The choice of the relatively small value $0.03$ is to ensure that the sum  $S_t$ of the stabilized martingale difference $X_1,X_2,\ldots$ sequence doesn't converge too fast to a normal. It is of course impossible to evaluate the event $\{\tau < \infty \}$ for any sequential test $\tau$, but, since the martingales we consider must follow the law of the iterated logarithm and the sequences we study have $\sqrt{t \log t}$ asymptotics, rejections should happen early on. We evaluate boundary crossings at the points of a time grid $t_1,t_2,\ldots, t_N$ such that $t_i \approx t_{\max} \sum_{j=1}^i j^\beta / \sum_{j=1}^N j^\beta$, with $N=20000$, $\beta = 2$ and $t_{\max} = 10^8$. This specification ensures that the time grid is denser early on the time axis. We refer the reader to our code for implementation details.

\begin{figure}%
    \centering
    \subfloat[Type-I error of the rmle-SPRT sequence, binmSPRT sequence with $\lambda = 1$ and binmSPRT sequence with $\lambda = 100$, as a function of the burn-in time $t_0$, at levels $\alpha \in \{5\cdot 10^{-3}, 1\cdot 10^{-2}, 5 \cdot 10^{-2}, 1 \cdot 10^{-1} \}$. We simulate 5000 trajectories of $S_t$ per point.]{{\includegraphics[width=6.5cm]{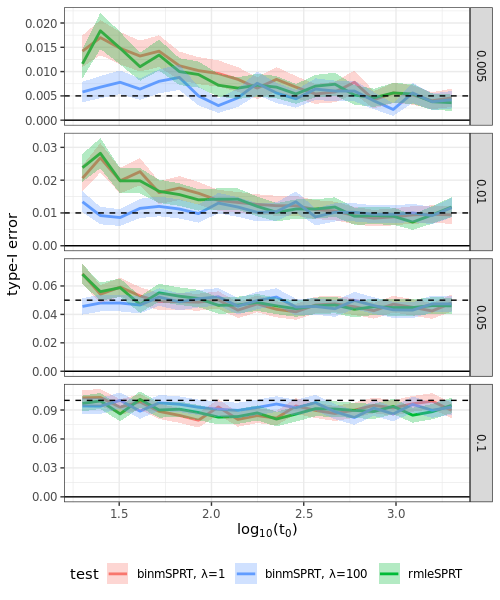} }}%
    \qquad 
    \subfloat[Comparison of type-I error of sequences with burn-in with sequences without burn-in period, namely \cite{howard2021}'s finite sample LIL, given by $1.7 \sqrt{t ( \log \log t + 0.72 \log (10.4 / \alpha))}$, and the nmSPRT sequences $\sqrt{(t+\lambda)(2\log \alpha^{-1} + \log ((t+\lambda)/\lambda))}$, with $\lambda=1$ and $\lambda = 100$. Chosen significance level is $\alpha = 5 \cdot 10^{-2}$. Number of simulated trajectories per point is 100000 for sequences with burn-in and 200000 for sequences without burn-in.]{{\includegraphics[width=6.5cm]{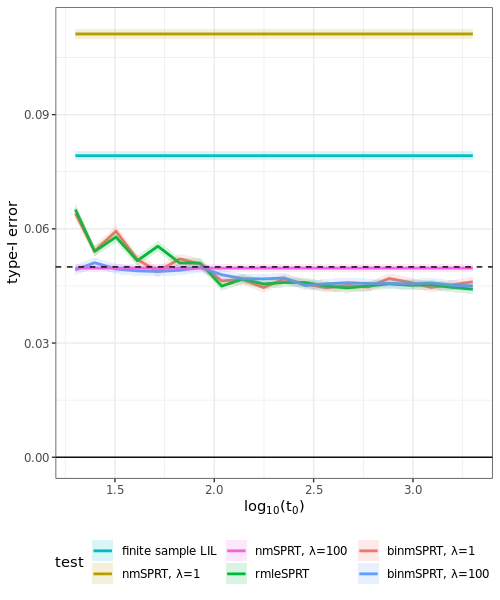} }}%
    \caption{Type-I error as a function of $t_0$. The confidence bands are 95\% pointwise Wald-type confidence sequences. We obtain each point in the above two plots by simulating trajectories of $(S_t, t \in \{t_1,\ldots,t_{\max}\})$. The dashed black line represents the nominal $\alpha$ level, of which the value is indicated on the right side of the plots.}
    \label{fig:type-I_error}%
\end{figure}
\cref{fig:type-I_error} shows convergence of the type-I error with the burn-in period, while the confidence sequences without burn-in period over reject, likely due to erroneous early rejections at time points when $S_t$ is still far from its Wiener process approximations. Note also that for $\lambda = 100$, we don't observe over-rejections even for small values of $t_0$. This is to be expected as setting $\lambda$ trades off early rejections for later tightness. It can be directly observed from the expression of the nmSPRT confidence sequence that the width of the sequence increases early as $\lambda$ increases.  is also to be expected from the interpretation of $\lambda^{-1/2}$ as a prior on the effect size: if the effect size is of order $\lambda^{-1/2}$ a confidence sequence optimally targeting that effect size should be loose for $t \ll \lambda$ and tight around $\lambda$, thereby preventing early rejections if $\lambda$ is large.

\subsection{Rejection time}

We now turn to evaluating features of the distribution of the rejection time. As the reader might have noticed, we haven't discussed in detail so far the choice of the parameter $\lambda$ in the nmSPRT expression, beyond the interpretation of $\lambda^{-1/2}$ as an a priori belief on the magnitude of the effect size. We investigate empirically the effect of $\lambda$ on the rejection time of the nmSPRT in the next subsection.

We now plot several measures of sample efficiency for the burn-in nmSPRT sequences and the burn-in rmlSPRT sequences. In particular, we plot the ratio of the median stopping time of our tests over the median stopping time of the oracle (in the sense that it uses the true value of $\mu$) simple-vs-simple SPRT with same burn-in period. (We compute the adjusted $\alpha$ level for the delayed-start simple-vs-simple SPRT in a similar fashion to that of the delayed-start nmSPRT and rmlSPRT). We refer to this ratio as the ``relative efficiency'' of the sequential tests.

\begin{figure}\label{fig:relative_efficiency}
    \centering
    \subfloat[Relative efficiency as a function of $\alpha$ for various values of $\mu$. The panel labels at the top represent the values of $\alpha$.]{{ \includegraphics[width=6.5cm]{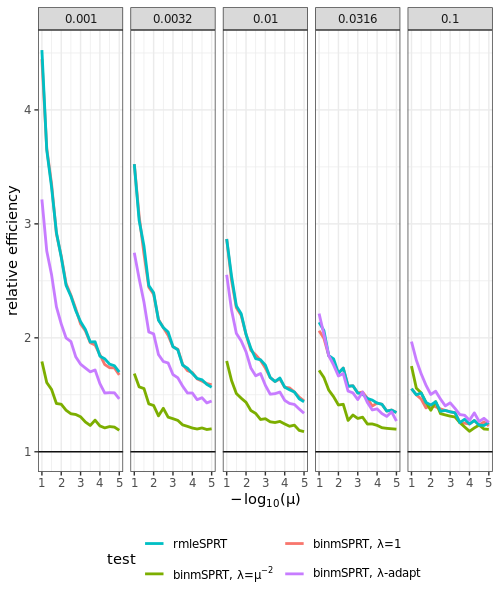} }}%
    \qquad 
	\subfloat[Relative efficiency as a function of $\mu$, for various values of $\alpha$. The panel labels at the top represent the values of $\alpha$]{{ \includegraphics[width=6.5cm]{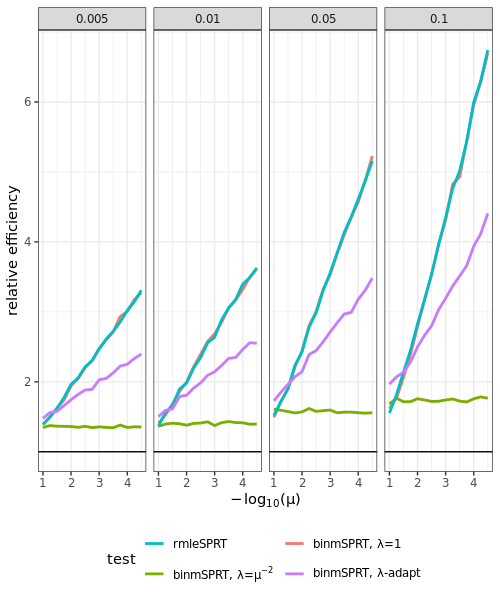} }}
\caption{Relative efficiency, defined as $\mathrm{median}(\tau) / \mathrm{median}(\tau^{\mathrm{svs}})$. We obtain each point by simulating 8000 trajectories of $(S_t)_{t \in \{t_1,\ldots,t_{\max}\} }$, for i.i.d. observations $O_1,O_2\ldots, \sim \mu + (\mathrm{Bernoulli}(0.03)-0.03)$.}
\end{figure}

We observe on the left subplot of \cref{fig:relative_efficiency} that the relative efficiency of the (delayed-start) rmlSPRT and the nmSPRT seem to converge to 1 as $\alpha \to 0$, as implied by \cref{thm:rejection_time_upper_bounds}. An empirical (as opposed to predicted by any of the theorems of the current article) finding we infer from the right plot of \cref{fig:relative_efficiency} is that the median stopping time of the nmSPRT at the optimal $\lambda$ value seems to be within a constant factor $c(\alpha) > 1$ of the median stopping time of the oracle simple-vs-simple SPRT, and that $c(\alpha) \to 1$ as $\alpha \to 0$.

\begin{figure}\label{fig:tau_plot}
\centering
\includegraphics[width=6.5cm]{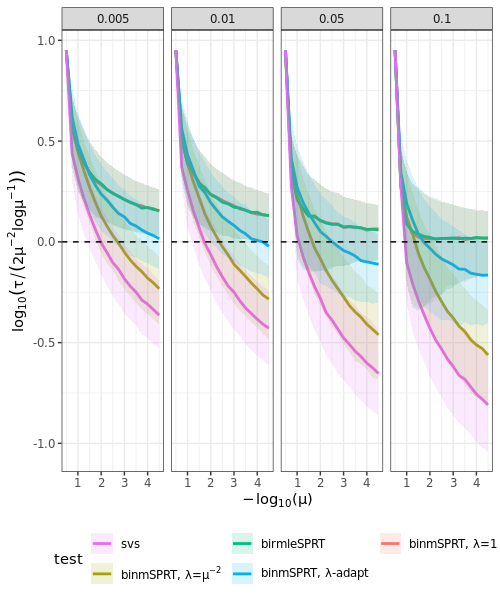}
    \caption{Median, first and third quartile of the ratio of the rejection time with the asymptotic upper bound $2 \mu^{-2} \log \mu^{-1}$ from \cref{thm:rejection_time_upper_bounds}, for the burn-in oracle simple-vs-simple SPRT, burn-in nmSPRT, burn-in rmlSPRT, and burn-in nmSPRT with adaptive $\lambda$. We obtain each point by simulating 8000 trajectories of $(S_t)_{t \in \{t_1,\ldots,t_{\max}\} }$, for i.i.d. observations $O_1,O_2\ldots \sim \mu + (\mathrm{Bernoulli}(0.03)-0.03)$. We set $\alpha = 5 \cdot 10^{-2}$ for the current plot. We set the burn-in period $t_0$ to 200. The panel labels at the top represent the value of $\alpha$.}
    \label{fig:tau_function_of_mu}
\end{figure}

\section{Application to quality-control of Netflix client application}\label{sec:realdata}

We illustrate the use of the delayed-start rmlSPRT and of the delayed-start nmSPRT on a real data case study. We also compare it to an empirical Bernstein sequential boundary based on a gamma-exponential mixture boundary, as proposed in section 4 of \citep{howard2021}.

Our data comes from an A/B test used to quality-control the release of a new Netflix client application. Users in the treatment and control groups experienced new and existing versions of the Netflix software respectively. The outcome of interest was the delay between requesting playback and the stream starting, which we term ``play delay''. Although this dataset was collected according to a pre-determined sample size, we use it to illustrate the performance of our sequential tests. In the following, pre-treatment measurements of play delay are available for all experimental units, which are leveraged for regression adjustment. We work with log-transformed play delay, and use an undisclosed base to protect proprietary information.

We simulate observation trajectories $(X_t)_{t \geq 1}$ as follows: for each $t$, we draw uniformly at random treatment assignment $A_t \in \{0,1\}$ and then we draw outcome $U_t$ and a pre-treatment covariate $L_t$ by sampling with replacement from actual realized pairs $(L, U)$ in arm $A_t$ in the test data. 
 The data structure is an instance of a Bernoulli trial with covariates as described in the second example at the beginning of the article. The outcome $U_t$ is, as we mentioned, the play delay. We use as covariate $L_t$ the pre-treatment outcome of unit $t$, that is the play-delay before treatment assignment. The probability of assignment to cell 1 is $p=0.5$. Our null hypothesis is that there is no difference in play delay on average between the newer and existing Netflix client, that is $H_0: \theta = 0$, where $\theta = E[U_1 \mid A_1  =1] - E[U_1 \mid A_1  =0]$ identifies this mean causal effect. 
As in our Bernoulli trial with covariates example, we let 
\begin{align}
Z_t = \widehat{\eta}_{t-1}(1, L_t) - \widehat{\eta}_{t-1}(0, L_t) + \frac{A_t - p}{p(1-p)}(U_t - \widehat{\eta}_{t-1}(A_t, L_t)),
\end{align}
where $\eta_{t-1}(a, l)$ is a least-squares estimator of the linear regression of $U$ on $L$, $L \times A$ and an intercept, computed from observations indexed by $\mathcal{I}_{1,t}$, where the index set $\mathcal{I}_{1,t}$ is as introduced in section \ref{sec:application}. We set $X_t = \widehat{\sigma}_{t-1}^{-1} Z_t$ where $\widehat{\sigma}_{t-1}$ is a $\mathcal{D}_{0,t}$-measurable estimator of the standard deviation of $Z_t$.

We set the burn-in period $t_0$ to 1000 observations and the normal mixture SPRT tuning parameter $\lambda$ to $10^2$. Measurements of play delay are bounded from above by $b$, as longer delays are simply abandoned. When using the empirical Bernstein boundary \citep{howard2021}, $\eta_{t-1}(a, l)$ is the same linear predictor projected onto the interval $[0, b]$, ensuring $Z_t \in [-2b, 2b]$. In their notation, we use a scale $c = 4b$, and a gamma-exponential mixture with parameter $\rho$. Following the recommendation formulated in their section 3.5, we set the gamma-exponential mixture parameter $\rho$ to the same as the Gaussian mixture parameter $\lambda=100$. 
We make these choices based on domain knowledge of these types of experiments. They correspond to a target effect size of the order of $1$\%, and to the fact that we typically get a hundredfold gain in sample size from pre-allocation outcome adjustment. While in the context of our simulation we of course know the effect size (it turns out the point estimate is -0.81\%) since we work from an already fully collected data set, the  $1$\% order of magnitude is what engineers at Netflix expect using domain specific expertise.

So as to make things more concrete, we plot in figure \ref{fig:play_delay_AB-one_data_trajectory} the trajectory of the test statistics for one arbitrary random draw of the data sequence. We see that if we had received the data in the particular order of the simulated stream, we would have called the test after collecting slightly after the end of the burn-in period $t_0 = 1000$ while the non-asymptotic empirical Bernstein test stops at almost $10^4$ observations. 

We now examine the behavior of the test statistics over many draws of the data sequence. We first examine the empirical Type-I error of the nmSPRT by simulating data under the null hypothesis. Specifically, we simulate data for both treatment and control by sampling with replacement from observed control dataset, which we refer to as a simulated A/A test. Figure \ref{fig:play_delay_AB}, left, shows that the Type-I error approaches the nominal 5\% the longer the simulation is performed. 

We then compare the behavior of the rmlSPRT, the nmSPRT and the non-asymptotic empirical Bernstein test on simulated trajectories of the A/B test (unlike in the A/A test, here observations for treatment and control units are drawn from their respective empirical distributions obtained from the observed dataset, following simulated treatment assignment $A_t$). We plot simulation results in figure \ref{fig:play_delay_AB}. In this example, the rmlSPRT and the nmSPRT perform relatively similary due to the fact that we chose $\lambda$ in the normal mixture to optimize for rejection times of the order of $10^2$, which is earlier than the end of the burn-in time. As expected, the empirical Bernstein test is much more conservative, and we observe it tends to reject a little less than an order of magnitude later.

\begin{figure}
\centering
\includegraphics[width=6.5cm]{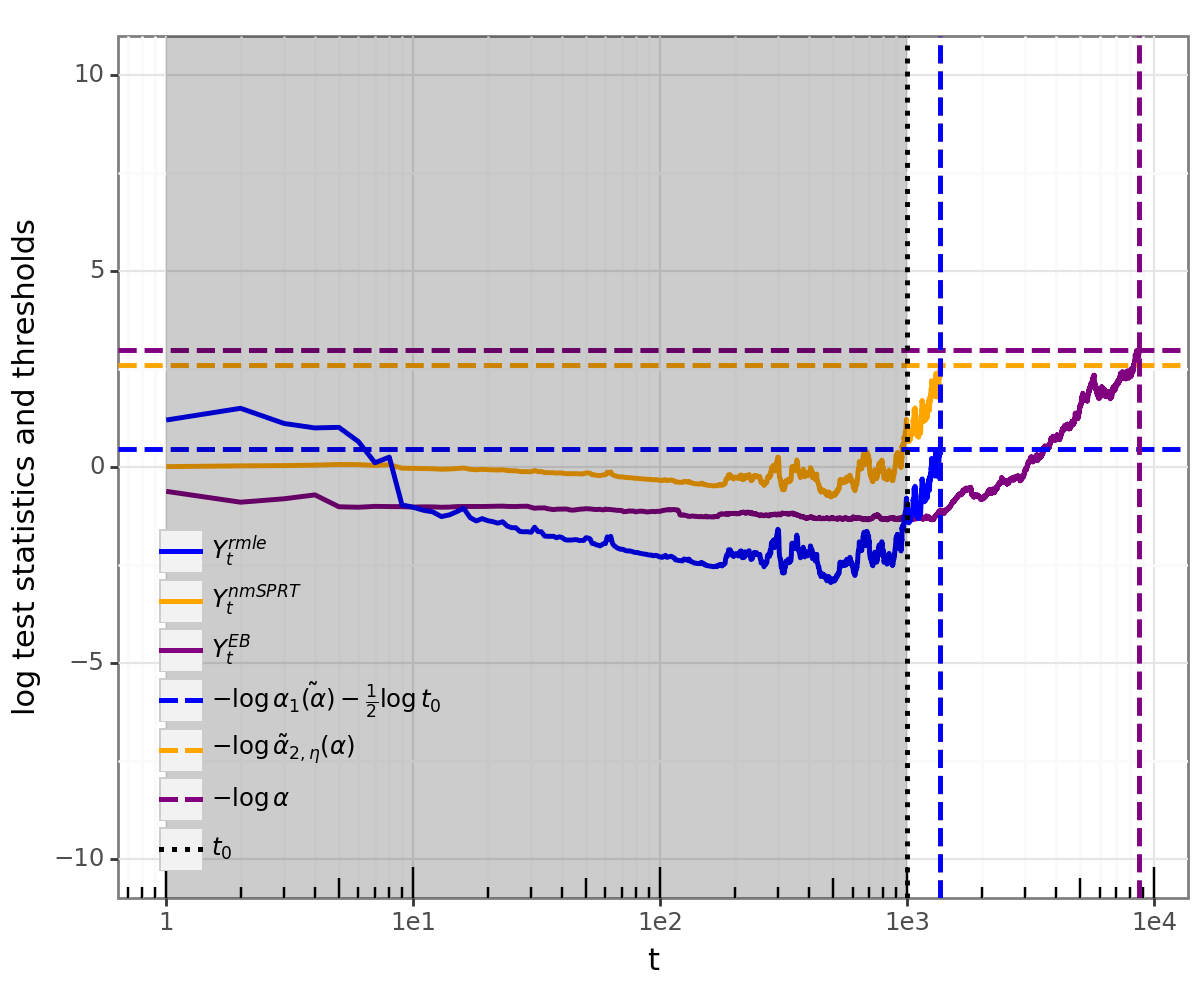}
    \caption{nmSPRT and rMLE test statistics trajectories on one arbitrary draw of the data sequence in the play delay case study.}
\end{figure}\label{fig:play_delay_AB-one_data_trajectory}

\begin{figure}
    \centering
    \subfloat[Empirical CDFs of the rejection times of the three sequential tests for 1000 simulated trajectories of the A/A experiment.]{{ \includegraphics[width=6.5cm]{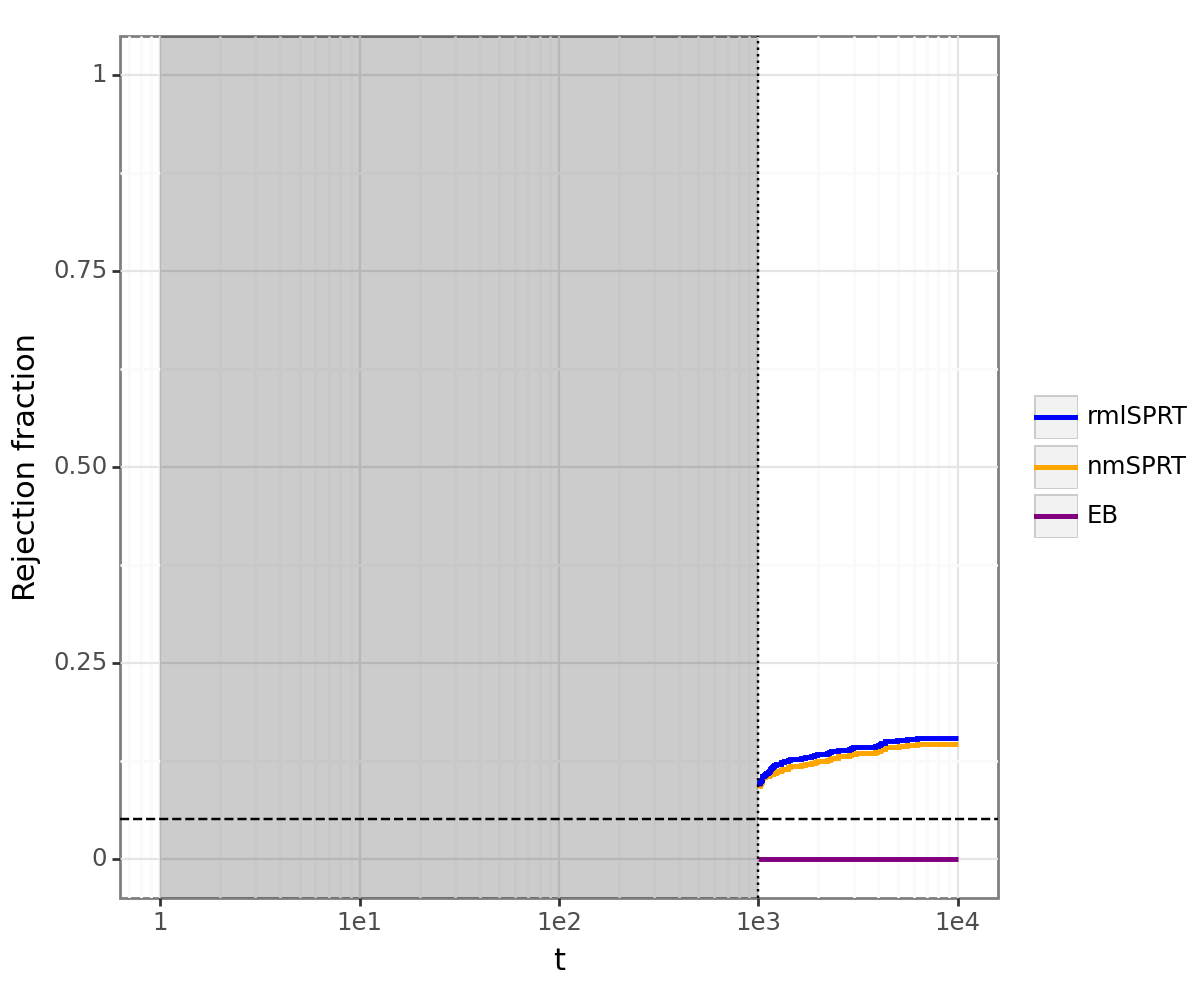} }}
    \qquad 
	\subfloat[Empirical CDFs of the rejection times of the three sequential tests for 1000 simulated trajectories of the A/B experiment.]{{ \includegraphics[width=6.5cm]{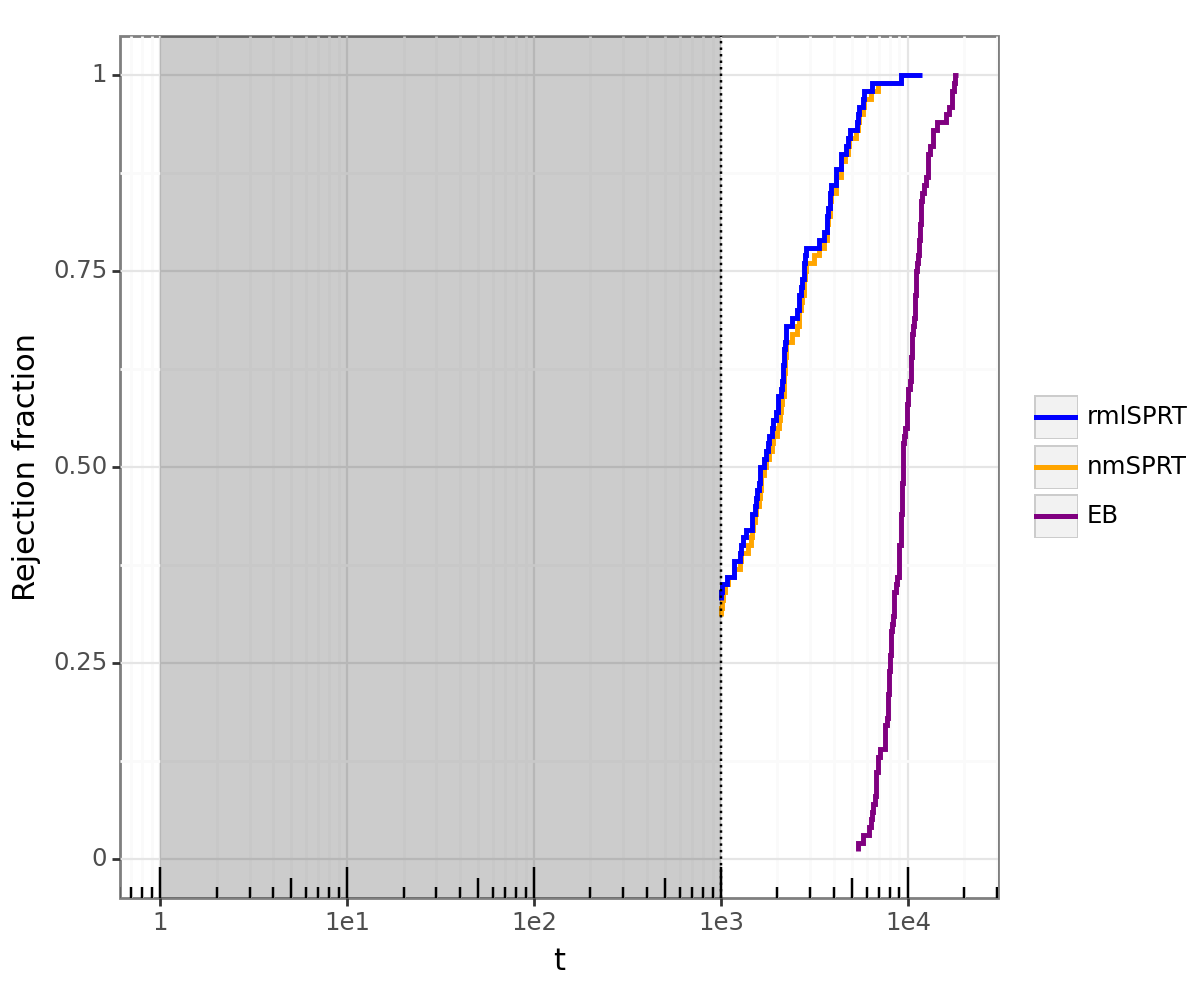} }}
\caption{Behavior of the rejection time of the mixture SPRT, running-maximum-likelihood SPRT, and the empirical Bernstein sequential test under simulated A/A experiments and simulated A/B experiments in the play delay case study.}
\end{figure}\label{fig:play_delay_AB}

\section{From Wald's simple-vs-simple SPRT to delayed-start running-mean-estimate SPRTs}\label{sec:historical_progression}

In this section, we expose the historical and logical progression from Wald's results \citep{wald, waldsequential, waldwolfowitz, waldwolfowitzbayes} on the optimality of simple-vs-simple parametric SPRTs to our results, that is the type-I error calibration and expected rejection time optimality of the non-parametric delayed-start running-mean-estimate.

\subsection{Type-I error and optimality of the oracle simple-vs-simple SPRT}

\paragraph{Definition and type-I error.} \cite{wald, waldsequential} introduced the sequential probability ratio test, defined as follows. Suppose that $X_1,X_2,\ldots$ are i.i.d. drawn from a common distribution with density $p$ w.r.t a certain measure. Consider two simple hypotheses $H_0 : p = p_0$ and $H_1 : p = p_1$, with $p_1 / p_0 < \infty$. The SPRT statistic at $t$ is the likelihood ratio
\begin{align}
Y^{\svs}_t = \prod_{s=1}^t \frac{p_1(X_t)}{p_0(X_t)},
\end{align}
and the $\alpha$-level SPRT is the stopping time $\tau^{\svs} = \inf \{t \geq 1 : Y^{\svs}_t  \geq -\log \alpha \}$. That type-I error is at most $\alpha$ follows, via Ville's inequality \citep{ville}, from the fact that $(Y_t^\svs)_{t \geq 1}$ is a martingale under $H_0$ with initial value 1.

\paragraph{Power, expected rejection time.} \cite{wald, waldsequential} shows that the SPRT has power 1 under the alternative $H_1$ and that its expected rejection time as $d_{KL}(p_1, p_0) \to 0$ is asymptotically equivalent to $- \log\alpha / d_{KL}(p_, p_0)$, where $d_{KL}$ is the Kullback-Leibler divergence. \cite{waldwolfowitz} further show that, under the present setting, the expected stopping time of the simple-vs-simple SPRT is optimal among all tests of level $\alpha$ under $H_0$, that is, for any sequential test $\tau'$ such that $\mathrm{Pr}_{H_0}[\tau' < \infty] \leq \alpha$, we must have $E_{H_1}[\tau'] \geq E_{H_1}[\tau^\svs]$.

\subsection{Mixture SPRTs for the Wiener process and for parametric sequence of i.i.d. random variables}

\paragraph{The case for mixture test statistics under composite alternatives.} Consider a family of densities $p_\psi$ indexed by a one-dimensional parameter $\psi$ and the null hypothesis $H_0 : \psi = \psi_0$. It is often the case that the alternative hypothesis is a composite alternative of the form $H_{\backslash 0} = \psi \neq \psi_0$. In clinical trials or A/B tests for instance, experimenters want to test the absence of average treatment effect (ATE) against the hypothesis that the ATE is non-zero, rather than against a specific non-zero value of the ATE. 

It can be shown that for $p \in H_{\backslash 0}$ such that $p$ is closer in Kullback-Leibler divergence to $H_0$ than it is to $H_1$, the simple-vs-simple SPRT of $H_0$ against $H_1$ has power strictly smaller than 1. It can further be shown that, even if $d_{\mathrm{KL}}(p, H_1) < d_{\mathrm{KL}}(p, H_0)$, the expected stopping time can get highly suboptimal as $\psi$ gets away from $\psi_1$.

A remedy to this limitation of the simple-vs-simple SPRT that ensures the resulting sequential test has power 1 against any $\psi_1 \neq \psi_0$, is to use a prior $F$ over the possible values of $\psi_1$. This yields so-called mixture SPRTs, introduced by \cite{robbins1970boundary, robbins1970statistical}, where $F$ is the so-called mixture distribution, in which the test statistic and the test are defined as 
\begin{align}
Y^{\mathrm{mixt}}_t = \int  \prod_{s=1}^t \frac{p_{\psi_1}(X_t)}{p_{\psi_0}(X_t)} dF(\psi_1) \qquad \text{and} \qquad \tau^{\mathrm{mixt}} = \inf \{ t \geq 1: Y^{\mathrm{mixt}}_t \geq - \log \alpha\}.
\end{align}
Integrating against $F$ preserves the martingale property under $H_0$ and therefore the type-I error guarantee. 

\paragraph{Mixture confidence sequences for normal data.} Under an i.i.d. sequence $X_1,X_2 \ldots \sim \mathcal{N}(\psi, 1)$, that is under $p_\psi(x) = (2 \pi)^{-1/2} \exp(-(x-\psi)^2/2)$, $Y_t^{\mathrm{mixt}}$ takes the form
\begin{align}
Y_t^{\mathrm{mixt}} = \int \exp\left((\psi_1-\psi_0) S_t - \frac{1}{2} (\psi_1 - \psi_0)^2 t\right) dF(\psi_1) \qquad \text{where} \qquad S_t = \sum_{s=1}^t X_s.
\end{align}
It is common to specify the null hypothesis by setting $\psi_0 = 0$ (think about testing the absence of average treatment effect against non-zero ATE). An analytically convenient mixture distribution $F$ is the normal distribution \citep{robbins1970boundary, robbins1970statistical} centered around 0 and with variance $\lambda^{-1}$. The variance of the mixture distribution encodes beliefs about the possible values of the effect size if it isn't zero. This yields 
\begin{align}
Y_t^{\mathrm{mixt}} = \sqrt{\frac{\lambda}{t+\lambda}} \exp\left(\frac{1}{2} \frac{S_t^2}{t+\lambda}\right).
\end{align}
Inverting $Y_t^{\mathrm{mixt}}$ yields a confidence sequence for $S_t$. Specifically, for any $\alpha > 0$, that $\Pr_{H_0}[\tau^{\mathrm{mixt}} < \infty] \leq \alpha$ is equivalent to the fact that
\begin{align}
|S_t| \leq \sqrt{(t+\lambda) \left( - 2 \log \alpha + \log\left( \frac{t + \lambda}{\lambda} \right)\right)} \qquad \forall t \geq 1.\label{eq:nmSPRT_boundary}
\end{align}

\paragraph{Exact calibration for Wiener processes.} \cite{robbins1970boundary} further show that for a standard Wiener process, the crossing probability of boundary \eqref{eq:nmSPRT_boundary} is exactly $\alpha$, while it is only known that it is at most $\alpha$ in the case of discrete normal i.i.d. data.

\subsection{Running-estimate SPRTs}

\paragraph{Definition of the running-estimate SPRTs.} 

A seemingly different strategy to modify the simple-vs-simple SPRT so as to obtain a test of power one against composite alternatives is to harness the sequentiality of data collection by using likelihood ratios of the form $Y^{\mathrm{rngest}}_t=\prod_{s=1}^t p_{\widehat{\psi}_{s-1}}(X_t)/p_{\psi_0}(X_t)$, where $\widehat{\psi}_{s-1}$ is a running $\mathcal{F}_{s-1}$-measurable estimate of $\psi$. This is the approach proposed by \cite{robbins1972class} to design one-sided tests of $H_{\leq \psi_0} : \psi \leq \psi_0$ against $H_{> \psi_0} : \psi > \psi_0$. A running estimate used by \cite{robbins1972class} in the case of a sequence of i.i.d data drawn from $\mathcal{N}(\psi,1)$ is the threhsolded maximum likelihood estimator (MLE) $\widehat{\psi}_t = t^{-1} (\sum_{s=1}^t X_s - \psi_0)_+ + \psi_0$, while another one is the posterior mean computed from $X_1,\ldots, X_t$, under a prior $F$ with support $[\psi_0, \infty)$.

\paragraph{Optimality.} It makes sense that the running-estimate SPRTs should have close to optimal expected rejection time, as $\widetilde{\psi}_{s-1}$ is a proxy for $\psi$, and we know from \cite{waldwolfowitz} that the level-$\alpha$-under-$H_0$ test with optimal rejection time under $\psi = \psi_1$ is the simple-vs-simple SPRT of $H_0$ against $H_1$.  \cite{robbins_siegmund1974} study the expected rejection time of one-sided tests of $H_{\leq \psi_0}$ against $H_{> \psi_0}$ of the form $\tau^{\mathrm{rngest}} = \inf \{ t \geq 1: Y^{\mathrm{rngest}}_t \geq -\log \alpha \}$ as $\psi \downarrow \psi_0$. They prove that when $\widehat{\psi}_t$ is taken to be the thresholded MLE, $E[\tau^{\mathrm{rngest}}] \sim P_{H_0}[ \tau^{\mathrm{rngest}} = \infty ](\psi - \psi_0)^{-2} \log (\psi - \psi_0)^{-1}$. This isn't too far off the lower bound $2 P_{H_0}[ \tau^{\mathrm{rngest}} = \infty ](\psi - \psi_0)^{-2} \log \log (\psi - \psi_0)^{-1}$ proven by \cite{farrell1964asymptotic} for such one-sided tests as $\psi \downarrow \psi_0$. 

\paragraph{Connection to mixture SPRTs and confidence sequences.} \cite{robbins_siegmund1974} show that the running-estimate SPRTs exhibit an interesting connection to mixture SPRTs in the case where the data stream is a time-continuous process of the form $(\widetilde{S}(t))_{t \geq 0}$,  with $\widetilde{S}(t) = \psi t + W(t)$ for all $t \geq 0$, where $(W(t))_{t \geq 0}$ is a standard Wiener process. We expose their observations here. Denote $\widetilde{\mathfrak{F}} = (\widetilde{\mathcal{F}}(t))_{t \geq 0}$ the canonical filtration to which $W$ is adapted. The Brownian continuous-time analog of the class of test statistics of the form $Y_t^{\mathrm{rngest}}$ is the class of test statistics of the form:
\begin{align}
\widetilde{Y}^{\mathrm{rngest}}_t = \exp\left(\int_0^t \widetilde{\psi}(s) d \widetilde{S}(s) - \frac{1}{2} \widetilde{\psi}(s)^2 ds \right),
\end{align}
where $\widetilde{\psi}(s)$ is an $\widetilde{\mathcal{F}}_s$-measurable estimate of $\psi$. Meanwhile, Brownian continuous-time mixture test statistics take the form
\begin{align}
\widetilde{Y}_t^{\mathrm{mixt}} = f(\widetilde{S}_t, t) \qquad \text{where} \qquad f(x,t) = \int \exp\left (\psi' x - \frac{1}{2} \psi'^2 t \right) dF(\psi'),
\end{align}
where $F$ is the mixing distribution. It\^o's lemma asserts that for any stochastic process $(X(t))_{t \geq 0}$ of the form $X(t) = X(0) + \int_0^t \mu(s) ds + \sigma(s) dW(s)$, where $\sigma(s)$ and $\mu(s)$ are $\widetilde{\mathcal{F}}(s)$-measurable, and any suitably differentiable $(x,t) \mapsto u(x,t)$, it holds that $u(X_t,t) = u(0,0) + \int_0^t \partial_x u(X_s,s)dX(s)+ (\partial_t u(X_s,s) + \frac{1}{2} \partial_{x,x} u(X_s,s))ds$.
Applying It\^o's lemma to $u(\widetilde{S}(t),t) = \log f(\widetilde{S}(t),t)$ yields
\begin{align}
\widetilde{Y}^\mathrm{mixt}_t = \exp \left( \int_0^t \widetilde{\psi}'(s) d\widetilde{S}(s) - \frac{1}{2} \left(\widetilde{\psi}'(s) \right)^2  ds \right) \qquad \text{where} \qquad \widetilde{\psi}'(s)  = \frac{\partial_x f}{f}(\widetilde{S}(s),s).
\end{align}
Notice that 
\begin{align}
\widetilde{\psi}'(s) =  \int \psi' \frac{\exp(\psi' \widetilde{S}(s) - \frac{1}{2} \psi'^2 t) dF(\psi')}{\int \exp(\psi' \widetilde{S}(s) - \frac{1}{2} \psi'^2 t) dF(\psi')}
\end{align}
is the posterior mean of $\psi$ given $\widetilde{\mathcal{F}}(s)$ under prior $F$. For $F= \mathcal{N}(0, \lambda^{-1})$, we have that $\widetilde{\psi}'(s) = \widetilde{S}(s) / (s + \lambda)$, that is the shrunken empirical mean with shrinkage parameter $\lambda$. This shows in particular that the normal mixture boundary \eqref{eq:nmSPRT_boundary} is equivalent to $\widetilde{Y}^{\mathrm{rngest}}_t \leq - \log \alpha$ with $\widetilde{\psi}(s) = \widetilde{\psi}'(s) = \widetilde{S}(s) / (s + \lambda)$. \cite{robbins_siegmund1974} show in their Theorem 3 that in the normal discrete case, the half-normal prior yields a running-posterior-mean SPRT with expected rejection time behaving as $2 P_{H_0}[\tau^{\mathrm{rngest}} = \infty] (\psi - \psi_0)^{-2} \log (\psi - \psi_0)^{-1}$.

\paragraph{Boundary associated to the running MLE SPRT} As far as we are aware, there doesn't seem to be a mixture distribution corresponding to the MLE or to the thresholded MLE alluded to earlier. However, application of It\^o's lemma to $h(x,t) = x^2 / (2t)$ yields that the running MLE SPRT log test statistic started at 1 can be rewritten as follows:
\begin{align}
\log Y^{\rmlSPRT}_t =& \int_1^t \frac{\widetilde{S}(s)}{s} d \widetilde{S}(s) - \frac{1}{2} \left( \frac{\widetilde{S}(s)}{s} \right)^2 ds \\
=& \int_1^t dh(\widetilde{S}(s), s) - \frac{1}{2s} ds = \frac{1}{2} \left( \frac{\widetilde{S}(t)^2}{s} - \widetilde{S}(1)^2 - \log t \right) .
\end{align}
Under the null $H_0$, $(Y_t^{\rmlSPRT})_{t \geq 1}$ is a martingale with initial value 1. Therefore, from Ville's inequality,
\begin{align}
&\Pr \left[ \frac{1}{2} \left( \frac{W(t)^2}{t} - \log t \right) \geq x \right]\\
 =& E \left[ \Pr \left[ \frac{1}{2} \left( \frac{W(t)^2}{t} - W(1)^2 - \log t \right) \geq x - \frac{1}{2}W(1)^2 \mid W(1)\right] \right] \\
=& E \left[ \exp\left( \left(-x + \frac{1}{2} W(1)^2 \right)_+ \right) \right].
\end{align}
It can readily be checked that putting $x = -\log \widetilde{\alpha}_1(\alpha)$ sets the above quantity to $\alpha$. Therefore, inverting $Y^{\rmlSPRT}_t$ gives that, with probability $1 -\alpha$,
\begin{align}
|W(t)| \leq \sqrt{t \left(-2 \log \widetilde{\alpha}_1(\alpha) + \log t \right)}, \qquad \forall t \geq 1.\label{eq:rmle_CS}
\end{align}

\subsection{Toward non-parametricity: delayed start mixture SPRTs under i.i.d. data}

Theorem 2 in \cite{robbins1970boundary} asserts (we present here a slight two-sided modification of the result) that for a sequence $X_1, X_2, \ldots$ of i.i.d. random variables with mean 0 and variance 1, and a boundary function $(c(u))_{u \geq u_0}$ such that (i) $c(u) u^{-1/2}$ is non-decreasing for $u$ large enough and (ii) $\int_{u_0}^\infty  u^{-3/2} c(u) \exp(-c(u)^2 / 2) du < \infty$, it holds that 
\begin{align}
\lim_{t_0 \to \infty} P\left[\forall t \in \mathbb{N} \cap [u_0 t_0, \infty), |S_t| \leq \sqrt{t_0} c(t / t_0) \right] = P\left[ \forall u \geq u_0, |W(u)| \leq c(u) \right].
\end{align}
This result therefore allows, in nonparametric i.i.d. settings, to obtain approximate confidence sequences from a confidence sequence for the Wiener process, by means of a burn-in period and a time-rescaling. Applying this result to \eqref{eq:nmSPRT_boundary} and \eqref{eq:rmle_CS} yields that
\begin{align}
\lim_{t_0 \to \infty} & \Pr \left[\forall t \in \mathbb{N} \cap [t_0, \infty),\ |S_t| \leq \sqrt{(t + \eta t_0) \left(-2 \log \widetilde{\alpha}_{2,\eta}(\alpha)  + \log \frac{t + \eta t_0}{\eta t_0} \right)} \right] \\
= & \Pr \left[\forall u \geq 1, |W(u)| \leq \sqrt{(u+\eta) \left(-2  \log \widetilde{\alpha}_{2,\eta}(\alpha)  + \log \frac{u + \eta}{\eta} \right)} \right] \label{eq:delayed_start_nmSPRT_i.i.d.}
\end{align}
and 
\begin{align}
\lim_{t_0 \to \infty} & \Pr \left[\forall t \in \mathbb{N} \cap [t_0, \infty),\ |S_t| \leq \sqrt{t \left(-2 \log \widetilde{\alpha}_1(\alpha)  + \log \frac{t}{t_0} \right)} \right] \\
= & \Pr \left[\forall u \geq 1, |W(u)| \leq \sqrt{u \left(-2 \log \widetilde{\alpha}_1(\alpha)  + \log u\right)} \right].\label{eq:delayed_start_rmlSPRT_i.i.d.}
\end{align}
The key enabling result in the proof of theorem 2 in \cite{robbins1970boundary} is Donsker's weak invariance principle for i.i.d. random variables.

\subsection{Modern invariance principle based CSs}\label{sec:modern_IP-based_CSs}

\paragraph{Weak invariance principle for sequential testing and confidence sequences.} Theorem 10 in \cite{bibaut2021sequential}  uses \cite{mcleish1974dependent}'s weak invariance principle for martingale-difference triangular arrays to provide a method for constructing non-parametric asymptotic confidence sequences from a confidence sequence for the Wiener process. We present the result here in the case that $X_1,X_2,\ldots$ is an $\mathfrak{F}$-adapted sequence with $\mathrm{Var}(X_t \mid \mathcal{F}_{t-1}) = 1$, $\psi_t = E[X_t \mid \mathcal{F}_{t-1}]$ and $H_0 : \psi_t = 0\ \forall t \geq 1$, that is $(X_t)$ is a martingale difference sequence.

Let $u_0 \in [0,1]$ and et $(c(u))_{u \in [u_0,1]}$ be a symmetric $(1-\alpha)$-confidence sequence for the Wiener process on $[u_0,1]$, that is $\Pr[\forall u \in [u_0,1], |W(u)| \leq c(u)] \geq 1-\alpha$. As a relatively direct corollary of theorem 3.2 in \cite{mcleish1974dependent}, it holds that 
\begin{align}
\lim_{T  \to \infty} \Pr\left[ \forall t \in \mathbb{N} \cap [u_0 T, T] |S_t| \leq \sqrt{T} c(t/T) \right] = \alpha.
\end{align}
Here $T$ plays the role of the maximum runtime of the experiment, while $u_0$ is the fraction of $T$ the experimenter uses as burn-in time. The need for $T$ is a theoretical limitation, although it might not be a practical one, as experimenters generally have a  time budget or sample size budget to spend on a trial. In the case of a martingale data sequence, we conjecture that concentration-inequality-based methods similar to the ones used to prove theorem 2 from \cite{robbins1970boundary} could be used to show $\lim_{T  \to \infty} \Pr[ \forall t \in \mathbb{N} \cap (T, \infty),\ |S_t| \leq \sqrt{T} c(t/T) ] = 0$, and therefore get rid of the need for a maximum experiment runtime. However, doing so might be more complex in the originally intended martingale-difference array setting.

\paragraph{Strong invariance principle for asymptotic time-uniform confidence sequences.} \cite{smith} introduce the use of strong invariance principles (also known as almost sure invariance principles or strong approximation results) to construct confidence sequences in nonparametric settings. One key contribution of their work is to introduce a definition of asymptotic confidence sequence (AsympCS). They say that a sequence $(c_t)_{t \in \mathbb{N}}$ is a symmetric $(1-\alpha)$-asymptotic confidence sequence for the partial sum process $(S_t)_{t \in \mathbb{N}}$ if there exists a $(1-\alpha)$-exact confidence sequence $(c_t^*)_{t \in \mathbb{N}}$ for $(S_t)_{t \in \mathbb{N}}$ such that $c_t / c^*_t \rightarrow 1$ almost surely.

Successive versions of this article use different strong approximation results. Starting from version 5, they have been using \cite{strassen1967}'s strong invariance principle for martingales and allows for asymptotic confidence sequences for the general setting where $(S_t)_{t \in \mathbb{N}}$ is a martingale. Starting from version 7, they also include type I guarantees for sequences of delayed start confidence sequences. Specifically, they replicate with their techniques (and lift some assumptions for) the type-I error guarantee for sequences of delayed-start normal-mixture sequences, initially proven under martingale data in version 1 of the current paper. They also extend the guarantee \eqref{eq:delayed_start_rmlSPRT_i.i.d.} for sequences of the delayed-start running MLE SPRT sequence to martingale data. 

In our view, the main focus of \cite{smith} is proposing a novel definition of AsympCS that is asymptotically close in \emph{parameter space} to an exact $1-\alpha$ confidence sequence for the parameter of interest. Meanwhile, our work focuses on the testing properties, that is, type-I error and expected rejection time, of \emph{sequences} of confidence sequences or of tests. In particular, the current version of their paper, version 7 at the time of writing of this work, does not provide a rejection time analysis.  One difference between our setting and theirs is that we impose the normalization condition of the conditional variance, which they don't. We use this condition heavily in our rejection time analysis. We leave to future work the discussion of whether this condition is actually necessary for a rejection time analysis.

\section{Related Literature}\label{sec:lit}
Classical approaches to hypothesis testing have predominantly dealt with experiments of fixed, predetermined sample sizes, which we refer to as the \textit{fixed-n} kind. The emphasis on \textit{fixed-n} tests by early pioneers such as Fisher is presumably a consequence of the motivating applications that drove the development of hypothesis testing procedures in the first half 19th century, in which outcomes of an experiment were only available long after the experiment had been designed, such as in agricultural research \citep{armitage2}. As tests could only be performed once, \textit{fixed-n} tests were designed to maximize power subject to a type-I error constraint \citep{neymanpearson}. Increasingly in modern experiments, however, observations from experimental units become available sequentially instead of simultaneously, providing many opportunities to perform a test instead of just one. The application of \textit{fixed-n} tests to sequential designs is made difficult because it requires making an undesirable trade-off balancing the competing objectives of detecting large effects early and detecting small effects eventually. Performing the test later risks exposing many experimental units to a potentially large and harmful treatment effect, while performing the test early risks a high type-II error for small effects. These desires have led to bad statistical practices whereby fixed-n procedures are naively applied to accumulating sets of data, see \cite{peeking} for a discussion pertaining to online A/B tests, which sacrifice type-I error guarantees \citep{armitage}, permitting the analyst to incorrectly sample to a foregone conclusion \citep{anscombe}.

For modern sequential designs, sampling until a hypothesis is proven or disproven appears to be a very natural form of scientific inquiry, which requires testing procedures to preserve their type-I/II error guarantees under continuous monitoring. Sequential inference is fundamentally tied to the theory of martingales \citep{ramdaskoolen}. A test martingale is a statistic that is a nonnegative supermartingale under the null hypothesis.  Ville's inequality \citep{ville} is then used to bound the supremum of the process to provide a time-uniform type-I error guarantee. Research into sequential analysis in the statistics literature began with the introduction of the sequential probability ratio test (SPRT) \citep{wald, waldsequential}. Although Wald did not reference martingale theory in the exposition of the SPRT, the connection is clear in hindsight by observing that the likelihood ratio is a nonnegative supermartingale under the null. The simple-vs-simple SPRT enjoys the optimality property of being the sequential test that minimizes the average sample number (expected stopping time) among all sequential tests with no larger type-I/II error probabilities \citep{waldwolfowitz}. This is extended to the continuous-time version in \cite{dkw}.

The SPRT for simple-vs-simple testing problems and the mixture SPRT (mSPRT) for composite testing problems can be interpreted as Bayes factors \citep{jeffreys, kass}, forming a bridge between Bayesian, frequentist, and conditional frequentist approaches to sequential testing \citep{conditional_frequentist_simple, conditional_frequentist_nested}. The SPRT also appears in Bayesian decision-theoretic approaches to sequential hypothesis testing in which there is a constant cost per observation \citep{waldwolfowitzbayes, bergerdecision}. However, care must be taken when specifying priors in composite testing problems, should one seek to have strict frequentist guarantees \cite{deHeide2021}. Composite tests in statistical models with group invariances can often be reduced to simple hypothesis tests by constructing \emph{invariant SPRTs} \citep{invariantsprt} based on a maximally invariant test statistic \citep{lehmann2005testing, LehmCase98}. These invariant SPRT test statistics can be obtained as Bayes factors by using the appropriate right-Haar priors on nuisance parameters in group invariant models \citep{hendriksen}. Such arguments were used by \citep{robbins1970statistical} to develop sequential tests for location-scale families with unknown scale parameters.

Confidence sequences \citep{darling67} can be obtained by inverting a sequential test, and sequential $p$-values can be obtained by tracking the reciprocal of the supremum of the test martingale. Together, these generalize the coverage and type-I guarantees held by \text{fixed-n} confidence intervals and $p$-values to hold uniformly through time. Procedures with these guarantees are appropriately referred to as ``anytime valid." Relationships between test-martingales, sequential $p$-values and Bayes factors are discussed in \cite{shafer}. Nonparametric confidence sequences under sub-Gaussian and Bernstein conditions are provided in \cite{howard2021}. These results are nonasymptotic, yielding valid confidence sequences for all times, but may be conservative. Confidence sequences for quantiles and anytime-valid Kolmogorov-Smirnov tests are provided in \cite{howardquantile}. \cite{smith} obtain asymptotic confidence sequences, in the sense that the intervals converge almost surely to a valid confidence sequence with an error that is orders smaller than the width of the latter. They achieve this by approximating the sample average process by a Gaussian process using strong invariance principles \citep{strassen1964invariance, strassen1967, Komlos1975,Komlos1976}, like us. Their focus is on having approximate confidence sequence width, which need not translate to type-I error guarantees. In particular, there is a risk of rejecting too early when the cumulative sum does not look normal yet. Moreover, they only guarantee that a similar-width confidence sequence has at-least-$\alpha$ coverage, but do not characterize its power, only that the width has the right rate dependence on $t$. Therefore, at the same time, if we do wait, the confidence sequences can be overly conservative.

For certain continuous-time martingales, Ville's inequality is an equality \citep[lemmas 1 and 2]{robbins1970boundary}. For discretely observed martingales, however, Ville's inequality is generally strict, meaning the type-I-error guarantees it yields for test martingale are conservative. The conservativeness follows from the amount by which the stopped sum process exceeds the rejection boundary (zero in the continuous case) and is often referred to as the ``overshoot'' problem with the SPRT \citep{siegmundbook}. Understanding the size of the overshoot is key to understanding how conservative existing bounds are on type-I error and expected stopping times. \cite{wald}'s approximation to the type-I error is obtained by simply ignoring the overshoot. \cite{siegmund75} obtains an approximation to the type-I error for the simple-vs-simple SPRT in exponential-family models as complete asymptotic expansions in powers of $\alpha^{-1}$ with exponentially small remainder as $\alpha \rightarrow 0$.  With mSPRTs the rejection boundary is curved, and studying the distribution of the overshoot is often tackled via nonlinear renewal theory \citep{woodroofe67, woodroofebook, zhang88}. As $\alpha\rightarrow 0$, \cite{laisiegmund1977, laisiegmund1979} derive asymptotic approximations to the expected value and distribution function of the nmSPRT stopping time under the null so as to study the type-I error resulting from truncated nmSPRT tests. Similar results for the expected stopping times can be found in \cite{hagwoodwoodrofe}. To our knowledge existing work has focused on asymptotic ($\alpha\rightarrow 0$) approximations to moments of stopping times for parametric SPRTs which yield sharper results than \cite{wald}'s when neglecting the overshoot. While previous authors also use these tools to obtain type-I errors for truncated sequential tests, no attention has been given to calibrating the type-I error for open-ended sequential tests.

As trends in online experimentation shift toward streaming approaches, sequential approaches to A/B testing have seen increased adoption \citep{johari2022always, lindon, lindon20}. In other applications, particularly in medicine, it may not be possible to test after every new observation. In clinical trials, a small number of interim analyses may be planned, which does not warrant a fully sequential test. Instead, group sequential tests \citep{pocock,obrien,demets,jennison1999group} can be performed which provide a calibrated sequential test over a fixed and finite number of analyses. Analogous to confidence sequences, \emph{repeated confidence intervals} provide strict coverage uniformly across all interim analyses \citep{rci, jennison84}. These procedures are useful when testing on a certain cadence, such as daily, suffices and when a terminal endpoint of the experiment is known. They are, however, not as flexible as fully sequential procedures as they do not allow the experiment to continue past the final analysis, having fully spent their $\alpha$-budget.

Test martingales are closely related to \emph{e-processes}. An e-variable is a random variable (or statistic) that has expectation at most 1 under the null hypothesis \citep{grunwald}. An e-process is a nonnegative process, upper bounded by a nonnegative supermartingale, such that the stopped process is an e-variable under any stopping rule \citep{ruf22}, although it itself may not be a nonnegative supermartingale \citep{ramdasexchange}. Thanks to this property it is possible to build sequential tests from e-processes.  \cite{gametheoryav} provide a review of test martingales, e-processes, anytime valid inference and game theoretic probability and its applications to sequential testing. See, for example, log-rank tests \citep{schure}, contingency tables \citep{schure} and changepoint detection \citep{shin}. See also the running-MLE sequential likelihood ratio test of \cite{wasserman2020universal}.

\bibliography{biblio}
\bibliographystyle{plainnat}

\appendix

\section{Proofs}

\section{Proofs of the type-I error results}

\subsection{Strong invariance principle}

The proof of \cref{thm:typeI_error} relies on the following corollary of \cite{strassen1967}'s strong invariance principle theorem 1.3.

\begin{proposition}\label{prop:sip}[Strong invariance principle]
Suppose that \cref{asm:sip} holds. Then, we can enlarge the underlying probability space so that it supports a standard Wiener process $W$ such that
$S_t^0 = W(t) + o( \sqrt{t / \log \log t})$ almost surely, as $t \to \infty$.
\end{proposition}

\begin{proof}[Proof of \cref{prop:sip}] 
From theorem 1.3 in \cite{strassen1967}, the probability space can be enlarged so that it supports a standard Wiener process $W$ such that it holds almost surely that
\begin{align}
S^0_t  - W(V_t) = o((V_t f(V_t))^{1/4} \log V_t) \qquad \text{ as } t \to \infty.
\end{align}
From \eqref{eq:sip_var} and lemma 4.2 in \cite{strassen1967}, 
\begin{align}
W(V_t) - W(t) = o((t f(t))^{1/4} \log t) \qquad \text{ as } t \to \infty.
\end{align}
Therefore, it holds almost surely that, as $t \to \infty$, 
\begin{align}
S^0_t - W(t) =& o \left( (V_t f(V_t))^{1/4} \log V_t + (t f(t))^{1/4} \log t \right) \\
=& o \left( \sqrt{\frac{V_t}{\log \log V_t}} + \sqrt{\frac{t}{\log \log t}}\right) \\
=& o \left( \sqrt{\frac{t}{\log \log t}} \right)
\end{align}
as from \eqref{eq:sip_var}, $V_t \sim t$ as $t \to \infty$ almost surely.
\end{proof}

\subsection{Proof of the delayed-start boundary crossing probabilities Wiener processes (\cref{lemma:typeIerr_delayed_Wiener})}

\begin{proof}[Proof of \cref{lemma:typeIerr_delayed_Wiener}]
We will make use of the fact that the process $W_{t_0}$ defined for any $s \geq 0$ by $W_{t_0}(s) = W(s t_0) / \sqrt{t_0}$ is a standard Wiener process.

\subsubsection{Proof of the claim on $\widetilde{Y}^\rmlSPRT$} We have that
\begin{align}
&\Pr\left[\sup_{t \geq t_0} \widetilde{Y}^\rmlSPRT (t) \geq - \log \widetilde{\alpha}_1(\alpha) - \frac{1}{2} \log t_0 \right] \\
=& \Pr \left[ \exists s \geq 1 : \frac{1}{2} \left( \frac{W(s t_0)^2}{s t_0} - \log s \right) \geq - \log \widetilde{\alpha}_1(\alpha) \right] \\
=& \Pr \left[ \exists s \geq 1 : \frac{1}{2} \left( \frac{W_{t_0}(s)^2}{s} - W_{t_0}(1)^2 - \log s  \right) \geq -\log \widetilde{\alpha}_1(\alpha) - \frac{1}{2} W_{t_0}(1)^2  \right] \\
=& E \left[ \Pr \left[ \exists s \geq 1 : \frac{1}{2} \left( \frac{W_{t_0}(s)^2}{s} - W_{t_0}(1)^2 - \log s  \right) \geq -\log \widetilde{\alpha}_1(\alpha) - \frac{1}{2} W_{t_0}(1)^2  \mid W_{t_0}(1)^2 \right] \right] \\
=& E\left[ \exp \left(-\log \widetilde{\alpha}_1(\alpha) - \frac{1}{2} W_{t_0}(1)^2 \right)_+ \right],
\end{align}
where the last line follows, via Ville's equality, from the fact that conditional on $W_{t_0}(1)$, the process $(\exp((W_{t_0}(s)^2 / s - W_{t_0}(1)^2 - \log s) / 2)_{s \geq 1}$ is a time-continuous positive martingale with initial value 1 and continuous sample paths.
For any $a > 0$, we have that
\begin{align}
&E\left[ \exp\left( - \left( a - \frac{1}{2} W_{t_0}(1)^2 \right)_+ \right)\right] \\
=& \int_{-\sqrt{2a}}^{\sqrt{2 a}} \frac{1}{\sqrt{2 \pi}} \exp\left( - a + \frac{1}{2} z^2 - \frac{1}{2} z^2 \right) dz + 2 \left( 1 - \Phi(\sqrt{2 a}) \right) \\
=& 2 \exp(-a) \sqrt{\frac{a}{\pi}} + 2 \left( 1 - \Phi(\sqrt{2 a}) \right),
\end{align}
which yields the claim.

\subsubsection{Proof of the claim on $Y^{\nmSPRT}$} For any $a > 0$, we have that
\begin{align}
&\Pr \left[ \sup_{t \geq t_0} Y^{\nmSPRT}_{t_0, \eta, t} \geq a \right] \\
=& \Pr\left[ \exists t \geq t_0 : \frac{1}{2} \left( \frac{W(t)^2}{t + \lambda} - \frac{W(t_0)^2}{t_0 + \lambda} - \log \frac{t + \lambda}{t_0 + \lambda}\right) + \frac{1}{2} \left( \frac{W(t_0)^2}{t_0 + \lambda} - \log \frac{t_0 + \lambda}{\lambda} \right) \geq a \right] \\
=& \Pr\left[ \exists s \geq 1 : \frac{1}{2} \left( \frac{W_{t_0}(s)^2}{s  + \eta} - \frac{W_{t_0}(1)^2}{1 + \eta} - \log \frac{s + \eta}{1 + \eta}\right) + \frac{1}{2} \left( \frac{W_{t_0}(1)^2}{1 + \eta} - \log \frac{1 + \eta}{\eta} \right) \geq a \right]\\
=& E \left[ \Pr\left[ \exists s \geq 1 : \frac{1}{2} \left( \frac{W_{t_0}(s)^2}{s  + \eta} - \frac{W_{t_0}(1)^2}{1 + \eta} - \log \frac{s + \eta}{1 + \eta}\right) \right. \right. \\
&\left. \left. \qquad \qquad \geq a + \frac{1}{2} \log \frac{1+\eta}{\eta} - \frac{1}{2} \frac{W_{t_0}(1)^2}{1+\eta}\mid W_{t_0}(1) \right] \right]\\
=& E \left[ \exp\left(-\left(a_\eta - \frac{1}{2(1+\eta)} W_{t_0}(1)^2 \right)_+ \right) \right],
\end{align}
with $a_\eta = a + 2^{-1} \log ((1+\eta)/\eta)$, where the last line follows, via Ville's equality, from the fact that, conditional on $W_{t_0}(1)$, the process
\begin{align}
\left( \exp\left( \frac{1}{2} \left( \frac{W_{t_0}(s)^2}{s  + \eta} - \frac{W_{t_0}(1)^2}{1 + \eta} - \log \frac{s + \eta}{1 + \eta}\right) \right) \right)_{s \geq 1}
\end{align}
is a time-continuous positive martingale with continuous sample paths and initial value 1. We have that
\begin{align}
&E \left[ \exp\left(-\left(a_\eta - \frac{1}{2(1+\eta)} W_{t_0}(1)^2 \right)_+ \right) \right] \\
=& \int_{-\sqrt{2(1+\eta) a_\eta}}^{\sqrt{2 (1+\eta) a_\eta}} \exp\left(-a_\eta + \frac{1}{2} \frac{z^2}{1+\eta} - \frac{1}{2} z^2 \right) \frac{1}{\sqrt{2 \pi}} dz + 2 \left( 1 - \Phi\left( \sqrt{2 ( 1 +\eta) a_\eta } \right) \right) \\
=& \int_{-\sqrt{2(1+\eta) a_\eta}}^{\sqrt{2 (1+\eta) a_\eta}} \exp\left(-a_\eta - \frac{1}{2} \frac{\eta}{1+\eta} z^2 \right) \frac{1}{\sqrt{2 \pi}} dz + 2 \left( 1 - \Phi\left( \sqrt{2 ( 1 +\eta) a_\eta } \right) \right) \\
=& \exp(-a_\eta) \sqrt{\frac{1+\eta}{\eta}} \left( 2 \Phi\left(\sqrt{2 \eta a_\eta} \right) -1 \right) + 2 \left( 1 - \Phi\left( \sqrt{2 ( 1 +\eta) a_\eta } \right) \right) \\
=& \exp(-a)  \left( 2 \Phi\left(\sqrt{2 \eta a_\eta} \right) -1 \right) + 2 \left( 1 - \Phi\left( \sqrt{2 ( 1 +\eta) a_\eta } \right) \right),
\end{align}
which yields the claim.
\end{proof}

\subsection{Proof of the type-I error theorem (\cref{thm:typeI_error})}

The proof of the type-I claim for the delayed-start normal-mixture SPRT makes use of the following technical lemma.

\begin{lemma}\label{lemma:h_Lipschitz}
For any $\eta > 0$, the functions $-\log h_1$ and $-\log h_{2, \eta}$ is  1-Lipschitz on $\mathbb{R}$.
\end{lemma}

\begin{proof}[Proof of lemma \ref{lemma:h_Lipschitz}]
From the derivation in the proof of \cref{lemma:typeIerr_delayed_Wiener},
\begin{align}
h_1(a) =& E\left[ \exp\left(- \left(a - \frac{1}{2}X^2\right)_+\right)\right]\\
\text{and} \qquad  h_{2,\eta}(a) =& E\left[ \exp\left(- \left(a + c_\eta - \frac{1}{2(1+\eta)}X^2\right)_+\right)\right],
\end{align}
where $X \sim \mathcal{N}(0,1)$ and $c_\eta = 0.5 \log ((1+\eta) / \eta)$.
Differentiating under the integral sign, we obtain that,
\begin{align}
(-\log h_1)'(a) = \frac{E\left[ \bm{1} \left\lbrace a - \frac{1}{2}X^2 > 0 \right\rbrace \exp \left(  -\left( a - \frac{1}{2} X^2\right)_+ \right) \right]}{E\left[ \exp \left(  -\left( a - \frac{1}{2} X^2\right)_+ \right) \right]} \in (0,1),
\end{align}
and
\begin{align}
&(- \log h_{2, \eta})'(a) \\
 =& \frac{E\left[ \bm{1} \left\lbrace  a + c_\eta - \frac{1}{2(1+\eta)}X^2 > 0 \right\rbrace  \exp\left( - \left( a + c_\eta - \frac{1}{2(1+\eta)}X^2 > 0\right)_+ \right) \right]}{E\left[  \exp\left( - \left( a + c_\eta - \frac{1}{2(1+\eta)}X^2 > 0\right)_+ \right) \right]} \in (0,1).
\end{align}
\end{proof}

\begin{proof}[Proof of \cref{thm:typeI_error}]
Let $((t_{0,m}, \alpha_m, \lambda_m))_{m \geq 1}$ be a sequence such that $t_{0,m} \to \infty$, and $\alpha_m \in [0,1],\ \lambda_m > 0$ for every $m$.

Let 
\begin{align}
A_{1,m} &= \sup_{t \geq t_{0,m}} Y_t^{\rmlSPRT} + \frac{1}{2} t_{0,m}, \qquad \qquad \widetilde A_{1,m} = \sup_{t \geq t_{0,m}} \widetilde Y_t^{\rmlSPRT} + \frac{1}{2} t_{0,m}\\
A_{2, m} &= \sup_{t \geq t_{0,m}} Y_{\lambda_m, t}^{\nmSPRT}, \qquad \text{and} \qquad  \widetilde A_{2, m} = \sup_{t \geq t_{0,m}} \widetilde Y_{\lambda_m, t}^{\nmSPRT}.
\end{align}

\paragraph{A preliminary observation} Under $H_0$, it holds that, for $i=1,2$,
\begin{align}
&|A_{i,m} - \widetilde A_{i,m}|\\
\leq & \frac{1}{2} \sup_{t \geq t_{0,m}} \frac{|(S^0_t)^2 - W(t)^2|}{t} \\
\leq & \frac{1}{2} \sup_{t \geq t_{0,m}} \frac{|S^0_t - W(t)|}{\sqrt{t / \log \log t}} \times \left( \sup_{t \geq 1} \frac{2 W(t)}{\sqrt{t \log \log t}} + \sup_{t \geq t_{0,m}} \frac{|S^0_t - W(t)|}{\sqrt{t \log \log t}} \right)\\
=& o(1) \text{ a.s. as } m \to \infty 
\end{align}
since $\sup_{t \geq t_{0,m}} |S^0_t - W(t)| / \sqrt{t / \log \log t} = o(1)$  a.s. as  $m \to \infty$ from proposition \ref{prop:sip}, and $\sup_{t \geq 1}  W(t) / \sqrt{t \log \log t} < \infty$ a.s. from the law of the iterated logarithm.

\paragraph{Transformed statistics with a fixed limit distribution}
Let $\eta_m  = \lambda_m / t_{0,m}$,
\begin{align}
B_{1,m} =& - \log h_1(A_{1,m}), \qquad \qquad \widetilde B_{1,m} = - \log h_1(\widetilde  A_{1,m}),\\
B_{2,m} =& - \log h_{2, \eta_m}(A_{2,m}), \qquad \text{and} \qquad \widetilde B_{2,m} = - \log h_{2, \eta_m}(\widetilde  A_{2,m}).
\end{align}
From the fact that $-\log h_1$ and $-\log h_{2, \eta}$ are 1-Lipschitz (lemma \ref{lemma:h_Lipschitz}) and from the preliminary observation, 
\begin{align}
B_{i,m} = \widetilde B_{i,m} + o(1) \text{ a.s. for } i =1,2. 
\end{align}
For any $\alpha \in [0,1]$, by definition of $\widetilde \alpha_1$ and $\widetilde \alpha_{2, \eta_m}$, and from lemma \ref{lemma:typeIerr_delayed_Wiener}, for $i=1,2$,
\begin{align}
\Pr\left[ \widetilde B_{i,m} \geq -\log \alpha \right] = \alpha,
\end{align}
that is, $\widetilde B_{i,m}$ is distributed as $\mathrm{Exp}(1)$ for any $i,m$. Observe that, again by definition of $\widetilde \alpha_1$ and $\widetilde \alpha_{2, \eta_m}$,
\begin{align}
\Pr\left[ A_{1,m} \geq -\log \widetilde \alpha_1(\alpha_m) \right] =& \Pr \left[ B_{1,m} \geq -\log \alpha_m \right] \\
\text{and} \qquad \Pr\left[ A_{2,m} \geq -\log \widetilde \alpha_{2,\eta_m}(\alpha_m) \right] =& \Pr \left[ B_{2,m} \geq -\log \alpha_m \right].
\end{align}
Therefore, the claim will follow if we show that, for $i=1,2$ $\Pr \left[ B_{2,m} \geq -\log \alpha_m \right] / \alpha_m \to 1$ as $m \to \infty$. This is indeed the case as, for $i=1,2$,
\begin{align}
&\Pr \left[ B_{i,m} \geq -\log \alpha_m \right] / \alpha_m \\
=& \Pr \left[ \widetilde B_{i,m} \geq -\log \alpha_m  + o(1) \right] / \alpha_m \\
=& \exp(-o(1)) \\
\to & 1,
\end{align}
where the above equalities hold a.s. and convergence is as $m \to \infty$.
\end{proof}

\section{Proofs of the representation results}

\subsection{Proof of the It\^o-like finite difference result (\cref{lemma:discrete_Ito_identity})}

\begin{proof}[Proof of \cref{lemma:discrete_Ito_identity}]
For any $s \in \mathbb{N}$, $\lambda \in \mathbb{R}^+$ such that $s + \lambda > 0$, we have that
\begin{align}
   &\frac{1}{2} \left( \frac{z_{s+1}^2}{s+1+\lambda} - \frac{z_s^2}{s+\lambda} \right) \\
   =& \frac{1}{2} \left(  \frac{z_{s+1}^2 - z_s^2}{s+1+\lambda} - \frac{z_s^2}{(s+1+\lambda)(s+\lambda)}\right)\\
   =& \frac{1}{2} \left( \frac{z_{s+1} + z_s}{s+1+\lambda} (z_{s+1} - z_s) - \frac{z_s^2}{(s+1+\lambda)(s+\lambda)} \right)\\
   =&  \frac{z_s}{s+1+\lambda} (z_{s+1} - z_s) - \frac{1}{2} \frac{z_s^2}{(s+1+\lambda)(s+\lambda)} + \frac{1}{2} \frac{(z_{s+1} - z_s)^2}{s+1+\lambda}\\
   =& \frac{s+\lambda}{s+1+\lambda} \left( \frac{z_s}{s + \lambda} (z_{s+1} - z_s) - \left(\frac{z_s}{s+ \lambda} \right)^2 \right) + \frac{1}{2} \frac{(z_{s+1} - z_s)^2}{s+1+\lambda}.
\end{align}
\end{proof}

\subsection{Proof of the expanded representation result (\cref{thm:expanded_representation})}

\begin{proof}[Proof of \cref{thm:expanded_representation}]
From \cref{thm:rep_non_anticip_magle}, we have that
\begin{align}
Y^{\rmlSPRT}_t =& \sum_{s=0}^{t-1} \frac{s}{s+1} \left( \left(\psi + \psi^0_s \right) \left( \psi + X^0_{s+1} \right) - \frac{1}{2} \left(\psi + \psi^0_s \right)^2 \right) + \frac{1}{2} \sum_{s=1}^t \frac{\left(\psi + X^0_s \right)^2}{s} - \log t + R^{\bias}_{0,t} \\
=& \frac{1}{2} \psi^2 \left\lbrace \sum_{s=0}^{t-1} \frac{s}{s+1} + \frac{1}{2} \sum_{s=1}^{t-1} \frac{1}{s+1} \right\rbrace + \psi \sum_{s=0}^{t-1} X^0_{s+1} \left( \frac{s}{s+1} + \frac{1}{s+1} \right) \\
&+ \sum_{s=0}^{t-1} \frac{s}{s+1} \psi^0_s X^0_{s+1} - \log t - \frac{1}{2} \sum_{s=0}^{t-1} \frac{s}{s+1} \left( \psi^0_{0,s}\right)^2 \\
&+ \frac{1}{2} \left\lbrace \sum_{s=1}^t \frac{(X^0_s)^2}{s} - \log t \right\rbrace + R^\bias_{0,t} \\
=& \frac{1}{2} \psi^2 t + M^{(0)}_t + M^{(1)}_{0,t} - R^\adpt_{0,t} + \Deltaqvar_{0,t} + R^\bias_{0,t} + \frac{1}{2} (X^0_1)^2.
\end{align}
We now turn to the expanded representation of $Y^{\nmSPRT}_{\lambda, t}$. From \cref{thm:rep_non_anticip_magle}, we have that
\begin{align}
Y^\nmSPRT_{\lambda, t} =& \sum_{s=0}^{t-1} \frac{s+\lambda}{s+1+\lambda} \left\lbrace \left(\psi \frac{s}{s+\lambda} + \psi^0_{\lambda, s} \right) ( \psi + X^0_{s+1}) - \frac{1}{2} \left(\psi \frac{s}{s+\lambda} + \psi^0_{\lambda, s} \right)^2 \right\rbrace \\
+&\frac{1}{2} \left\lbrace \sum_{s=1}^t \frac{(\psi + X^0_s)^2}{s+\lambda} - \log \frac{t+\lambda}{\lambda} \right\rbrace + R^\bias_{\lambda, t} \\
=& A + B + C +  M^{(1)}_{\lambda, t} - R^{\adpt}_{t,\lambda} + \Deltaqvar_{\lambda, t}
\end{align}
with 
\begin{align}
A =& \psi^2 \sum_{s=0}^{t-1} \left\lbrace \frac{s}{s+1+\lambda} - \frac{1}{2} \frac{s^2}{(s+1+\lambda)(s+\lambda)} + \frac{1}{2} \frac{1}{s+1+\lambda} \right\rbrace \\
=& \frac{1}{2} \psi^2 t + \frac{1}{2} \sum_{s=0}^{t-1} \frac{2s (s + \lambda) - s^2 + s + \lambda - (s + \lambda)(s+1+\lambda)}{(s+1+\lambda)(s+\lambda)} \\
=& \frac{1}{2} \psi^2 t - \frac{1}{2} \psi^2 \lambda^2 \sum_{s=0}^{t-1} \frac{1}{(s+\lambda)(s+1+\lambda)}\\
= & \frac{1}{2} \psi^2 t - \frac{1}{2} \psi^2 \left( \frac{1}{\lambda} - \frac{1}{t+\lambda} \right) \\
= & \frac{1}{2} \psi^2 t - R^{\skg,1}_{\lambda, t},
\end{align}
\begin{align}
B = & \psi \sum_{s=0}^{t-1} X^0_{s+1} \left\lbrace \frac{s}{s+1+\lambda} + \frac{1}{s+1+\lambda} \right\rbrace \\
=& M^{(0)}_t - M^{\skg}_{\lambda,t},
\end{align}
and
\begin{align}
C = &  \psi \sum_{s=0}^{t-1} \psi^0_{\lambda, s} \left\lbrace \frac{s + \lambda}{s + 1  + \lambda} - \frac{s}{s + 1 + \lambda} \right\rbrace \\
=& R^{\skg, 2}_{\lambda, t}.
\end{align}

\end{proof}

\section{Proofs of results on expected stopping times}

\subsection{Proofs of the results pertaining to the random threshold}

\subsubsection{Proof of \cref{lemma:stopping_time_and_random_thresh}}

\begin{proof}[Proof of \cref{lemma:stopping_time_and_random_thresh}]
We only present the proof for $\tau_1$ as the proof for $\tau_2$ is identical.
    Observe that we can write $\tau_1$ as
\begin{align}
    \tau_1 = \min \left\lbrace t \geq t_0 : Y^\rmlSPRT_t - Y^\rmlSPRT_{t_0} \geq -\log \widetilde{\alpha}_1(\alpha) - Y^\rmlSPRT_{t_0} \right\rbrace,
\end{align}
in which $-\log \widetilde{\alpha}_1(\alpha) - Y^\rmlSPRT_{t_0}$ plays the role of the effective threshold.

We argue that writing the stopping time this way implies the second inequality in the claim, that is
\begin{align}
\left(-\log \widetilde{\alpha}_1(\alpha) - Y^\rmlSPRT_{t_0}\right)_+ \leq
    Y^\rmlSPRT_{\tau_1} - Y^\rmlSPRT_{t_0}.
\end{align} Indeed, this trivially holds by definition of $\tau_1$ if the effective threshold $-\log \widetilde{\alpha}_1(\alpha) - Y^\rmlSPRT_{t_0}$ is non-negative. If it is negative, then $\tau_1 = t_0$, and the inequality also holds. Thus, the second inequality holds in both cases.

We now turn to the first inequality in the claim.
If the threshold is positive, $\tau_1$ must be at least $t_0 + 1$, and then by definition of $\tau_1$, 
\begin{align}
     Y^\rmlSPRT_{\tau_1-1} - Y^\rmlSPRT_{t_0} \leq -\log \widetilde{\alpha}_1(\alpha) - Y^\rmlSPRT_{t_0} = \left(-\log \widetilde{\alpha}_1(\alpha)- Y^\rmlSPRT_{t_0} \right)_+.
\end{align}
If the threshold is non-positive, then $\tau$ equals $t_0$, and then 
\begin{align}
Y^\rmlSPRT_{\tau_1-1 \vee t_0} - Y^\rmlSPRT_{t_0} = 0 \leq 0 = \left(-\log \widetilde{\alpha}_1(\alpha)- Y^\rmlSPRT_{t_0} \right)_+.
\end{align}
Therefore, the first inequality holds in both cases.
\end{proof}

\subsubsection{Proof of  \cref{lemma:random_thresh_lemma1}}

\cref{lemma:random_thresh_lemma1} is the immediate corollary of the following two lemmas.

\begin{lemma}\label{lemma:random_thresh_lemma1_part1}
Suppose \cref{asm:L1_conv_psi}, \cref{asm:L1_conv_vs} and \cref{asm:cond_Lindeberg} hold. Then, as $\psi \sqrt{t_0} \to 0$,
\begin{align}
E\left[ \left(a - \frac{1}{2} \log t_0 - Y^\rmlSPRT_{t_0} \right)_+ \right] =&  E \left[ \left(a - \frac{1}{2} Z^2 \right)_+ \right] + o(1),\\
\text{ and } \qquad E \left[ \left(a  - Y^\nmSPRT_{\lambda, t_0} \right)_+ \right] =& E \left[ \left(a - \frac{1}{2} \log \frac{\eta}{1+\eta} - \frac{1}{2} \frac{Z^2}{1 + \eta} \right)_+ \right] + o(1),
\end{align}
where $Z \sim \mathcal{N}(0,1)$.
\end{lemma}

\begin{lemma}\label{lemma:random_thresh_lemma1_part2} Let $Z \sim \mathcal{N}(0,1)$. For any $a > 0$, $\eta > 0$ it holds that
\begin{align}
E\left[\left(a - \frac{1}{2} Z^2 \right)_+ \right] =& \sqrt{\frac{a}{\pi}} e^{-a} + \left(a - \frac{1}{2} \right) \left(2 \Phi(\sqrt{2 a}) - 1 \right),\\
\text{and} \qquad E\left[ \left(a - \frac{1}{2 ( 1 + \eta)} Z^2 \right)_+ \right] =& \sqrt{\frac{a}{\pi (1 + \eta)}} e^{-(1 + \eta) a} \\
&+ \left(a - \frac{1}{2 ( 1 + \eta)} \right) \left( 2 \Phi(\sqrt{2 (1+\eta) a}) - 1 \right).
\end{align}
\end{lemma}

\begin{proof}[Proof of \cref{lemma:random_thresh_lemma1_part1}] We start with the proof of the result for the running MLE test statistic. 
\paragraph{Running MLE threshold} We have that 
\begin{align}
E\left[ \left( a - \frac{1}{2} \log t_0 - Y^{\rmlSPRT}_{t_0} \right)_+ \right] = E \left[ \left(a - \frac{1}{2} Z^2 \right)_+ \right] + A + B,
\end{align}
where
\begin{align}
A =& E \left[\left( a - \frac{1}{2} \frac{S_{t_0}^2}{t_0}  \right)_+ \right] - E \left[ \left( a - \frac{1}{2} \frac{(S^0_{t_0})^2}{t_0} \right)_+\right], \\
\text{and} \qquad B=& E \left[ \left( a - \frac{1}{2} \frac{(S^0_{t_0})^2}{t_0} \right)_+ \right] - E\left[ \left( a - \frac{1}{2} Z^2 \right)_+ \right],
\end{align}
with $Z \sim \mathcal{N}(0,1)$.
We have that
\begin{align}
|A| \leq & E \left[ \frac{|S_{t_0}^2 - \breve{S}_{t_0}^2|}{2 t_0} \right] + E \left[ \frac{|(S^0_{t_0} + \psi t_0)^2 - (S^0_t)^2|}{2 t_0} \right] \\
\leq & E \left[|R^\bias_{0,t_0}| \right] + \frac{1}{2} \psi^2 t_0 + \psi \sqrt{t_0} E \left[ \frac{|S^0_{t_0}|}{\sqrt{t_0}} \right].
\end{align}
From lemma \ref{lemma:ERbias_t0}, $E \left[|R^\bias_{0,t_0}| \right] = o(1)$. From Cauchy-Schwarz, 
\begin{align}
E \left[ |S^0_{t_0}| / \sqrt{t_0} \right] \leq E\left[ (S^0_{t_0})^2 / t_0 \right]^{1/2} = \left( \frac{1}{t_0} \sum_{s=1}^{t_0} v_s \right)^{1/2} = O(1),
\end{align}
where the last equality follows from \cref{asm:L1_conv_vs} and Ces\`aro's lemma. Therefore, $A=o(1)$ as $\psi \sqrt{t_0} \to 0$. 

We now turn to $B$. Let $\xi_{s,t} = X^0_s / \sqrt{t}$ and $v_{s,t} = E[\xi_{s,t}^2 \mid \mathcal{F}_{s-1}] = v_s / t$. For any $\epsilon > 0$, from Markov's inequality and the triangle inequality,
\begin{align}
\Pr \left[ \left\lvert \frac{1}{t} \sum_{s=1}^t v_s - 1 \right\rvert \geq \epsilon \right] \leq \frac{1}{t} \sum_{s=1}^t E \left[|v_s - 1| \right] = o(1),
\end{align}
where the last equality follows from \cref{asm:L1_conv_vs} and Ces\`aro's lemma. Therefore, $\sum_{s=1} v_{s,t} \xrightarrow{P} 1$. This fact and \cref{asm:cond_Lindeberg} constitute the conditions of the martingale central limit theorem Corollary 3.1 in \cite{hall1980}. Therefore, $S^0_{t_0} / \sqrt{t_0} \rightsquigarrow \mathcal{N}(0,1)$. Since $x \mapsto (a - x^2/2)_+$ is a bounded continuous function, we have that $B=o(1)$.

\paragraph{Normal mixture SPRT threshold.} The proof is similar to that of the previous case. We have that
\begin{align}
E\left[ \left(a - Y^{\nmSPRT}_{\lambda, t_0} \right)_+ \right]  =& E \left[ \left(a- \frac{1}{2} \log \frac{\eta}{1+\eta} - \frac{1}{2} \frac{S_{t_0}^2}{(1+\eta) t_0} \right)_+ \right] \\
=& E \left[ \left(a - \frac{1}{2} \log \frac{\eta}{1+\eta} - \frac{1}{2}\ \frac{Z^2}{1 + \eta} \right)_+ \right] + A + B,
\end{align}
with $Z \sim \mathcal{N}(0,1)$,
where 
\begin{align}
A =& E \left[\left(a - Y^{\nmSPRT}_{\lambda, t_0} \right)+ \right] -E\left[ \left(a - \frac{1}{2} \log \frac{\eta}{1+\eta} - \frac{S^0_{t_0}}{t_0}\right)\right]\\
\qquad \text{and} \qquad B =& E\left[ \left( a - \frac{1}{2} \log \frac{\eta}{1+\eta} - \frac{(S^0_{t_0})^2}{2(1+\eta)t_0}\right)_+ \right] \\
&- E\left[ \left(a - \frac{1}{2} \log \frac{\eta}{1+\eta}  - \frac{Z^2}{2 ( 1 + \eta)} \right)_+\right].
\end{align}
We have that
\begin{align}
|A| \leq E \left[ |R^\bias_{\lambda, t_0} |\right] + E \left[ \frac{(S^0_{t_0} + \psi t_0)^2 - (S^0_{t_0})^2}{t_0} \right].
\end{align}
We have already shown in the case of the running MLE threshold that the right-hand side above is $o(1)$. Similarly, $B$ can be shown to be $o(1)$ as well using the same arguments as in the case of the running MLE threshold.
\end{proof}

\begin{proof}[Proof of \cref{lemma:random_thresh_lemma1_part2}]
We begin by proving an identity we will need later in the proof. For any $a' > 0$, by integration by parts, we have that
\begin{align}
\int_{-\sqrt{2a'}}^{\sqrt{2a'}} (-z^2) \frac{e^{-\frac{z^2}{2}}}{\sqrt{2 \pi}} dz =& 2 \left[z \frac{e^{-\frac{z^2}{2}}}{\sqrt{2 \pi}} \right]_{-\sqrt{2 a'}}^{\sqrt{2 a'}} - \int_{-\sqrt{2 a'}}^{\sqrt{2 a'}} \frac{e^{-\frac{z^2}{2}}}{\sqrt{2 \pi}} dz = 2 \sqrt{\frac{a'}{\pi}} - (2 \Phi(\sqrt{2 a'}) - 1).\label{eq:int_z2phi}
\end{align}
We start with the first identity in the statement of the lemma. Using \eqref{eq:int_z2phi}, we have that
\begin{align}
E \left[ \left(a - \frac{1}{2} Z^2 \right)_+ \right] =& \int_{-\sqrt{2a}}^{\sqrt{2a}} \left(a - \frac{1}{2} z^2 \right) \frac{e^{-\frac{z^2}{2}}}{\sqrt{2 \pi}} dz \\
=& a (2 \Phi(\sqrt{2 a} - 1) + \frac{1}{2} \int_{-\sqrt{2a}}^{\sqrt{2a}} (-z^2) \frac{e^{-\frac{z^2}{2}}}{\sqrt{2 \pi}} dz \\
=& \left(a  - \frac{1}{2} \right) (2 \Phi(\sqrt{2 a} - 1) + \sqrt{\frac{a}{\pi}} e^{-a}.
\end{align}
We now turn to the second identity in the statement of the lemma. Using \eqref{eq:int_z2phi}, we have that
\begin{align}
& E \left[ \left( a - \frac{1}{2(1+\eta)}  Z^2 \right)_+\right] \\
=& \int_{-\sqrt{2(1+\eta)a}}^{\sqrt{2(1+\eta)a}} \left( a - \frac{1}{2 ( 1 + \eta) } z^2 \right) \frac{e^{-\frac{z^2}{2}}}{\sqrt{2 \pi}} dz \\
=& a \left( 2 \Phi(\sqrt{2 (1+\eta) a}) - 1 \right) + \frac{1}{2 ( 1 + \eta)} \int_{-\sqrt{2(1+\eta)a}}^{\sqrt{2(1+\eta)a}} (-z^2) \frac{e^{-\frac{z^2}{2}}}{\sqrt{2 \pi}} dz \\
=& \left(a - \frac{1}{2 (1 + \eta)} \right) \left( 2 \Phi(\sqrt{2 (1+\eta) a}) - 1 \right) + \sqrt{\frac{a}{(1+\eta) \pi}} e^{-(1+\eta) a}.
\end{align}
\end{proof}

\subsubsection{Proof of \cref{lemma:random_thresh_lemma2}}

\begin{proof}[Proof of \cref{lemma:random_thresh_lemma2}]
We prove both claims separately.
\paragraph{Proof of the first claim} Recall that 
\begin{align}
h_1(a) = 2 \exp(-a) \sqrt{\frac{a}{\pi}} + 2 (1 - \Phi(\sqrt{2 a})).
\end{align}
Both terms in $h_1(a)$ are positive, so for $h_1(a)$ to solve $h_1(a) = \alpha$ as $\alpha \to 0$, both must converge to $0$. In particular, for the second term to converge to 0, $a$ must diverge to $\infty$.

Since $-\log h_1(a) = -\log \alpha$, if we can show that $-\log h_1(a) \sim a$, it will follow that $a \sim -\log \alpha$, which will readily implies the claim. We have that 
\begin{align}
-\log h_1(a) = a - \log \left(2 \sqrt{\frac{a}{\pi}} + \exp(a) (1 - \Phi(\sqrt{2 a}). \right)
\end{align}
From Mills's ratio \citep{Wainwright_2019_tail_bounds},
\begin{align}
0 \leq \exp(a) (1 - \Phi(\sqrt{2 a})) \lesssim \frac{1}{\sqrt{2 a}} - \frac{1}{(2 a)^{3/2}} + \frac{3}{(2a)^{5/2}} = o(1),
\end{align}
which readily implies that $-\log h_1(a) = a + o(a)$, which yields the claim.

\paragraph{Proof of the second claim.} Recall that 
\begin{align}
h_{2, \eta}(a) = \exp(-a) \left\lbrace 2 \Phi \left(\sqrt{2 \eta ( a + c_\eta)} \right) - 1  \right\rbrace + 2 \left\lbrace 1 - \Phi\left( \sqrt{\frac{1+\eta}{\eta} 2 \eta ( a + c_\eta)} \right) \right\rbrace,
\end{align}
with $c_\eta = 0.5 \log ((1+\eta)/\eta)$. Both terms are non-negative, therefore, for $a$ to solve $h_{2,\eta}(a) = \alpha$ as $\alpha \to 0$, both need to converge to 0. Since $(1+\eta) / \eta \to 1$, for the second term to converge to 0, we must have $2\eta(a + c_\eta) \to \infty$. This in turn implies that $2 \Phi ( \sqrt{2 \eta (a+ c_\eta)}) - 1 \to 1$, which implies that for the first term to converge to 0, we must have $\exp(-a) \to 0$, that is $a \to \infty$.

We now want to show that $a \sim -\log \alpha$, which will readily imply the claim. Since $ - \log h_{2, \eta}(a) = -\log \alpha$, the result will follow if we show that $-\log h_{2,\eta}(a) \sim a$. We have that 
\begin{align}
-\log h_{2,\eta}(a) = a - \log &\left( \left\lbrace 2 \Phi\left( \sqrt{2 \eta (a + c_\eta)} \right)- 1  \right\rbrace \right. \\
&+\left. \exp(a) \left\lbrace 1 - \Phi \left( \sqrt{\frac{1+\eta}{\eta} 2 \eta (a + c_\eta) }\right) \right\rbrace \right).
\end{align}
It remains to show that the second term is $o(a)$. Trivially, $2 \Phi(\sqrt{2 \eta (a + c_\eta)}) - 1 \to 1$. The second term in the $\log$ is more delicate. Invoking Mills's ratio,
\begin{align}
0 \leq & \exp(a) \left\lbrace 1 - \Phi\left(\sqrt{2(1+\eta)(a+c_\eta)}\right)\right\rbrace\\
 \lesssim & \exp(a - (1+\eta)(a+c_\eta))  \\
 &\times \left( \frac{1}{\sqrt{2(1+\eta)(a+c_\eta)}}  - \frac{1}{(2(1+\eta)(a+c_\eta))^{3/2}} + \frac{3}{(2(1+\eta)(a+c_\eta))^{5/2}}\right)\\
 =& o(1).
\end{align}
Therefore, the two terms inside the log are $1+o(1)$, which implies that $-\log h_{2,\eta}(a) = a + o(1)$, which in turn yields the claim.
\end{proof}

\subsection{Rejection times are almost surely finite}

\begin{proposition}[Almost sure finite stopping times]\label{prop:as_finite_tau}
Suppose that \cref{asm:fourth_moment}, \cref{asm:sup_L2_norm_vs_min1}, and \cref{asm:relative_bias_moments_summable} hold. Then, $S_t / t \to \psi$ almost surely and therefore, for $i=1,2$, $\tau_i < \infty$ almost surely.
\end{proposition}

\begin{proof}
We have that 
\begin{align}
\frac{S_t}{t} - \psi = \frac{1}{t} \sum_{s=1}^t X_s - \psi_s + \frac{1}{t} \sum_{s=1}^t \psi_s - \psi = \psi^0_{0,t} + \frac{1}{t} \sum_{s=1}^t \psi_s - \psi.
\end{align}
From \cref{lemma:fourth_moment_psi0s}, \cref{asm:fourth_moment} and \cref{asm:sup_L2_norm_vs_min1} imply that $E[(\psi^0_{0,t})^4] = O(t^{-2})$. Therefore, for any $\epsilon > 0$, $\sum_{t=1}^\infty P [|\psi^0_{0,t}| \geq \epsilon] \leq \epsilon^{-4} \sum_{t=1}^\infty E[(\psi^0_{0,t})^4] < \infty$. Thus, from the Borel-Cantelli lemma, $\psi^0_{0,t}  = o(1)$ almost surely. 

We now turn to $t^{-1} \sum_{s=1}^t \psi_s - \psi$. From \cref{asm:relative_bias_moments_summable}, $\sum_{t=1}^\infty P [|\psi_t - \psi| \leq  \epsilon] \leq \epsilon^{-(2+\delta)} \sum_{t=1}^\infty E[|\psi_t - \psi|^{2+\delta}] < \infty$, therefore, from the Borel-Cantelli lemma, $\psi_t - \psi = o(1)$ almost surely and therefore $t^{-1} \sum_{s=1}^t \psi_s - \psi = o(1)$ almost surely.
\end{proof}

\subsection{Proof of the asymptotic equivalent of the adaptivity term (\cref{lemma:adaptivity_term})}

The proof of \cref{lemma:adaptivity_term} relies on the following two intermediate results.

\begin{lemma}\label{lemma:expected_psi0s_squared}
Suppose \cref{asm:L1_conv_vs} holds. Then $E[(\psi^0_{0,s})^2] \sim s^{-1}$.
\end{lemma}

\begin{lemma}\label{lemma:fourth_moment_psi0s}
Suppose \cref{asm:fourth_moment} and \cref{asm:sup_L2_norm_vs_min1} hold. Then $E[(\psi^0_{0,s})^4] = O(s^{-2})$.
\end{lemma}

\begin{proof}[Proof of \cref{lemma:expected_psi0s_squared}]
Using the martingale property of $(X^0_s)_{s \geq 1}$ to eliminate cross-terms, and Ces\`aro's lemma, in that order, yields
\begin{align}
E[(\psi^0_{0,s})^2] = \frac{1}{s^2} \sum_{r=1}^s E[v_r] \sim \frac{1}{s}.
\end{align}
\end{proof}

\begin{proof}[Proof of \cref{lemma:fourth_moment_psi0s}]
From Rosenthal's inequality for martingales (see e.g. Theorem 2.12 in \cite{hall1980}), we have that
\begin{align}
E\left[ (\psi^0_{0,s})^4 \right] \lesssim & \frac{1}{s^4} \left\lbrace \sum_{r=1}^s E \left[ (X^0_r)^4 \right] + E\left[\left( \sum_{r=1}^s E\left[ (X^0_r)^2 \mid \mathcal{F}_{r-1} \right] \right)^2\right] \right\rbrace \\
\lesssim & \frac{1}{s^2} \left\lbrace \sup_{s \geq 1} s^{-2} \sum_{r=1}^s E[(X^0_r)^4] + E \left[ \left( 1 + \frac{1}{s} \sum_{r=1}^s v_r - 1 \right)^2 \right] \right\rbrace \\
\lesssim & O\left( \frac{1}{s^2} \right) + \frac{2}{s^2} \left\lbrace 1 + \left( E \left[ \left(\frac{1}{s} \sum_{r=1}^s v_r - 1  \right)^2 \right]^{1/2}  \right)^2 \right\rbrace \\
\lesssim  & O\left( \frac{1}{s^2} \right) + \frac{2}{s^2} \left( \frac{1}{s} \sum_{r=1}^2 E\left[(v_r -1)^2\right] \right)^2 \\
= & O\left( \frac{1}{s^2} \right),
\end{align}
where the second line follows from the definition of $v_r$, the third line follows from \cref{asm:fourth_moment} and the fact that $(a+b)^2 \leq 2 a^2 + 2 b^2$ for any two numbers $a$ and $b$, the fourth line follows from the triangle inequality, and the last line from \cref{asm:sup_L2_norm_vs_min1} and Cesa\`aro's lemma.
\end{proof}

We can now prove \cref{lemma:adaptivity_term}.
\begin{proof}
We present the proof in the RMLE case only. (The proof in the nmSPRT case will follow directly from the fact that that for any $\lambda > 0$, $ \rho_{\lambda, s} (\psi^0_{\lambda, s})^2 \leq (\psi_{0,s}^0)^2$.) We have that
\begin{align}
E\left[R^\adpt_{0, \tau} - R^\adpt_{0,t_0} \right] \leq & \frac{1}{2}(A_\psi + B_\psi + C_\psi),
\end{align}
where
\begin{align}
A_\psi =& \sum_{s=1}^{\left\lfloor \psi^{-2} \log \psi^{-1} \right\rfloor} E \left[ (\psi^0_{0,s})^2 \right], \\
B_\psi =& \sum_{s=\left\lfloor \psi^{-2} \log \psi^{-1} \right\rfloor + 1}^\infty E\left[ \bm{1}\left\lbrace |\psi^0_{0,s}| \leq \psi (\log \psi^{-1})^{-1/4}, \tau \geq s \right\rbrace (\psi^0_{0,s})^2   \right],\\
C_\psi =& \sum_{s=\left\lfloor \psi^{-2} \log \psi^{-1} \right\rfloor + 1}^\infty E \left[ \bm{1} \left\lbrace |\psi^0_{0,s}| > \psi (\log \psi^{-1})^{-1/4} \right\rbrace  (\psi^0_{0,s})^2 \right].
\end{align}

From \ref{asm:L1_conv_vs} and \cref{lemma:expected_psi0s_squared}, we have that
$A_\psi \sim 2 \log \psi^{-1}.$ By definition of the terms in $B_\psi$, we immediately have that $B_\psi \leq \psi^2 (\log \psi^{-1})^{-1/2} E[\tau]$. It remains to bound $C_\psi$. 

By integration by parts followed by application of Markov's inequality, 
\begin{align}
&E\left[\bm{1} \left\lbrace |\psi^0_{0,s}| > \psi (\log \psi^{-1})^{-1/4} \right\rbrace (\psi^0_{0,s})^2 \right] \\
= & \int_{\psi (\log \psi^{-1})^{-1/4}}^\infty x^2 P \left[ |\psi^0_{0,s}| \in dx \right] \\
=& \psi^2 (\log \psi^{-1})^{-1/2} P \left[ |\psi^0_{0,s}| \geq \psi (\log \psi^{-1})^{-1/4}\right] + 2 \int_{\psi (\log \psi^{-1})^{-1/4}}^\infty x P \left[ |\psi^0_{0,s}| \geq x \right] dx \\ 
\leq & \left( \psi (\log \psi^{-1})^{-1/4} \right)^{-2} E \left[ (\psi^0_{0,s})^4 \right] + 2 \int_{\psi (\log \psi^{-1})^{-1/4}}^\infty x^{-3} E \left[ (\psi^0_{0,s})^4 \right] dx \\
=& 2 \psi^{-2} (\log \psi^{-1})^{1/2}  E \left[ (\psi^0_{0,s})^4 \right] \\
\lesssim & \psi^{-2} (\log \psi^{-1})^{1/2} s^{-2},
\end{align}
where the last line follows from \cref{asm:fourth_moment} and \cref{lemma:fourth_moment_psi0s}. Therefore, by summation, 
\begin{align}
C_\psi \lesssim \psi^{-2} (\log \psi^{-1})^{1/2} \times \psi^2 (\log \psi^{-1})^{-1} = (\log \psi^{-1})^{-1/2}.
\end{align}
\end{proof}

\subsection{Proof of the bound on the bias remainder term (\cref{lemma:bias_term})}

We first state and prove intermediate lemmas.

\begin{lemma}\label{lemma:ERbias_t0}
Suppose \cref{asm:L1_conv_psi} and \cref{asm:L1_conv_vs} hold. Then, for any $\eta \geq 0$, $E[|R^\bias_{\lambda, t_0}|] = o(1)$ as $t_0 \to \infty$, $\psi \sqrt{t_0} \to 0$.
\end{lemma}

\begin{proof}[Proof of \cref{lemma:ERbias_t0}]
We have that
\begin{align}
E[|R^\bias_{\lambda, t_0}|] \leq \frac{\psi}{2 t_0} E \left[ \sum_{s=1}^{t_0} \frac{|\psi_s - \psi|}{\psi} \left( 2 \psi t_0 + \psi \sum_{s=1}^{t_0} \frac{|\psi_s - \psi|}{\psi} + | \sum_{s=1}^{t_0} X^0_s | \right) \right] \leq
A_{\psi, t_0} + B_{\psi, t_0} + C_{\psi, t_0},
\end{align}
where 
\begin{align}
A_{\psi, t_0} =& \psi^2 \sum_{s=1}^{t_0} E \left[ \frac{|\psi_s - \psi|}{\psi} \right],\\
B_{\psi, t_0} =& \frac{1}{2} \frac{\psi^2}{t_0} E \left[ \left( \sum_{s=1}^{t_0} \frac{|\psi_s - \psi|}{\psi} \right)^2 \right], \\
\text{and} \qquad  C_{\psi, t_0} =& \frac{\psi}{2 t_0} E \left[ \sum_{s=1}^{t_0} \frac{|\psi_s - \psi|}{\psi} \times | \sum_{s=1}^{t_0} X^0_s | \right].
\end{align}
Using Jensen's inequality, \cref{asm:L1_conv_psi} and applying Ces\`aro's lemma yields that
\begin{align}
A_{\psi,t_0} \leq \psi^2 \sum_{s=1}^{t_0} E \left[ \left\lvert \frac{\psi_s - \psi}{\psi} \right\rvert^2 \right]^{1/2} = o(\psi^2 t_0) = o(1).
\end{align}
Using the Cauchy-Schwarz inequality, \cref{asm:L1_conv_psi} and applying Ces\`aro's lemma yields
\begin{align}
B_{\psi,t_0} \leq \frac{1}{2} \psi^2 \sum_{s=1}^{t_0} E \left[  \left\lvert \frac{\psi_s - \psi}{\psi} \right\vert^2 \right] = o\left( \psi^2 t_0 \right) = o(1).
\end{align}
Two consecutive applications of the Cauchy-Schwarz inequality give that
\begin{align}
C_{\psi, t_0} \leq & \frac{\psi}{2 t_0} E\left[ \left( \sum_{s=1}^{t_0} \frac{|\psi_s - \psi|}{\psi} \right)^2 \right]^{1/2} \times \sqrt{t_0} E \left[ \left( \frac{\sum_{s=1}^{t_0} X^0_s}{\sqrt{t_0}} \right)^2 \right]^{1/2} \\
\leq & \frac{\psi}{2}  \left(\sum_{s=1}^{t_0} E \left[\left\lvert \frac{\psi_s - \psi}{\psi} \right\rvert^2 \right]\right)^{1/2}  E \left[ \left( \frac{\sum_{s=1}^{t_0} X^0_s}{\sqrt{t_0}} \right)^2 \right]^{1/2}
\end{align}
Using the martingale difference property of $(X^0_s)_{s\geq 1}$ to eliminate cross-terms, we have that $E[ (\sum_{s=1}^{t_0} X^0_s / \sqrt{t_0} )^2 ] = E [t_0^{-1} \sum_{s=1}^{t_0} v_s ] = O(1)$, where the last equality follows from \cref{asm:L1_conv_vs} and Ces\`aro's lemma.
 Using \cref{asm:L1_conv_psi} and Ces\`aro's lemma as in the two previous cases yields that $ \sum_{s=1}^{t_0} E[|\psi_s - \psi|^2 \psi^{-2} ] = o(t_0)$. Therefore, 
\begin{align}
C_{\psi, t_0} = o(\psi \sqrt{t_0}) = o(1).
\end{align}
\end{proof}

\begin{lemma}\label{lemma:expectation_sum_to_tau_under_moment_summability}
Suppose that $(Z_t)_{t \geq 1}$ is an $\mathcal{F}$-predictable sequence such that there exists $\delta >0 $ such that $(E[|Z_t|^{1+\delta})_{t \geq 1})$ is summable, and let $(\tau(t_0))_{t_0 \geq 1}$ be an $\mathfrak{F}$-stopping time sequence such that $E[\tau(t_0)] \to \infty$ as $t_0 \to \infty$. Then
\begin{align}
E \left[ \sum_{s=1}^{\tau(t_0)} |Z_s| \right] = o(E[\tau(t_0)])
\end{align}
as $t_0 \to \infty$.
\end{lemma}

\begin{proof}[Proof of \cref{lemma:expectation_sum_to_tau_under_moment_summability}]
Let $\eta > 0$. We will specify the value of $\eta$ later in the proof. We have that
\begin{align}
E \left[ \sum_{s=1}^{\tau} |Z_s| \right]  \leq A + B + C,
\end{align}
where 
\begin{align}
A =& \sum_{s=1}^{\left\lfloor E[\tau] \right\rfloor} E[|Z_s|],\\
B =& \sum_{s=\left\lfloor E[\tau] \right\rfloor + 1}^\infty E \left[|Z_s| \bm{1}\left\lbrace |Z_s| \leq \eta, \tau \geq s \right\rbrace \right],\\
\text{and} \qquad C =& \sum_{s=\left\lfloor E[\tau] \right\rfloor + 1}^\infty E \left[|Z_s| \bm{1}\left\lbrace |Z_s| > \eta \right\rbrace \right].
\end{align}
That $(E[|Z_s|^{1+\delta}])_{s \geq 1}$ is summable implies that $E[|Z_s|^{1+\delta}] \to 0$ and thus, via Jensen's inequality that $E[|Z_s|] \to 0$. Therefore, from Ces\`aro's lemma, $A = o(E[\tau])$. We immediately have that $B \leq \eta E[\tau]$. We now turn to $C$. We have that
\begin{align}
E\left[|Z_s| \bm{1} \left\lbrace |Z_s| > \eta \right\rbrace \right] =& \eta P\left[|Z_s| > \eta \right] + \int_\eta^\infty P \left[ |Z_s| \geq z\right] dz \\
\leq & \eta^{-\delta} E \left[ |Z_s|^{1+\delta} \right] + \int_\eta^\infty z^{-1-\delta} E\left[ |Z_s|^{1+\delta} \right] dz \\
=& \frac{\delta+1}{\delta} \eta^{-\delta} E\left[ |Z_s|^{1+\delta} \right].
\end{align}
Therefore, $C \lesssim \eta^{-\delta} \sum_{s=\left\lfloor E[\tau] \right\rfloor + 1}^\infty E[|Z_s|^{1+\delta}]$. We set $\eta = (\sum_{s=\left\lfloor E[\tau] \right\rfloor + 1}^\infty E[|Z_s|^{1+\delta}])^{1/(2 \delta)}$. As $E[|Z_s|^{1+\delta}]$ is summable, we have that $\eta = o(1)$ as $t_0 \to \infty$. We thus have that $B = o(E[\tau])$ as $t_0 \to \infty$ and that $C=O(\eta^\delta) = o(1)$ as $t_0 \to \infty$.
\end{proof}

\begin{lemma}\label{lemma:expectation_S0_tau_squared_over_tau}
Suppose \cref{asm:Evsmin1_over_s_summable} holds. Then, for any almost surely finite $\mathcal{F}$-stopping time $\tau$, we have that
\begin{align}
E \left[  \frac{(S^0_{\tau})^2}{\tau} \right] \leq \log E[\tau] + O(1).
\end{align}
\end{lemma}

\begin{proof}[Proof of \cref{lemma:expectation_S0_tau_squared_over_tau}]
Observe that 
\begin{align}
\left( \frac{(S^0_t)^2}{t} - \sum_{s=1}^t \frac{v_s}{s} \right)_{t \geq 1} 
\end{align}
is a submartingale with initial value 0. Then, from the optional stopping theorem,
\begin{align}
E \left[ \frac{(S^0_{\tau})^2}{\tau} \right] \leq & E \left[ \sum_{s=1}^{\tau} \frac{v_s}{s} \right] \\
\leq & E \left[ \log \tau \right] + \sum_{s=1}^\infty \frac{E \left[ |v_s - 1| \right]}{s} \\
\leq & \log E[\tau] + O(1),
\end{align}
where the last line follows Jensen's inequality and \cref{asm:Evsmin1_over_s_summable}.
\end{proof}

\begin{proof}[Proof of \cref{lemma:bias_term}]
Since $|E[R^\bias_{\lambda, \tau}] - E[R^\bias_{\lambda, t_0}]| \leq |E[R^\bias_{\lambda, \tau}]| + |E[R^\bias_{\lambda, t_0}]|$ and from  \cref{lemma:ERbias_t0}, $|E[R^\bias_{\lambda, t_0}]| = o(1)$, it remains to bound the first term. We have that
\begin{align}
|E[R^\bias_{\lambda, \tau}]| \leq A + B + C,
\end{align}
where
\begin{align}
A =& E\left[\psi^2 \sum_{s=1}^{\tau} \frac{|\psi_s - \psi|}{\psi}\right]\\
B =& E\left[\psi^2 \frac{1}{\tau} \left( \sum_{s=1}^{\tau}  \frac{|\psi_s - \psi|}{\psi} \right)^2 \right]\\
\text{and} \qquad C=& E\left[ \psi \frac{1}{\sqrt{\tau}} \sum_{s=1}^{\tau} \frac{|\psi_s - \psi|}{\psi} \times \frac{S^0_{\tau}}{\sqrt{\tau}}\right].
\end{align}
From \cref{asm:relative_bias_moments_summable} and \cref{lemma:expectation_sum_to_tau_under_moment_summability}, $A =o(\psi^2 E[\tau])$. From Cauchy-Schwarz, $B \leq \psi^2 E[ \sum_{s=1}^{\tau} (|\psi_s - \psi|/\psi)^2]$, and therefore, from \cref{asm:relative_bias_moments_summable} and \cref{lemma:expectation_sum_to_tau_under_moment_summability}, $B = o(\psi^2 E[\tau])$. We now turn to $C$. Applying Cauchy-Schwarz twice yields 
\begin{align}
C \leq & \psi E \left[ \frac{1}{\tau} \left( \sum_{s=1}^{\tau} \frac{|\psi_s - \psi|}{\psi} \right)^2 \right]^{1/2} \times E \left[ \frac{(S^0_{\tau})^2}{\tau} \right]^{1/2} \\
\leq & \psi E \left[ \sum_{s=1}^{\tau} \left\lvert \frac{\psi_s - \psi}{\psi} \right\rvert^2 \right]^{1/2} \times E \left[ \frac{(S^0_{\tau})^2}{\tau} \right]^{1/2}.
\end{align}
From \cref{asm:relative_bias_moments_summable} and \cref{lemma:expectation_sum_to_tau_under_moment_summability}, the first expectation above is $o(E[\tau])$. From \cref{asm:Evsmin1_over_s_summable} and \cref{lemma:expectation_S0_tau_squared_over_tau}, the second expectation above is $o(\sqrt{\log E[\tau]})$. Therefore 
\begin{align}
C = o\left(\psi \sqrt{E[\tau] \log E[\tau]}\right).
\end{align}
\end{proof}

\subsection{Proofs of the bounds on shrinkage terms (\cref{lemma:ERskg1} and \cref{lemma:ERskg2})}

\begin{proof}[Proof of \cref{lemma:ERskg1}] We have that
\begin{align}
\left\lvert R^{\skg,1}_{\lambda, \tau(t_0)} - R^{\skg,1}_{\lambda, t_0} \right\rvert =& \frac{1}{2} \psi^2 \lambda \left\lvert \frac{ \tau(t_0)}{ \tau(t_0) + \lambda} - \frac{t_0}{t_0 + \lambda} \right\rvert  \\
=& \frac{1}{2} \psi^2 \lambda \frac{ \tau(t_0) - t_0}{ \tau(t_0) + \lambda} \times \frac{\lambda}{t_0 + \lambda} \\
\leq & \frac{1}{2} \psi^2 \lambda \\
=& o(\log \psi^{-1}) 
\end{align}
as $\lambda = o(\psi^{-2} \log \psi^{-1})$.
\end{proof}

\begin{proof}[Proof of \cref{lemma:ERskg2}]
We have that 
\begin{align}
    R^{\skg,2}_{\lambda, \tau(t_0)} - R^{\skg,2}_{\lambda, t_0} =& \psi \lambda \sum_{s=t_0}^{\tau(t_0) - 1} S^0_s \left( \frac{1}{s+\lambda} - \frac{1}{s+1+\lambda} \right) \\ 
    = & \psi \lambda \left(\sum_{s=t_0}^{\tau(t_0) -1} \frac{S^0_s}{s+\lambda} - \sum_{s=t_0 + 1}^{\tau(t_0)} \frac{S^0_{s-1}}{s+\lambda} \right) \\
    =& \psi \lambda \left( - \frac{S^0_{\tau(t_0) - 1}}{\tau + \lambda} + \sum_{s = t_0 + 1}^{\tau(t_0) - 1} \frac{X^0_s}{s + \lambda} + \frac{S^0_{t_0 }}{t_0  + \lambda}\right).
\end{align}
The third term trivially has expectation 0. The terms of the sum in the second term form an MDS, and therefore this term has expectation 0 from the optional stopping theorem. We now turn to the first term. From Jensen's inequality,
\begin{align}
    \left\lvert E \left[ \frac{S^0_{\tau(t_0)-1}}{\tau(t_0) + \lambda} \right] \right\rvert \leq E \left[ \left( \frac{S^0_{\tau(t_0)-1}}{\tau(t_0) + \lambda} \right)^2  \right]^{1/2}.
\end{align}
It is straightforward to observe that 
\begin{align}
    \left( \frac{(S^0_t)^2}{(t + 1 + \lambda)^2} - \sum_{s=1}^t \frac{v_s}{(s + 1 + \lambda)^2 }\right)_{t \geq 1}
\end{align}
is a supermartingale of which the initial term has expectation $0$.
The optional stopping theorem then implies that 
\begin{align}
    E \left[ \frac{(S_{\tau(t_0)}^0)^2}{(\tau + 1 + \lambda)^2} \right] \leq &  E \left[ \sum_{t=1}^{\tau(t_0)} \frac{v_t}{(t+1+\lambda)^2}\right]  \leq  \sup_{t \geq 1} E[v_t] \sum_{t = \left\lceil \lambda \right\rceil + 1}^\infty \frac{1}{t} =O(\lambda)
\end{align}
where the last line follows from \cref{asm:L1_conv_vs}. Therefore,
\begin{align}
E\left[ R^{\skg,2}_{\lambda, \tau(t_0)} - R^{\skg,2}_{\lambda, t_0} \right] = O (\psi \sqrt{\lambda}) = o(\log \psi^{-1})
\end{align}
since $\lambda = o(\psi^{-2} (\log \psi^{-1})^2)$.
\end{proof} 

\subsection{Proofs of the bound on the expected difference in quadratic variation (\cref{lemma:Deltaqvar})}

\begin{proof}[Proof of \cref{lemma:Deltaqvar}]
We have that
\begin{align}
\left\lvert E\left[\Deltaqvar_{\lambda, \tau(t_0)} - \Deltaqvar_{\lambda, t_0} \right] \right\rvert =& \left\lvert E\left[  \sum_{s=t_0 + 1}^{\tau(t_0)} \frac{(X^0_s)^2}{s + \lambda} - \log \frac{\tau(t_0) + \lambda}{t_0 + \lambda}\right] \right\rvert \\
=& \left\lvert E \left[ \sum_{s=t_0 + 1}^{\tau(t_0)} \frac{v_s}{s+\lambda} - \log \frac{\tau + \lambda}{t_0 + \lambda}  \right] \right\rvert \\
\leq & o(1) + \sum_{t_0+1}^\infty \frac{E[|v_s - 1|]}{s} \\
=& o(1)
\end{align}
as $t_0 \to \infty$, where the last inequality follows from \cref{asm:Evsmin1_over_s_summable}.
\end{proof}

\section{Proof of the estimating equations results (section \ref{sec:application}) }

\subsection{Proof of the type-I error result (\cref{prop:typeI_err_esteq})}
\begin{lemma}\label{lemma:vs_almost_sure_conv}
Suppose that \cref{asm:moment_estimating_function} and \cref{asm:variance_lower_bound} hold. Then, it holds that
$v_t - 1 = o(t^{-((1-\delta) / (3-\delta)- 2 \iota)})$ almost surely.
\end{lemma}

\begin{proof}[Proof of \cref{lemma:vs_almost_sure_conv}]
Let $n(t) = |\mathcal{I}_{0,t}|$ and let $\bar{D}_t(\eta',\theta') = n(t)^{-1} \sum_{s\in\mathcal{I}_{0,t}} D(O_s; \eta', \theta')$. We have that
\begin{align}
\widehat{\sigma}_t - \sigma^2(\widehat{\eta}_t,\theta_0) =& \frac{1}{n(t)-1} \sum_{s \in \mathcal{I}_{0,t}} \left( D(O_s;\widehat{\eta}_t, \theta_0) - \bar{D}_t(\widehat{\eta}_t, \theta_0) \right)^2 \\
=& \frac{n(t)}{n(t)-1} \left( m_{2,t} - m_{1,t}^2 \right)
\end{align}
where
\begin{align}
m_{2,t} =& \frac{1}{n(t)} \sum_{s \in \mathcal{I}_{0,t}} Z^{(2)}_{t,s}   \\ 
\text{with} \qquad  Z_{t,s}^{(2)} =& \left( D(O_s ; \widehat{\eta}_t, \theta_0) - E[ D(O_s ; \widehat{\eta}_t, \theta_0) \mid \mathcal{D}_{1,t}] \right)^2 - \sigma^2(\widehat{\eta}_t, \theta_0), \\
\text{and} \qquad m_{1,t} =& \frac{1}{n(t)} \sum_{s \in \mathcal{I}_{0,t}} Z^{(1)}_{t,s} \\
\text{with} \qquad Z^{(1)}_{t,s} =& D(O_s; \widehat{\eta}_t, \theta_0) - E[D(O_s; \widehat{\eta}_t, \theta_0) \mid \mathcal{D}_{1,t}],
\end{align}
Observe that $\{ Z^{(2)}_{t,s} : t \geq 4; s \in \mathcal{I}_{0,t} \}$ forms an array of random variables that are row-wise independent conditional on $\mathcal{D}_{1,t}$. From \cref{asm:moment_estimating_function}, we have for any $\delta \in (0,1)$ that
$\sup E[(Z^{(2)}_{t,s})^{2 p + \delta } \mid \mathcal{D}_{1,t} ] < \infty$ with $p = (3-\delta) / 2$. As $1 / p = 1 - (1-\delta)/(3 - \delta)$. Corollary 1 in \cite{hu1989strong} then gives that $m_{2,t} = o(t^{-(1-\delta)/(3-\delta)})$ almost surely. The same arguments show that $m_{1,t} = o(t^{-(1-\delta)/(2-\delta)})$ almost surely, for any $\delta \in (0,1)$. Therefore, for any $\delta \in (0,1)$,
\begin{align}
\widehat{\sigma}_t^2 - \sigma^2(\widehat{\eta}_t, \theta_0) = o(t^{-(1-\delta)/(3-\delta)}) \qquad \text{almost surely}.
\end{align}
From \cref{asm:variance_lower_bound}, there exists $t_0$ such that for every $t \geq t_0$, $\chi t^{-\iota} < \sigma(\widehat{\eta}_t, \theta_0)$. We then have that for every $t \geq t_0$, 
\begin{align}
|v_{t+1} - 1| =& \frac{|\sigma^2(\widehat{\eta}_t, \theta_0) - \left( \widehat{\sigma}_t \vee \chi t^{-\iota} \right)^2|}{\left( \widehat{\sigma}_t \vee \chi t^{-\iota} \right)^2} \\
\leq & \chi^2 t^{2 \iota} |\widehat{\sigma}_t^2  -  \sigma^2(\widehat{\eta}_t,  \theta_0)| \\
= & o\left( t^{2 \iota - \frac{1-\delta}{3- \delta}} \right).
\end{align}
\end{proof}

\begin{proof}[Proof of \cref{prop:typeI_err_esteq}]
From the assumption on $\iota$, there exists $\kappa \in (0,1)$ small enough that $\kappa / 2 < (1-\delta)/(3-\delta) - 2 \iota$ and $(1 + \delta /2)(1 - \kappa - 2 \iota) > 1$. Let $f(t) = t^{1-\kappa}$ for such a value of $\kappa$.

From \cref{lemma:vs_almost_sure_conv}, $V_t - t = o(t^{1 - \{(1-\delta)/(3-\delta) - 2\iota\}}) $ and $\sqrt{t f(t)} \log t = t^{1 - \kappa /2} \log t$. Therefore, since $\kappa / 2 < (1-\delta)/(3-\delta) - 2\iota$, it holds that $V_t - t = o(\sqrt{t f(t)} \log t)$. Therefore, the second part of \cref{asm:sip} holds.

We now turn to the first part of \cref{asm:sip}. We have that
\begin{align}
& \frac{1}{f(V_t)} E \left[(X_t - \psi_t) \bm{1}\{(X_t-\psi_t)^2 > f(V_t) \} \mid \mathcal{F}_{t-1} \right] \\
\leq & \frac{1}{f(V_t)^{1+\nu  / 2}} E[(X_t - \psi_t)^{2+\nu}] \\
\lesssim & t^{- \{ (1+\nu/2)(1-\kappa) - (2+\nu)\iota \} } \sup_{\eta',\theta'} E[D^{2+\nu}(O_1; \eta', \theta')],
\end{align}
which is summable as $(1+\nu/2)(1-\kappa) - (2+\nu)\iota > 1$, and as $\sup_{\eta',\theta'} E[D^{2+\nu}(O_1; \eta', \theta')] < \infty$ from \cref{asm:moment_estimating_function}. 

We have thus shown that \cref{asm:sip} holds. The claims then follow from \cref{thm:typeI_error}.
\end{proof}

\subsection{Proof of the expected rejection time result (\cref{prop:expected_rejection_time})}

The results of this section rely on the following result on the moments of the sample variance.

\begin{lemma}\label{lemma:moments_sample_variance}
Suppose that \cref{asm:moment_estimating_function} holds. Then $E[|\widehat{\sigma}_t - \sigma(\widehat{\eta}_t,\theta_0)|^2] = O(t^{-1})$ and $E[|\widehat{\sigma}_t - \sigma(\widehat{\eta}_t,\theta_0)|^3] = O(t^{-2})$.
\end{lemma}

The bound on the moment of order 2 follows from the U-statistic representation of the sample variance and Hoeffding's theorem on the variance of U-statistics \citep{hoeffding1948}. A proof of the bound on the moment of order 3 can be found, e.g., in \cite{angelova2012moments}.

\begin{lemma}\label{lemma:checking_vt_asms}
Suppose that \cref{asm:variance_lower_bound} and \cref{asm:moment_estimating_function} hold, and that $4 \iota < 1$. Then \cref{asm:Evsmin1_over_s_summable}, \cref{asm:sup_L2_norm_vs_min1} and \cref{asm:L1_conv_vs} hold.
\end{lemma}

\begin{proof}[Proof of lemma \ref{lemma:checking_vt_asms}]
From \cref{asm:variance_lower_bound}, there exists $t_1$ such that for any $t \geq t_1$, $\chi t^{-\iota} \leq \inf_{\eta' \in \mathcal{T}, \theta' \in \Theta} \sigma(\eta', \theta')$. For any $t \geq t_1$, we have that
\begin{align}
E\left[ |v_{t+1} - 1|^2 \right] =& E\left[ \frac{\left\lvert (\widehat{\sigma}_t \vee \chi t^{-\iota})^2 - \sigma^2(\widehat{\eta}_t, \theta_0)\right\rvert^2}{(\widehat{\sigma}_t \vee \chi t^{-\iota})^4} \right] \\
\leq & \chi^{-4} t^{4 \iota} E\left[ \left\lvert \widehat{\sigma}_t^2 - \sigma(\widehat{\eta}_t, \theta_0)^2\right\rvert \right] \\
=& O \left( t^{-(1 - 4 \iota} \right),
\end{align}
where the last line follows from \cref{lemma:moments_sample_variance}. This directly implies \cref{asm:sup_L2_norm_vs_min1}. From Jensen's inequality, $E[|v_{t+1} - 1|] = O(t^{-(1 - 4\iota)/2})$, which implies \cref{asm:Evsmin1_over_s_summable} and \cref{asm:L1_conv_vs}.
\end{proof}

\begin{lemma}\label{lemma:checking_asm_sum_fourth_moment}
Suppose that \cref{asm:moment_estimating_function} holds and that $4 \iota < 1$. Then \cref{asm:fourth_moment} holds.
\end{lemma}

\begin{proof}[Proof of \cref{lemma:checking_asm_sum_fourth_moment}]
We have that
\begin{align}
t^{-2} \sum_{s=1}^t E [(X^0_s)^4] \lesssim t^{-2} \sum_{s=1}^t s^{4 \iota} \sup_{\eta' \in \mathcal{T}, \theta' \in \Theta} E\left[ D(O_s;\eta',\theta_0)^4 \right] = O(t^{-(1 - 4 \iota)}) = o(1).
\end{align}
\end{proof}

\begin{lemma}\label{lemma:checking_conditional_Lindeberg}
Suppose that \cref{asm:moment_estimating_function} holds and that $\iota$ is small enough that $3-6 \iota > 1$. Then, \cref{asm:cond_Lindeberg} holds.
\end{lemma}

\begin{proof}[Proof of \cref{lemma:checking_conditional_Lindeberg}]
Let $\epsilon >0$. We have 
\begin{align}
E\left[(t^{-1/2} X^0_s)^2 \bm{1}\left\lbrace |t^{-1/2} X^0_s| > \epsilon \right\rbrace \right] \leq & \frac{1}{\epsilon^{4}} E \left[ |t^{-1/2} X^0_s|^{6} \mid \mathcal{F}_{s-1} \right] \\
\lesssim & t^{-(3(1 - 2\iota))} \sup_{\eta' \in \mathcal{T}, \theta' \in \Theta} E \left[D(O_1; \eta', \theta')^6 \right],
\end{align}
which is  summable, as from \cref{asm:moment_estimating_function}, $E[\sup_{\eta' \in \mathcal{T}, \theta' \in \Theta} D(O_1; \eta', \theta')^6 ] < \infty$, and as $3-6 \iota > 1$.
\end{proof}

\begin{lemma}\label{lemma:checking_relative_bias_condition}
Suppose that \cref{asm:moment_estimating_function}, \cref{asm:robustness}, \cref{asm:variance_lower_bound} and \cref{asm:nuisance_convergence} hold, and that $\iota$ is small enough that $(2+\delta)\iota - 2 > 1$ and that $\nu > \iota ( 2 + \delta)$, with $\delta$ from \cref{asm:nuisance_convergence}. Then, \cref{asm:relative_bias_moments_summable} and \cref{asm:L1_conv_psi} hold.
\end{lemma}

\begin{proof}[Proof of \cref{lemma:checking_relative_bias_condition}]
From \cref{asm:variance_lower_bound}, there exists $t_1$ such that for every $t \geq t_1$, $\chi t^{-\iota} \leq \inf_{\eta' \in \mathcal{T}, \theta' \in \Theta} \sigma(\eta', \theta')$. Let $\psi = \sigma(\eta_1, \theta_0)^{-1} \mu(\theta_0)$ where $\eta_1$ is defined in \cref{asm:nuisance_convergence} and $\mu$ is defined in \cref{asm:robustness}.  From \cref{asm:robustness}, we then have that, for every $t \geq t_1$, 
\begin{align}
\left\lvert \frac{\psi_{t+1} - \psi}{\psi} \right\rvert = \frac{|(\widehat{\sigma}_t \vee \chi t^{-\iota})^{-1} - \sigma(\eta_1, \theta_0)^{-1}|}{\sigma(\eta_1, \theta_0)^{-1}} \leq \chi^{-1} t^\iota |\widehat{\sigma}_t - \sigma(\widehat{\eta}_t, \theta_0)|.
\end{align}
Let $\delta$ be as in \cref{asm:nuisance_convergence}. From Jensen's inequality, we have that
\begin{align}
E \left[\left\lvert \frac{\psi_{t+1} - \psi}{\psi} \right\rvert^{2+\delta}\right] \lesssim t^{\iota ( 2 + \delta)} E \left[ | \widehat{\sigma}_t - \sigma(\widehat{\eta}_t, \theta_0)|^{2 + \delta} \right] + t^{\iota(2+\delta)} E\left[ |\sigma(\widehat{\eta}_t,\theta_0)-\sigma(\eta_1, \theta_0)|^{2+\delta} \right].
\end{align}
The second term is summable from \cref{asm:nuisance_convergence} as $\iota(2 + \delta) < \nu$. We now turn to the first term. We have that
\begin{align}
t^{(2+\delta) \iota} E\left[\left\lvert \widehat{\sigma}_t - \sigma(\widehat{\eta}_t,\theta_0) \right\rvert^{2+\delta} \right] \leq & \frac{t^{\iota (2+\delta)}}{\inf_{\eta' \in \mathcal{T}, \theta' \in \Theta} \sigma(\eta', \theta')^{3-\delta'}} E\left[\left\lvert \widehat{\sigma}^2_t - \sigma(\widehat{\eta}_t,\theta_0)^2 \right\rvert^{3} \right]\\
 =& O (t^{(2+\delta) \iota - 2}),
\end{align}
where the equality follows from the fact that $\inf_{\eta' \in \mathcal{T}, \theta' \in \Theta} \sigma(\eta', \theta') > 0$ from \cref{asm:variance_lower_bound} and the fact that $E\left[\left\lvert \widehat{\sigma}^2_t - \sigma(\widehat{\eta}_t,\theta_0)^2 \right\rvert^{3} \right] = O(t^{-2})$ from \cref{asm:moment_estimating_function} via \cref{lemma:moments_sample_variance}.
\end{proof}

\end{document}